\documentclass[11pt, oneside]{article}   	

\usepackage{jheppub}

\usepackage{amsmath}
\usepackage{amssymb}
\usepackage{amsfonts}
\usepackage{amsthm}
\usepackage{amsbsy}
\usepackage{array}
\usepackage{multicol}

\usepackage{tikzit}

\tikzstyle{new style 0}=[fill=white, draw=black, shape=circle]
\tikzstyle{new style 1}=[fill=red, draw=black, shape=circle]
\tikzstyle{new style 2}=[fill=green, draw=black, shape=circle]
\tikzstyle{new style 3}=[fill=magenta, draw=black, shape=circle, {scale=0.3}]

\tikzstyle{new edge style 0}=[-, fill=none, draw=red]
\tikzstyle{new edge style 1}=[-, draw=blue]
\tikzstyle{new edge style 2}=[{|-|}]
\tikzstyle{new edge style 3}=[-, draw={rgb,255: red,100; green,0; blue,100}]
\tikzstyle{new edge style 4}=[<->, draw={rgb,255: red,167; green,167; blue,167}]
\tikzstyle{new edge style 5}=[-, dashed]
\tikzstyle{new edge style 6}=[-, dashed, draw=blue]
\tikzstyle{new edge style 7}=[draw=blue, ->]
\tikzstyle{new edge style 8}=[draw=blue, <-]
\tikzstyle{new edge style 9}=[draw=blue, dashed, ->]
\tikzstyle{new edge style 10}=[draw=blue, dashed, <-]
\tikzstyle{new edge style 11}=[dashed, ->]
\tikzstyle{new edge style 12}=[dashed, <-]
\tikzstyle{new edge style 13}=[-, dashed, draw={rgb,255: red,128; green,128; blue,128}]
\tikzstyle{new edge style 14}=[draw=red, <-]

\usepackage{pgfkeys}

\usepackage{mathtools}

\usepackage{subcaption}

\usepackage{dsdshorthand}

\usepackage{romanbar}

\usepackage[vcentermath]{youngtab}

\renewcommand\vol{\mathop{\mathrm{vol}}}

\renewcommand\bm{\mathbf{m}}

\newcommand{\thalf}{\tfrac{1}{2}}

\makeatletter
\def\@fpheader{\ }
\makeatother

\renewcommand{\cO}{{\mathcal{O}}}

\def\({\left(} \def\){\right)}
\def\[{\left[} \def\]{\right]}

\newcommand{\nhb}{\hat{\nabla}^{\bot}}

\newcommand{\hga}{\hat\gamma}

\newcommand{\md}{\mathcal{D}}

\newcommand{\hp}{\hat{\Romanbar{II}}}

\newcommand{\ve}{\varepsilon}
\renewcommand{\S}{\Sigma}

\newcommand{\II}{\Romanbar{II}}

\title{
Effective theory for fusion of conformal defects
}
\author{Petr Kravchuk, Alex Radcliffe, Ritam Sinha}
\affiliation{
Department of Mathematics, King's College London, Strand, London, WC2R 2LS, UK
}
\date{}
\abstract{We construct an effective field theory for fusion of conformal defects of any codimension in $d\geq 3$ conformal field theories. We fully solve the constraints of Weyl invariance for defects of arbitrary shape on general curved bulk manifolds and discuss the simplifications that arise for spherical defects on the conformal sphere. As applications, we study the structure of cusp anomalous dimensions in the anti-parallel lines limit and derive high-energy spin-dependent asymptotics for the one-point functions of bulk operators. We point out the potential importance of defects that break transverse rotations and initiate a classification of their Weyl anomalies.
}

\begin{document}

\maketitle
\newpage
\section{Introduction}

In this paper we study the limit of CFT correlation functions in which two conformal defects are brought close together, i.e.\ the problem of ``fusion'' of conformal defects. While the analogous problem for topological defects has a long history, the conformal case has received much less attention. In $d=2$ dimensions, it has been studied in~\cite{Bachas:2007td, Bachas:2013ora, Konechny:2015qla}. In $d\geq 3$, fusion of conformal defects (or correlation functions of two defects) was analyzed in~\cite{Soderberg:2021kne, Rodriguez-Gomez:2022gbz, SoderbergRousu:2023zyj}.\footnote{Furthermore, the papers~\cite{Diatlyk:2024qpr,Diatlyk:2024zkk} appeared while this paper was in preparation, see the end of this section for further comments.} For the most part, these works focused on calculations in specific theories. The main goal of this work is to study the general properties of fusion of conformal defects and to explore some initial applications.

To be more concrete, we consider a correlation function
\be\label{eq:maincorrelator}
	\<\cD_1\cD_2\cdots\>
\ee
with insertions of two conformal defects $\cD_1$ and $\cD_2$. In this paper, we use $\<\cdots\>$ to denote partition functions normalised so that the partition function with no insertions is $1$. In particular, we do not normalise by defect partition functions in any way.

For simplicity, we assume that no defect local operators are present on $\cD_i$, although the forthcoming discussion can be easily adapted to the more general case. The dots ``$\cdots$'' represent any other insertions which stay $O(1)$ distance away from the defects $\cD_i$. We do not make any assumptions about the shape or the topology of the defects $\cD_i$, except that they do not intersect and are homotopic to each other, see figure~\ref{fig:fusionexamples}. Our goal is to determine the behaviour of the correlation function~\eqref{eq:maincorrelator} in the limit where $\cD_1$ and $\cD_2$ approach each other and eventually fully overlap. In other words, if $L$ is the typical separation between $\cD_1$ and $\cD_2$, we are interested in the limit $L\to 0$.

Our main assumption is that the appropriate language for this problem is that of a low-energy effective theory (EFT) on top of a new conformal defect $\cD_\Sigma$,
\be\label{eq:mainEFT}
	\<\cD_1\cD_2\cdots\> \sim \<\cD_\Sigma [e^{-S_\text{eff}}]\cdots\>,
\ee
where $S_\text{eff}$ describes the perturbation of the conformal defect $\cD_\Sigma$ by identity and the irrelevant operators. The effective action $S_\text{eff}$ depends on the relative configuration of $\cD_1$ and $\cD_2$, and the high-dimension operators are suppressed by powers of defect separation $L$. In practice, this EFT will express the correlation function~\eqref{eq:maincorrelator} as an asymptotic series in powers of $L/R$, where $R$ is the IR length scale set by the size of $\cD_\Sigma$ and the other insertions.

	\begin{figure}[t]
	\centering
	\begin{subfigure}[t]{0.45\textwidth}
		\centering
		\begin{tikzpicture}
			\node[anchor=south west,inner sep=0] (pdf) at (0,0) {\includegraphics[scale=0.75]{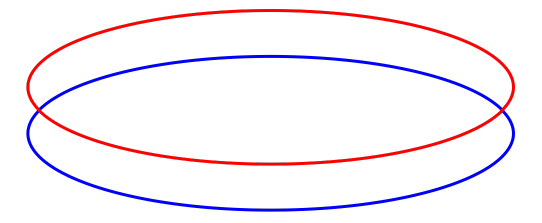}};
			
			\draw[<->, thick, black] (3.5,0.2) -- (3.5,0.7);  
			
			\node at (0,1) {$\cD_1$};
			\node at (0,1.7) {$\cD_2$};
			\node at (3.8, 0.5) {$L$};
			
			\node at (0,-1){};
			
		\end{tikzpicture}
	\end{subfigure}
	~
	\begin{subfigure}[t]{0.45\textwidth}
		\centering
		\scalebox{.8}{\includegraphics{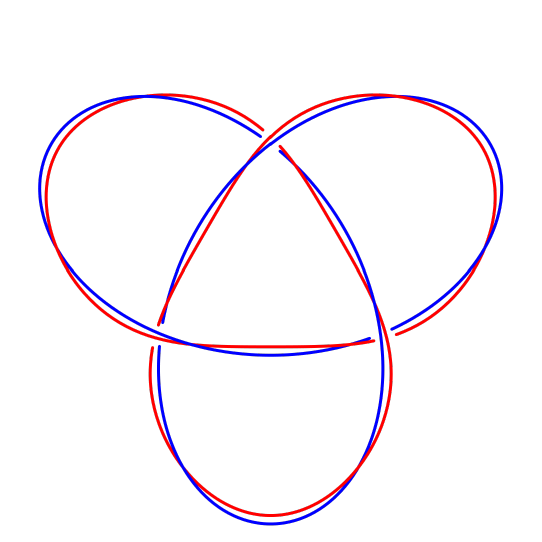}}
	\end{subfigure}
	\caption{Examples of fusion configurations. Left: a simple fusion configuration, the separation scale $L$ is taken to $0$. Right: a generic fusion configuration with knotted and interlinked defects.}
	\label{fig:fusionexamples}
\end{figure}
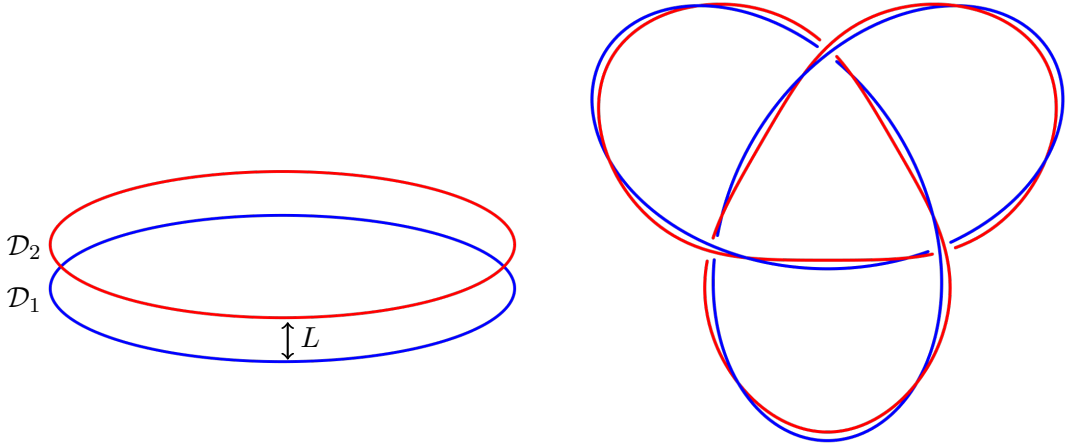

We defer the more precise discussion of~\eqref{eq:mainEFT} to section~\ref{sec:Seff}, and discuss here only the motivation behind and the status of this assumption. The key intuition is that for small $L$, we may view the pair of defects $\cD_1,\cD_2$ as giving a UV definition of a new non-conformal defect. In this picture $L$ plays the role of the UV scale analogous to, for example, a lattice scale. The process of taking the limit $L\to 0$ (equivalently, $R\to \oo$) can be then understood as an RG flow, which we expect to generically terminate at a conformal fixed point $\cD_\Sigma$. Small but finite $L\ll R$ can then be described in terms of an EFT around $\cD_\Sigma$.\footnote{The defect $\cD_\S$ need not be simple, although we expect that generically it will be, see section~\ref{sec:simplicity}.} Note that this point of view on fusion is not by any means new and was discussed already in~\cite{Bachas:2007td}.

Our main result is a general formalism for writing down the effective action $S_\text{eff}$ appearing in~\eqref{eq:mainEFT}, properly taking into account the constraints imposed by Weyl invariance on either side of the equation. We describe it in section~\ref{sec:Seff}. In the case of local operator product expansion, Weyl invariance fixes the form of the contribution of descendants in terms of that of the primary. In the defect case, it constrains the terms allowed in the effective action, and we classify all possible terms. In order to achieve this, we derive in section~\ref{sec:confgeom} a Weyl-invariant parameterisation of the relative position of the two defects and show that it is possible to construct from it a Weyl-invariant bulk metric $\hat g$, to which we refer as the fusion metric.

As an application of the general formalism, we study in section~\ref{sec:Wilson} the implications of the effective theory for the cusp anomalous dimension $\G_\text{cusp}$, which appears when two line defects meet with an opening angle $\a$. We show that a systematic small-$\a$ expansion of $\G_\text{cusp}$ and its excited versions can be computed from $S_\text{eff}$. We discuss some simple consequences of this for supersymmetric Wilson lines in $\cN=4$ SYM and compute a subleading Wilson coefficient from the known expressions for $\G_\text{cusp}$. We also make some comments regarding the connection of ``ultra-soft'' degrees of freedom~\cite{Pineda:2007kz,Correa:2012nk} and the simplicity of $\cD_\S$ in section~\ref{sec:simplicity}.

As another application,  in section~\ref{sec:operatorasymptotics} we consider the two-point function  $\<\cD\bar \cD\>$ of a defect and its conjugate from two points of view. On the one hand, the asymptotic expansion~\eqref{eq:mainEFT} applies. On the other hand, each defect can be expanded in terms of local operators~\cite{Gadde:2016fbj}. Consistency of these two expansions allows us to derive high-energy asymptotics of the one-point functions $\<\cO\cD\>$ of local operators $\cO$ in the presence of the defect $\cD$, as functions of both the scaling dimension and the spin of $\cO$.\footnote{A similar computation of the density of states in terms of thermal EFT was recently performed in~\cite{Benjamin:2023qsc}.}

An important contribution to the effective action $S_\text{eff}$ comes from terms which compensate for the difference of Weyl anomalies between the two sides of~\eqref{eq:mainEFT}. This is however complicated by the fact that $\cD_\Sigma$ does not have to preserve the same space-time symmetries as $\cD_1$ and $\cD_2$.  Indeed, if we view the pair $\cD_1\cD_2$ as the UV definition of an RG flow, this definition explicitly breaks transverse rotations around the defects. This transverse rotation symmetry may or may not be restored at the IR fixed point $\cD_\Sigma$. If it is not, we show in section~\ref{sec:anomalies} that the possible Weyl anomalies of $\cD_\Sigma$ are more general than usually considered in the literature (see~\cite{Chalabi:2021jud} for a review), and classify them in the case of line defects in general $d$ and surface defects in $d=4$.

While we believe that the argument in favor of~\eqref{eq:mainEFT} presented above is generically correct, we stress that it is not something that we can prove (with a reasonable degree of rigor) in general CFTs. This is contrast to the situation with the local operator product expansion (OPE), where quite formal arguments exist~\cite{Mack:1976pa}. Even though local OPE can be formally viewed as fusion of 0-dimensional defects, one has to be cautious about trying to take this analogy too far. We discuss the similarities and the key differences between fusion of conformal defects and local OPE in section~\ref{sec:fusionandOPE}. 

In section~\ref{sec:summary} we give a somewhat more detailed summary of our findings.

We conclude in section~\ref{sec:discussion}. Appendices contain some additional details, most notably the 2-derivative couplings for the identity operator are classified in appendix~\ref{app:two-derivatives}.

\paragraph{Note added} When this work was largely complete, the papers~\cite{Diatlyk:2024zkk,Diatlyk:2024qpr} appeared which partially overlap with our discussion. The paper~\cite{Diatlyk:2024zkk} discussed some general properties of defect fusion but did not study the effective action beyond the leading cosmological constant term, focusing instead on multiple explicit examples of fusion. We have added comparisons of some our results with these examples where applicable. The more recent paper~\cite{Diatlyk:2024qpr} studied fusion in the case when $\cD_1$ and $\cD_2$ are conformal boundaries with the bulk CFT in between them. They also focused only on the leading term and derived the high-energy asymptotics of one-point functions. This partially overlaps with the more general results of our section~\ref{sec:operatorasymptotics}. We have also learned from~\cite{Diatlyk:2024qpr} of an upcoming related work~\cite{zoharfuture}.

\subsection{Summary}
\label{sec:summary}
This section contains a condensed summary of our main results. For the more detailed discussion we refer the reader to the main body of the paper.

We work in a conformal field theory in $d> 2$ dimensions and consider a pair of $p$-dimensional conformal defects $\cD_1$ and $\cD_2$ which are approximately parallel to one another, see figure~\ref{fig:fusionexamples}. We denote the codimension by  $q=d-p$. We are interested in the limit in which the two defects fuse. That is, we consider a smooth family of configurations parameterised by an ``average distance'' $L>0$ such that in the limit $L\to 0$ the two defects coincide. We only use $L$ to denote a scale (except in some examples where it is given a precise meaning), and the precise dependence on the shape of the defects will be captured by other data.

At fixed $L$, a generic configuration of this sort breaks most conformal symmetries even if $\cD_1$ and $\cD_2$ are spherical defects in flat space. We work under the assumption that in the limit $L\to 0$ the conformal symmetry is restored, and a conformal defect $\cD_\S$ is obtained. Such a fusion process can be described in terms of an RG flow. In this picture, the product $\cD_1\cD_2$ can be viewed as the UV definition of a non-conformal defect which in the IR flows to $\cD_\S$. In the neighbourhood  of the IR fixed point $\cD_\S$, i.e.\ in the limit of small $L$, this RG flow can be described in terms of an effective action $S_\text{eff}$ on top of $\cD_\S$. Our goal is to analyze constraints on $\cD_\S$ and to describe the general form of the effective action $S_\text{eff}$.

\paragraph{Kinematics} To study these questions systematically, we consider the problem on a general curved bulk manifold $M$ and allow arbitrarily shaped $\cD_1$ and $\cD_2$. We assume that the bulk theory and all conformal defects are diffeomorphism- and Weyl-invariant (modulo Weyl anomalies). This allows us to use diffeomorphism and Weyl invariance to constrain $\cD_\S$ and the effective action $S_\text{eff}$.

The effective action $S_\text{eff}$ has to depend on the relative position of $\cD_1$ and $\cD_2$. To the leading order in $L$ we can parameterise the position of $\cD_2$ relative to $\cD_1$ by specifying a normal vector field $v^\mu=O(L)$ on $\cD_1$ such that, up to a reparameterisation of $\cD_2$,
\be\label{eq:vleading}
	X_2^\mu(z)=X_1^\mu(z)+v^\mu(z)+O(v^2),
\ee
where $z^a$ are defect coordinates and $X_i^\mu(z)$ are the functions which embed the defects $\cD_i$ into the bulk manifold $M$. Since we want to perform a systematic expansion of correlation functions in powers of $v\sim L$, it is necessary to extend~\eqref{eq:vleading} to higher orders in $v$. In section~\ref{sec:displacement} we provide such an extension which is both Weyl- and diffeomorphism-invariant. This is achieved by constructing a one-parameter family of deformations $X_\text{def}(z,t)$ such that, up to a reparameterisation of $\cD_2$,
\be
	X_\text{def}(z,0)&=X_1(z),\\
	X_\text{def}(z,1)&=X_2(z).
\ee
The function $X_\text{def}(z,t)$ is defined by a Weyl-invariant partial-differential equation~\eqref{eq:surface_evolution} with the initial condition
\be
	X^\mu_\text{def}(z,0)&=X^\mu_1(z),\quad \ptl_t X^\mu_\text{def}(z,0)=v^\mu(z).
\ee
In particular, equation~\eqref{eq:surface_evolution} can be solved order-by-order in $v$, providing a systematic Weyl-invariant extension of~\eqref{eq:vleading}.\footnote{Note that we do not claim that our definition of $v^\mu$ is unique. The non-trivial claim is that a Weyl- and diffeomorphism-invariant definition exists.} 

Note that this construction also gives a Weyl-covariant definition for the coupling to the displacement operator at non-linear orders. This in particular implies that there exist renormalisation schemes in which the displacement operator $D$ transforms as a Weyl tensor, without any anomalous terms proportional to $D$.

In section~\ref{sec:dilaton} we show that the vector field $v^\mu(z)$ on $\cD_1$ allows us to define a preferred Weyl frame in the neighbourhood of the two defects. It is clear that $v^\mu(z)$ defines a length scale on $\cD_1$ via $\ell^2=g_{\mu\nu} v^\mu v^\nu$ and 
\be
	\hat g_{\mu\nu}=\ell^{-2} g_{\mu\nu}
\ee
becomes a Weyl-invariant bulk metric, defined on $\cD_1$. The main result of section~\ref{sec:dilaton} that the length scale $\ell$ can be defined in a covariant way in the neighbourhood of $\cD_1$ and $\cD_2$. This allows us to extend also $\hat g_{\mu\nu}$ to a bulk neighbourhood. This dramatically simplifies the problem of constructing $S_\text{eff}$ as it can now be written in terms of the Weyl-invariant metric $\hat g_{\mu\nu}$ instead of the physical metric $g_{\mu\nu}$. We will refer to $\hat g_{\mu\nu}$ as the fusion metric. 

In section~\ref{sec:fusionmetricprops} we show that the fusion metric  $\hat g_{\mu\nu}$ is not completely generic and satisfies certain local constraints on $\cD_1$. Up to two-derivative order the constraints are exhausted by $\hat g_{\mu\nu} v^\mu v^\nu=1$, the vanishing of the trace of the second fundamental form $\hp^\mu=0$, and the vanishing of the purely normal components of the Schouten tensor. We also show that for spherical defects in flat space stronger constraints such as $\hp^\mu_{ab}=0$ follow.

Sections~\ref{sec:exampleI} and~\ref{sec:exampleII} contain various examples of the above constructions. In particular, in section~\ref{sec:exampleI} we describe in full generality the vector field $v^\mu$ for a pair of spherical defects (of any codimension) embedded in flat space. In section~\ref{sec:exampleII} we extend some of the examples of section~\ref{sec:exampleI} by computing the fusion metric for them.

\paragraph{Effective action} Using the fusion metric $\hat g_{\mu\nu}$ we can systematically construct the couplings in $S_\text{eff}$ in a derivative expansion suppressed by powers of $L$. The leading term comes from couplings to the identity operator, and in particular from the cosmological constant term
\be
	S_\text{eff}\ni -a_0\int\,d^p z\,\sqrt{\hga} = -a_0\int\,d^p z\,\sqrt{\g}\ell^{-p}\sim L^{-p}.
\ee
We discuss this term in detail in section~\ref{sec:cosmo} and prove theorems~\ref{thm:oppositesattract} and~\ref{thm:cauchyschwarz} for the Wilson coefficient $a_0$. The former states that $a_0$ is non-negative, $a_0\geq 0$, while the latter provides inequalities for $a_0$ between different fusion processes. 

In section~\ref{sec:subleading} we study one-derivative identity operator couplings, which only arise for line defects in $d=3$. We discuss their relevance in pure Chern-Simons theory. In section~\ref{sec:displacment_and_schemes} we discuss the contributions of irrelevant operators. Appendix~\ref{app:two-derivatives} contains a discussion of two-derivative identity couplings. For example, for line defects in $d=4$,
\be\label{eq:2-deriv-summary}
S^{(2)}_\text{eff} = &\int dz\sqrt{\hga}\p{
	a_{2,1}\hat\nabla_t^\perp v\. \hat\nabla_t^\perp v+a_{2,2}\hat R+ia_{2,3}\hat P_{tv}+a_{2,4}C_{vtvt}+ia_{2,5}\tl C_{vtvt}
},
\ee
where $\hat P_{\mu\nu}$ is the Schouten tensor,\footnote{In principle, the Schouten tensor can be expressed in terms of the Ricci tensor. However, as we discuss in section~\ref{sec:fusionmetricprops}, Schouten tensor generally has simpler form for fusion metrics than the Ricci tensor.} $C$ is the Weyl tensor and $t^\mu$ is the unit tangent to the line defect. Hats indicate that the quantities are computed in the fusion metric.

\paragraph{Constraints on $\cD_\S$} The position of the IR conformal defect $\cD_\S$ can be taken by convention to exactly coincide with $\cD_1$. This choice is possible because the effective action $S_\text{eff}$ will involve a coupling to the displacement operator, and we can always compensate for our choice of the position of $\cD_\S$ by modifying this coupling.\footnote{An alternative choice is to set the displacement coupling to $0$, in which case we would need to write down a derivative expansion for the position of $\cD_\S$.} We discuss these aspects further in section~\ref{sec:displacment_and_schemes}.

Consider the case when $\cD_1$ and $\cD_2$ are flat parallel defects in flat space. From the RG perspective, we treat $\cD_1\cD_2$ as the UV starting point of an RG flow. This configuration only preserves $\SO(q-1)$ transverse rotations, and there is no guarantee that the symmetry will be enhanced to $\SO(q)$ at the IR fixed point. Therefore, it is possible that $\cD_\S$ breaks $\SO(q)$ down to $\SO(q-1)$. In fact, curvature couplings allow for breaking to even smaller symmetry groups. Whenever some transverse rotations are preserved at the IR fixed point, the requirement that $\cD_\S$ be a stable fixed point implies that it should not have relevant operators with various transverse spins. Whenever $\cD_\S$ breaks transverse rotations, it must support tilt operators and the corresponding conformal manifold, analogously to the breaking of global symmetries. We discuss these questions in detail in section~\ref{sec:transverse}.

In section~\ref{sec:simplicity} we argue that generically we expect the IR defect $\cD_\S$ to be simple. We also highlight some important exceptions, including in particular the fusion of Wilson lines at weak coupling. 

\paragraph{Cusps and Wilson lines in $\cN=4$ SYM} In section~\ref{sec:Wilson} we apply the general effective theory formalism in the context of cusp anomalous dimension. When $\cD_1$ and $\cD_2$ meet at an opening angle $\a$, the resulting cusp carries a non-trivial scaling dimension $\G_\text{cusp}$. Via the exponential map, this configuration can be mapped to the cylinder $\R\x S^{d-1}$ where the defects $\cD_1$ and $\cD_2$ are inserted along $\R$, separated by angular distance $\a$ on $S^{d-1}$. In this picture $\G_\text{cusp}$ becomes the ground state energy of the Hamiltonian $H$ that generates translations along $\R$. In the limit $\a\to 0$ the defects fuse and the effective theory can be applied. We explain how the spectrum of $H$ can be obtained from Rayleigh-Schr\"odinger perturbation theory on $\cD_\S$. 

We then briefly discuss the implications of this for supersymmetric Wilson lines in planar $\cN=4$ SYM. We argue that the IR defect $\cD_\S$ is trivial and we show that the identity operator contribution to $S_\text{eff}$ is responsible for the value of $\G_\text{cusp}$ with error at most $O(\a^2)$, and possibly to all orders in $\a$. More concretely,
\be
	\G_\text{cusp}(\theta,\f=\pi-\a)=-\frac{a_0(\theta)}{\a}-3a_{2,2}(\theta)\a+O(\a^2),
\ee
where the two-derivative Wilson coefficient $a_{2,2}$ defined in~\eqref{eq:2-deriv-summary}. Matching the known expressions for $\G_\text{cusp}$ in the ladder and in the strong-coupling limits, we compute the value of $a_{2,2}$.

\paragraph{One-point functions of bulk operators} In section~\ref{sec:operatorasymptotics} we apply fusion EFT to the two-point function $\<\cD\bar \cD\>$. Comparing it with the expansion of this two-point function in terms of bulk operators~\cite{Gadde:2016fbj}, we derive the large-$\De$ asymptotic of one-point functions  $\<\cO\cD\>$ of local bulk operators $\cO$ in the presence of $\cD$. This is done for defects of any dimension $p$. We also derive the spin-dependent asymptotics for line defects in $d=3$ and $d=4$. As an example, for line defects in $d=3$ we find
\be
&\log \r(\De,J)\sim \sqrt{8\pi a_0\De}\p{1-\thalf j^2-\tfrac{1}{4}j^6+O(j^8)}+O(\log \De),
\ee
where $j=J/\De\leq 1$ and $\r(\De,J)$ is the one-point function density for local operators (both descendants and primaries) with scaling dimension $\De$ and spin projection $J$. The expansion in small $j$ is given for illustration purposes, the full $j$-dependence is determined by the transcendental equation~\eqref{eq:O1eqn}.

\paragraph{Anomalies} In general, the conformal defects $\cD_1,\cD_2$ and $\cD_\S$ are not exactly Weyl-invariant but can carry non-trivial Weyl anomalies. In this case, the effective action $S_\text{eff}$ also needs to transform non-trivially under Weyl transformations so that both sides of~\eqref{eq:mainEFT} transform in the same way. To achieve this, Weyl-anomaly matching terms need to be added to $S_\text{eff}$, and we derive their general form in section~\ref{sec:anomalymatching}. The resulting contribution is
\be
\cA_{\cD_1}(g,-\log \ell)+\cA_{\cD_2}(g,-\log \ell)-\cA_{\cD_\S}(g,-\log \ell),
\ee
where $\cA(g,\w)$ is the defect Weyl anomaly action for a finite Weyl transformation $\w$, and $\ell$ is evaluated on the appropriate submanifolds. As we discuss in section~\ref{sec:fusionandOPE}, the existence of the scale $\ell$ is crucial for this construction, and the analogous approach doesn't work in the case of local operator OPE.

Generally speaking, the Weyl anomaly $\cA_{\cD_\S}$ belongs to a class that is slightly more general than usually considered in the literature (see~\cite{Chalabi:2021jud} for a review) since the defect $\cD_\S$ might break some transverse rotations. In sections~\ref{sec:anomalyline} and~\ref{sec:anomalysurface} we begin the classification of Weyl anomalies for conformal defects that break the transverse rotations down to $\SO(q-1)$. We find that such line defects admit a new B-type Weyl anomaly in $d=3$ and no Weyl anomalies in $d>3$. For surface defects, we find 5 new B-type terms in $d=4$.

\section{Conformal geometry of a pair of defects}
\label{sec:confgeom}
In subsection~\ref{sec:displacement} we derive a way of parameterizing the relative position of two conformal defects $\cD_1$ and $\cD_2$ in a Weyl-invariant way. In subsection~\ref{sec:dilaton} we show that similar ideas lead to a definition of a canonical Weyl frame in the neighbourhood of $\cD_1,\cD_2$. In section~\ref{sec:fusionmetricprops} we discuss the special relations satisfied in this canonical Weyl frame. We give explicit examples of these constructions in subsections~\ref{sec:exampleI} and~\ref{sec:exampleII}.

The defects we consider have dimension $p>0$ and co-dimension $q$, with $p+q=d$ the dimension of the bulk. We assume that $d>2$. We work in Euclidean signature and on a bulk manifold $M$ with a general metric $g$. We assume that the defect manifold is $N$ for both $\cD_1$ and $\cD_2$, embedded into $M$ via $X_1:N\to M$ and $X_2:N\to M$ respectively. We use Greek letters for bulk indices ($x^\mu$) and Latin letters for defect indices ($z^a$). We will use $x$ to denote points in $M$ and $z$ to denote points in $N$, and we will often use $\cD_i$ to refer to the images $X_i(N)$. To declutter the notation, and when there is no fear of confusion, we may use $z$ to refer to the point $X_i(z)$ on $\cD_i$. Some additional comments regarding the geometry are given in appendix~\ref{app:geometry}.

\subsection{A Weyl-invariant displacement}
\label{sec:displacement}

We focus on describing the position of $\cD_2$ relative to $\cD_1$ when the defects are close to one another. We will do this by describing the possible deformations of the embedding function $X_1$ and then determining the deformation that yields $X_2$ (modulo defect diffeomorphisms).

To the leading order, a deformation of $X_1$ can be parameterised by a bulk vector field $v^\mu(z)$ defined on $\cD_1$,\footnote{Formally, one would describe $v$ as a section of the vector bundle over $N$ that one obtains from the tangent bundle $TM$ by pulling back along $X_1$. To simplify the discussion, we will appeal to more informal descriptions of various quantities throughout the paper.} so that the deformed embedding function $X_{1,v}$ is defined as
\be\label{eq:leadingdeformation}
	X_{1,v}^\mu(z) = X_1^\mu(z)+v^\mu(z)+O(v^2).
\ee
Components of $v^\mu(z)$ tangent to $\cD_1=X_1(N)$ do not affect the deformed defect $X_{1,v}(N)$ at this order because they can be absorbed into diffeomorphisms of $N$. We can fix this redundancy by requiring that $v^\mu$ be normal to $\cD_1$.

This parameterisation is explicitly Weyl-invariant. Indeed, Weyl transformations only affect the metric which enters here only through the condition that $v$ is orthogonal to $\cD_1$, and angles are unchanged by Weyl rescalings. However, so far we have only defined the deformation $X_{1,v}$ to the leading order in $v$.  Our goal, on the other hand, is to define the deformation to all non-linear orders in the deformation parameter $v$.

A standard way to extend the definition~\eqref{eq:leadingdeformation} to non-linear orders is by using affine geodesics with initial velocity $v^\mu$. For this, we first introduce a $t$-dependent deformation $X_\text{def}(z,t)$, so that the final deformation is given by
\be
	X_{1,v}^\mu(z) \equiv  X_\text{def}(z,t=1).
\ee
We could then require that for a fixed $z$, $X_\text{def}(z,t)$ satisfies the geodesic equation
\be\label{eq:geodesic}
	\nabla_t^2 X^\mu_\text{def} = \ptl_t^2 X_\text{def}^\mu+\G^\mu_{\a\b}\ptl_t X^\a_\text{def}\ptl_t X^\b_\text{def}=0,\quad{\color{gray}\text{(not our definition)}}
\ee
subject to the initial conditions
\be\label{eq:initialconditions}
	X^\mu_\text{def}(z,0)=X^\mu_1(z),\quad \ptl_t X^\mu_\text{def}(z,0)=v^\mu(z).
\ee
This is essentially a version of Riemann normal coordinates in a neighbourhood of a submanifold. While this definition is diff-invariant,\footnote{Something like $X_{1,v}^\mu(z)\equiv X_1^\mu(z)+v^\mu(z)$ is, on the other hand, not even diff-invariant.} it fails to be Weyl-invariant since the Christoffel symbols $\G$ transform inhomogeneously as
\be
	\delta_\w \Gamma^\rho{}_{\mu\nu}&=\delta^\rho_\mu\partial_\nu\w+\delta^\rho_\nu\partial_\mu\w-g_{\mu\nu}\partial^\rho\w.
\ee
It is easy to check that the problem arises already in the $O(v^2)$ term in $X_{1,v}^\mu$ resulting from this definition,
\be\label{eq:problematic_displacement}
	X^\mu_{1,v}=X_1^\mu+v^\mu-\half \G^\mu_{\a\b} v^\a v^\b+O(v^3).
\ee

To overcome this problem we use a Weyl-invariant equation in place of~\eqref{eq:geodesic},
\be
	\label{eq:surface_evolution}
	\nabla_t^2X^\mu_\text{def}+\frac{2}{p}\ptl_t X^\mu_\text{def} \ptl_t X^\s_\text{def} \II_\s -\frac{1}{p}(\ptl_t X_\text{def})^2 \II^\mu+\half \ptl_a X^\mu_\text{def} \ptl^a(\ptl_t X_\text{def})^2=0,
\ee
which is now a partial-differential equation in $z,t$.\footnote{The intuition behind this equation is as follows. The equation $\nabla_t^2 X_\text{def} = 0$ transforms by first derivatives of the Weyl factor on the defect $X_\text{def}(\cdot,t)$. So, if we can fix the first derivatives of a Weyl scale on this defect, we can turn $\nabla_t^2 X_\text{def} = 0$ into a conformally-invariant equation. We first restrict to the so-called minimal scales in which $\II^\mu=0$, which fixes the normal derivatives of the scale. We then fix the tangential derivatives of the scale by requiring that $(\ptl_t X_\text{def})^2=1$. It is then easy to see that in this scale~\eqref{eq:surface_evolution} reduces to $\nabla_t^2 X_\text{def} = 0$.} Here, $\II^\mu=\II^\mu_{ab}\hat g^{ab}$ is the trace of the second fundamental form calculated for the embedding $X_\text{def}(\.,t)$. This equation is to be solved with initial conditions~\eqref{eq:initialconditions}. Using the transformation rules given in appendix~\ref{app:geometry}, it is easy to check that the equation is invariant under Weyl transformations  as long as long as
\be
	\ptl_t X_\text{def}\.\ptl_a X_\text{def}=0.
\ee
This condition is in turn true because it is satisfied by the initial condition~\eqref{eq:initialconditions} and equation~\eqref{eq:surface_evolution} implies that
\be
	\ptl_t\p{\ptl_t X_\text{def}\.\ptl_a X_\text{def}} = -\frac{2}{p} \ptl_t X_\text{def}^\nu \II_\nu \p{\ptl_t X_\text{def}\.\ptl_a X_\text{def}}.
\ee

The solution to~\eqref{eq:surface_evolution} can be systematically constructed order by order in $v$. For example, the first two orders yield
\be\label{eq:good_expansion}
	X^\mu_{1,v}=&X_1^\mu+v^\mu-\half \G^\mu_{\a\b} v^\a v^\b-\frac{1}{p}v^\mu v^\s \II_\s +\frac{1}{2p}v^2 \II^\mu-\frac{1}{4} \ptl_a X^\mu_1 \ptl^a (v^2)+O(v^3).
\ee
Higher orders can be easily obtained by taking further $t$-derivatives of~\eqref{eq:surface_evolution} at $t=0$. Unlike~\eqref{eq:problematic_displacement}, this expansion is Weyl-invariant. This has two important consequences.

Firstly, this construction allows us to parameterise the relative position of the defects $\cD_1$ and $\cD_2$ using a vector field $v^\mu$. For this, we search for $v^\mu$ such that
\be
	X_{1,v}(z) = X_2(\vf(z))
\ee
for some diffeomorphism $\vf:N\to N$. Note that by construction, $v^\mu$ does not transform under Weyl transformations. On the other hand, its length
\be
	\ell(z) \equiv \sqrt{g_{\mu\nu}v^\mu v^\nu}
\ee
transforms with Weyl weight $1$, as expected from an infinitesimal distance. The function $\ell(z)$ formalises the idea of ``coordinate-dependent'' distance between the two defects and will feature prominently (alongside $v^\mu$) in the effective actions that we study in section~\ref{sec:Seff}.

Secondly, this proves the existence of Weyl-covariant schemes for the coupling of the displacement operator. Normally, the displacement operator $D_\mu$ is defined so that
\be
	\<\cD_1[e^{-\int d^p z\sqrt{\g} v^\mu D_\mu}]\cdots \>,
\ee
where $\g$ is the defect metric, computes a correlator with $\cD_1$ displaced to $X_{1}'=X_1+v+O(v^2)$ (see, for example~\cite{Billo:2016cpy}). Specifying this displacement to higher orders in $v$ is equivalent to defining a renormalisation scheme for correlation functions of $D_\mu$. The fact that we are able to make the Weyl-invariant definition $X_1' \equiv X_{1,v}$ implies\footnote{For this it is important that $X_{1,v}$ has a local series expansion in terms of $v$ and its derivatives, see~\eqref{eq:good_expansion}. If this were violated, we wouldn't expect to be able to use $v$ as the expansion parameter in conformal perturbation theory.} that the corresponding renormalisation scheme contains no scale anomalies in contact terms of $D_\mu$ with itself. Put differently, the action $\int d^p z\sqrt{\hat g} v^\mu D_\mu$ can be made Weyl-invariant at quantum level.

We should note that other Weyl-invariant definitions of $X_{1,v}$ might be possible. We will not rely on any special properties of the above definition beyond those already mentioned.

\subsection{Example: spherical defects in flat space I}
\label{sec:exampleI}
The discussion in section~\ref{sec:displacement} applies in general curved backgrounds and for arbitrarily shaped defects. If the spacetime manifold is flat $M=\R^d$ (more accurately, its conformal compactification $S^d$) and the defects $\cD_1$ and $\cD_2$ are flat or spherical, the definitions can be made much more explicit, which is the goal of this section. In what follows we set $g_{\mu\nu}=\de_{\mu\nu}$, and we use ``spherical'' to refer to both spherical and flat defects (the latter understood as spheres of infinite radius) of any dimension. 

To begin, we first fix our conventions for conformal Killing vectors (CKVs). Let $\xi^\mu(x)$ be a CKV on $M=\R^d$, so that it satisfies the equation
\be\label{eq:CKV}
\ptl_\mu \xi_\nu +\ptl_\nu \xi_\mu = \frac{2}{d}\de_{\mu\nu}(\ptl\.\xi).
\ee
We define the following basis of Killing vectors
\be
	\bd^\mu(x) &= x^\mu,\\
	(\bp_\nu)^\mu(x) &= \de^\mu_\nu,\\
	(\bk_\nu)^\mu(x) &= 2x^\mu x_\nu - \de^\mu_\nu x^2,\\
	(\bm_{\r\s})^\mu(x) &= \de^\mu_\r x_\s-\de^\mu_\s x_\r.
\ee
The general solution to~\eqref{eq:CKV} is a linear combination of these basis elements,
\be
	\xi=\l \bd+a^\mu \bp_\mu+b^\mu \bk_\mu +\half w^{\mu\nu}\bm_{\mu\nu}.
\ee
The CKVs form a Lie algebra under the Lie bracket of vector fields -- the Lie algebra of conformal transformations of $\R^d$, and the non-zero values of the corresponding Killing form are
\be\label{eq:killing}
	\<\bd,\bd\> = 2d,\quad \<\bp_\mu,\bk_\nu\>=-4d\de_{\mu\nu},\quad \<\bm_{\mu\nu},\bm_{\r\s}\> = -2d\p{\de_{\mu\r}\de_{\nu\s}-\de_{\mu\s}\de_{\nu\r}}.
\ee

Recall that $\vf= e^\xi$ is a conformal transformation that is constructed as follows. To compute $\vf(x_0)$ we solve the differential equation
\be\label{eq:expmap}
\dot x^\mu(t) = \xi^\mu(x(t))
\ee
with the initial condition $x(0)=x_0$ and then set $\vf(x_0) = x(1)$.

We now fix a CKV $\xi$ and define
\be\label{eq:flatsoln}
X_\text{def}(z,t) = e^{t\xi}(X_1(z)),
\ee
and suppose that $X_\text{def}(z,t)$ solves~\eqref{eq:surface_evolution} at $t=0$. We will discuss this condition later. We further assume that $\xi^\mu$, when restricted to $\cD_1$, is orthogonal to $\cD_1$ so that $v=\xi\vert_{\cD_1}$ provides an initial condition~\eqref{eq:initialconditions} of the kind discussed in section~\ref{sec:displacement}. We will see later that there is a healthy supply of such $\xi$ for spherical $\cD_1$. We claim that~\eqref{eq:flatsoln} then solves~\eqref{eq:surface_evolution} for all times~$t$.

To prove our claim, we first show that $\xi^\mu$ is orthogonal to the image of $X_\text{def}(\.,t)$ for any~$t$. Indeed, notice that $e^{t\xi}$ is a conformal transformation, and therefore $\xi$ being orthogonal to $e^{t\xi}(X_1(\.))$ is equivalent to $d(e^{-t\xi})(\xi)$ being orthogonal to $X_1(\.)$. Here, $d(e^{-t\xi})(\xi)$ is the pushforward of $\xi$ by the diffeomorphism $e^{-t\xi}$ and can be written as $e^{-t\cL_\xi}\xi$, where $\cL_\xi$ is the Lie derivative. Since $\cL_\xi \xi = [\xi,\xi]=0$, it follows that $d(e^{-t\xi})(\xi)=\xi$ and therefore $d(e^{-t\xi})(\xi)$ is indeed orthogonal to $X_1(\.)$ by our assumptions.

Equations~\eqref{eq:expmap} and~\eqref{eq:flatsoln} show that $\xi(X_\text{def}(z,t))$ coincides with $\ptl_t X_1(z,t)$. Therefore, the above orthogonality statement implies, following the discussion in section~\ref{sec:displacement}, that the left-hand side of~\eqref{eq:surface_evolution} is conformally-invariant when evaluated on $X_\text{def}(z,t)$. By the definition~\eqref{eq:flatsoln}, time evolution of $X_\text{def}(z,t)$ is precisely the conformal transformation $e^{t\xi}$. Therefore, similarly to the above argument for orthogonality, the fact that~\eqref{eq:surface_evolution} is satisfied for $t=0$ implies that it is satisfied for all $t$.

Let us consider a concrete example where $\cD_1$ is a flat defect passing through $x=0$, and spanning the coordinate directions $1,\cdots, p$. In this case,~\eqref{eq:surface_evolution} becomes at $t=0$
\be
\ptl^2_t X_\text{def}^\mu+\half \ptl_a X_\text{def}^\mu \ptl^a(\ptl_t X_\text{def})^2=0.
\ee
In terms of $\xi$ and $X_1$ this becomes
\be\label{eq:flat_surface_evolution_in_xi}
\xi^\nu \ptl_\nu \xi^\mu+\half \ptl_a X_1^\mu \ptl^a \xi^2=0.
\ee
The condition that $\xi$ is orthogonal to $\cD_1$ allows for the following choices for $\xi$:
\begin{itemize}
	\item translations $\bp_\mu$ normal to $\cD_1$, i.e.\ $\mu = p+1,\cdots, d$,
	\item special conformal transformations $\bk_\mu$  normal to $\cD_1$, i.e.\ $\mu = p+1,\cdots, d$,
	\item rotations $\bm_{\mu\nu}$ with at least one index normal to the defect, i.e.\ $\mu = p+1,\cdots, d$ or $\nu = p+1,\cdots, d$, or both.
\end{itemize}
By an explicit calculation, one can check that if the rotations $\bm_{\mu\nu}$ with both indices normal to the defect are excluded, then any linear combination of the remaining CKVs satisfies~\eqref{eq:flat_surface_evolution_in_xi}. We will refer to such $\xi$ as standard deformations of $\cD_1$. Therefore, the standard deformations are spanned by $\bp_\mu, \bk_\mu, \bm_{\mu\nu}$ with $\mu = p+1,\cdots, d$ and $\nu = 1,\cdots, p$.

Let $\mathfrak{h}\subseteq \mathfrak{so}(d+1,1)$ be the subalgebra of conformal transformations which leave $\cD_1$ invariant and $H$ be the corresponding subgroup of $\SO(d+1,1)$. That is, 
\be
	\mathfrak{h}\simeq \mathfrak{so}(p+1,1)\oplus \mathfrak{so}(d-p)
\ee
and $\mathfrak{h}$ is spanned by $\bd,\bp_\mu,\bk_\mu, \bm_{\mu\nu}$ with $\mu,\nu\in 1,\cdots,p$ and by $\bm_{\mu\nu}$ with $\mu,\nu=p+1,\cdots, d$. It is then easy to check that the standard deformations form precisely the subspace $\mathfrak{h}^\perp$ orthogonal to $\mathfrak{h}$ with respect to the Killing form~\eqref{eq:killing}. This subspace can be identified with the quotient $\mathfrak{so}(d+1,1)/\mathfrak{h}$. In other words, at least infinitesimally, the standard deformations generate $\SO(d+1,1)/H$ which is the space of all spherical defects of dimension $p$.

This implies that the solutions~\eqref{eq:flatsoln} constructed above describe all possible deformations of $\cD_1$ which preserve the property of it being spherical. In other words, the position of a spherical defect $\cD_2$ relative to a spherical defect $\cD_1$ can always be described by $v=\xi|_{\cD_1}$ where $\xi$ is a standard deformation of  $\cD_1$. Furthermore, using the invariance properties of the Killing form on $\mathfrak{so}(d+1,1)$  we can see that the space of standard deformations of a spherical defect $\cD_1$ is always\footnote{I.e.\ not just in the above example of a flat defect.} $\mathfrak{h}^\perp$, where $\mathfrak{h}$ is the subalgebra that stabilises $\cD_1$.

\begin{figure}[t]
	\centering
	\includegraphics[scale=0.8]{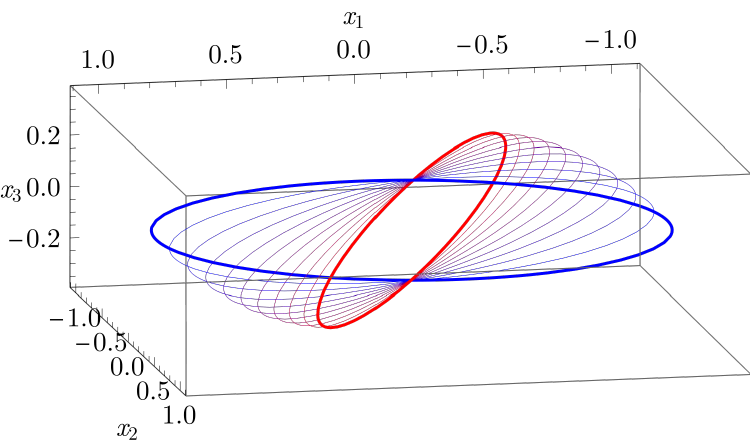}
	\caption{An example configuration of line defects in $d=3$. The thick blue circle is $\cD_1$, the thick red circle is $\cD_2$. The thin lines in between represent different values of $t$ in $X_\text{def}(z,t)$ defined in section~\ref{sec:displacement}, ranging from $0$ to $1$ with step $0.1$.}
	\label{fig:standard_line_3d_config}
\end{figure}

\paragraph{Example: line defects in $d=3$} A pair of line defects in $d=3$ can always be mapped to the configuration where $\cD_1$ is a circle of radius $1$ in the plane 1-2 with its centre at $0$, while $\cD_2$ is smaller circle obtained by applying a dilatation and a rotation in the plane $1-3$, see figure~\ref{fig:standard_line_3d_config}. Specifically, $\cD_2$ is obtained from $\cD_1$ by action with (note that $[\bd,\bm_{13}]=0$)
\be
	e^{\log r \bd-\theta \bm_{13}}=e^{\log r \bd}e^{-\theta \bm_{13}},
\ee
where $r<1$ is the radius of $\cD_2$ and $\theta$ is an angle. The intersections of $\cD_1$ and $\cD_2$ with the $x^2=0$ plane mark four points which kinematically behave similarly to operator insertions in a four-point function.\footnote{The plane $x^2=0$ is the unique plane or sphere which meets both defects at right angles.} 

If we introduce a complex coordinate $x^1+i x^3$ in this plane, then $\cD_1$ intersects it at points $\pm 1$, while $\cD_2$ intersects it at $\pm \r$, where 
\be	
	\r=re^{i\theta}.
\ee
In these kinematics $\r$ is a cross-ratio, and it is completely analogous to the radial $\r$ cross-ratio in four-point functions~\cite{Hogervorst:2013sma}.

The subalgebra $\mathfrak{h}$ that preserves $\cD_1$ is given by
\be
	\mathfrak{h}=\Span\{\bm_{13},\bp_1-\bk_1,\bp_2-\bk_2,\bp_3+\bk_3\}.
\ee
We can easily see from~\eqref{eq:killing} that both $\bd$ and $\bm_{13}$ belong to $\mathfrak{h}^\perp$. Therefore, the CKV
\be
	\xi = \log r \bd-\theta \bm_{13}
\ee
is a standard deformation and we already know that $e^\xi$ maps $\cD_1$ to $\cD_2$. 

Restricting $\xi$ to $\cD_1$ we find
\be
	v=\xi|_{\cD_1} = \log r \hat e_r + \theta\cos \phi \hat e_3,
\ee
where $\hat e_r$ and $\hat e_3$ are unit vectors in the radial and the second directions, respectively, while $\phi$ is an angular coordinate on $\cD_1$ such that $x_1 = \cos \phi$ and $x_2=\sin\phi$. The function $\ell$ takes on the defect the values
\be
	\ell(\phi) = \sqrt{(\log r)^2 + \theta^2 \cos^2\phi}.
\ee

\paragraph{Example: dimension-$p$ defects}  The previous example can be generalised to the situation when $\cD_1$ and $\cD_2$ are $p$-dimensional and the bulk is $d$-dimensional. Specifically, we take $\cD_1$ to be a $p$-dimensional spherical defect of radius $1$ centred at $0$ and lying in the subspace spanned by the coordinate directions $1,\cdots,p+1$. The defect $\cD_2$ is then obtained by the action of
\be\label{eq:general_rotation}
	e^{\log r \bd-\sum_{i} \theta_i\bm_{i,p+i+1}}
\ee
where the sum is over $i=1,\cdots,m=\min\{p+1, d-p-1\}$. Note that by construction, all the generators in the exponential commute with each other. Furthermore, the total number of parameters $r,\theta_1,\cdots, \theta_m$ is $m+1=\min\{d+2-q,q\}$. This is the same as the number of cross-ratios for a pair of $p$-dimensional defects~\cite{Gadde:2016fbj}.

Generalizing the previous example, we have $\mathfrak{h}=\mathfrak{so}(1,p+1)\oplus \mathfrak{so}(q)$. The subalgebra $\mathfrak{so}(p+1)\subseteq\mathfrak{so}(1,p+1)$ is given by rotations in the first $p+1$ directions, while the remaining $p+1$ generators of $\mathfrak{so}(1,p+1)$ are given by $\bp_\mu-\bk_\mu$ with $\mu\in\{1,\cdots,p+1\}$. In other words,
\be
	\mathfrak{so}(1,p+1)=\Span\{\bp_\mu-\bk_\mu|\mu\in \{1,\cdots,p+1\}\}+\Span\{\bm_{\mu\nu}|\mu,\nu\in \{1,\cdots,p+1\}\}.
\ee
Similarly $\mathfrak{so}(q-1)\subset \mathfrak{so}(q)$ is given by rotations in the last $q-1$ directions, while the remaining $q-1$ generators of $\mathfrak{so}(q)$ are given by $\bp_\mu+\bk_\mu$ with $\mu\in\{p+2,\cdots,d\}$, i.e.
\be
	\mathfrak{so}(q)=\Span\{\bp_\mu+\bk_\mu|\mu\in \{p+2,\cdots,d\}\}+\Span\{\bm_{\mu\nu}|\mu,\nu\in \{p+2,\cdots,d\}\}.
\ee
This leaves
\be\label{eq:hperpgeneral}
	\mathfrak{h}^\perp=&\;\Span\{\bd\}+\Span\{\bp_\mu+\bk_\mu|\mu\in \{1,\cdots,p+1\}\}+\Span\{\bp_\mu-\bk_\mu|\mu\in \{p+2,\cdots,d\}\}\nn\\
	&+\Span\{\bm_{\mu\nu}|\mu\in \{1,\cdots,p+1\},\nu\in \{p+2,\cdots,d\}\}.
\ee
In particular, the CKV appearing in~\eqref{eq:general_rotation} is in $\mathfrak{h}^\perp$, and is therefore a standard deformation. So, we can take
\be
	\xi = \log r \bd-\sum_{i} \theta_i\bm_{i,p+i+1},
\ee
and therefore
\be
	v=\xi|_{\cD_1} = \log r\hat e_r + \sum_i \theta_i x^i\hat e_{p+i+1}.
\ee
Here, $x^i$ are the bulk coordinates of a point on $\cD_1$. The function $\ell$ is then given by (on $\cD_1$)
\be\label{eq:length-scale}
	\ell(z) = \sqrt{(\log r)^2 + \textstyle\sum_i \theta_i^2 (x^i)^2}.
\ee

\begin{figure}[t]
	\centering
	\includegraphics[scale=0.8]{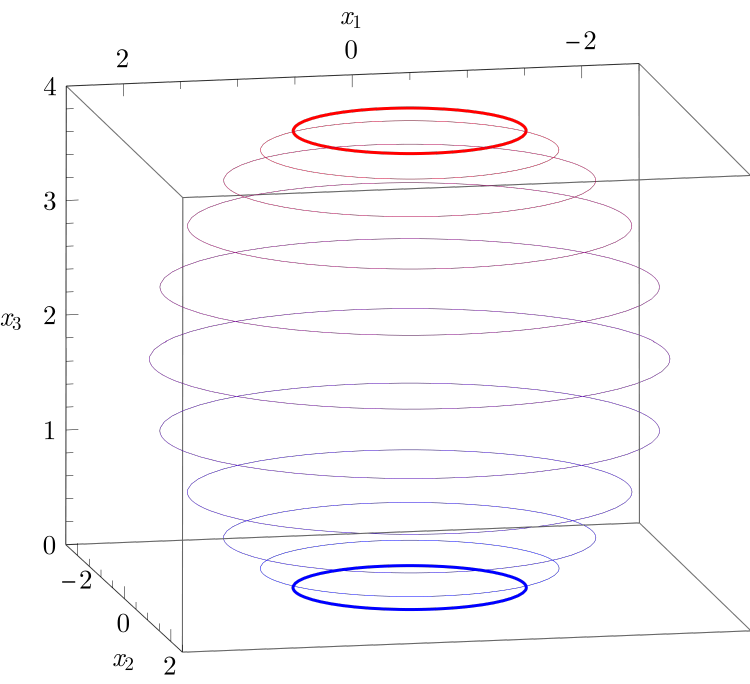}
	\caption{An example configuration of line defects in $d=3$. The thick blue circle is $\cD_1$, the thick red circle is $\cD_2$. The thin lines in between represent different values of $t$ in $X_\text{def}(z,t)$ defined in section~\ref{sec:displacement}, ranging from $0$ to $1$ with step $0.1$. In this picture, $L=4$ is not small to show the path of the deformation more clearly.}
	\label{fig:parallel-p}
\end{figure}

\paragraph{Example: parallel dimension-$p$ spherical defects}
We now consider the case when $\cD_1$ is as above, but $\cD_2$ is obtained from it by a translation by $L$ in the coordinate $x^d$. This is conformally equivalent to a configuration in the previous example with $\theta_i=0$ and
\be
r=\exp\p{-2\,\mathrm{arcsinh}\frac{L}{2}}=\p{\sqrt{1+\tfrac{L^2}{4}}+\tfrac{L}{2}}^{-2}.
\ee
However, it is instructive to consider this configuration directly.

Note that while $\cD_2$ can be obtained by applying $e^{L \bp_d}$ to $\cD_1$, $L\bp_d$ is not a standard deformation. This follows from the characterisation of $\mathfrak{h}^\perp$ in~\eqref{eq:hperpgeneral}. Instead, we have to take
\be
	\xi = \a \tanh \tfrac{\a}{2}\bd + \frac{\a}{2\cosh \frac{\a}{2}}(\bp_d-\bk_d)
\ee
where $\a$ is such that
\be
	L = 2\sinh\tfrac{\a}{2}.
\ee
That this works can be verified by an explicit calculation using embedding space methods. One can also check that as $t$ goes from $0$ to $1$, the transformation $e^{t\xi}$ moves $\cD_1$ into $\cD_2$ along the unique $(p+1)$-dimensional sphere on which both $\cD_1$ and $\cD_2$ lie, see figure~\ref{fig:parallel-p}. 

This choice of $\xi$ gives
\be
v=\xi|_{\cD_1} = \a \tanh \tfrac{\a}{2}\hat e_r + \frac{\a}{\cosh \frac{\a}{2}}\hat e_d
\ee
and 
\be\label{eq:parallelLinesCylinderEll}
	\ell(z) = \a.
\ee

\paragraph{Example: parallel dimension-$p$ flat defects}
Next, we consider the case when $\cD_1$ is a flat defect passing through $x=0$, and spanning the coordinate directions $1,\cdots, p$, but $\cD_2$ is again obtained from it by a translation by $L$ in the coordinate $x^d$.

We represent this translation by setting $\xi=L\bp_d$, which we already know is a standard deformation. This then clearly gives that
\be
	v^\mu=L \, \hat{e}_d
\ee
and 
\be
	\ell(z) = L.\label{eq:ellParallel}
\ee

\paragraph{Example: parallel flat line defects in $\mathbb{R}\x S^{d-1}$}

Finally, we consider the case when $\cD_1$ and $\cD_2$ are flat parallel line defects (so $p=1$) in $\mathbb{R}\x S^{d-1}$, placed along the time direction $\R$. This setup is conformally equivalent to a setup in $\R^d\setminus\{0\}$ via the exponential map $\exp: \mathbb{R}\x S^{d-1}\to \R^{d}\setminus\{0\}$, defined as
\be
	\exp(\tau, \Omega)\equiv e^\tau \Omega\label{eq:expMapCyl}
\ee
for $\tau\in\mathbb{R}$ and $\Omega\in{S}^{d-1}$ (parameterised as a unit vector in $\mathbb{R}^d$). In $\R^d\setminus\{0\}$, the defects $\cD_1$ and $\cD_2$ become rays starting at $0$, and pointing in two different directions $\Omega_1$ and $\Omega_2$, respectively.

Without loss of generality, we can assume that
\begin{align}
	\Omega_1=(1, 0, 0, \dots, 0),\qquad \Omega_2=(\cos\alpha, \sin\alpha, 0, 0, \dots 0),
\end{align}
where $\a$ is the angle between $\Omega_1$ and $\Omega_2$. From the preceding discussion, we know that in the $\R^{d}\setminus\{0\}$ picture, $\xi=\alpha\bm_{21}$ is a standard deformation for $\cD_1$. On the other hand, $e^{\xi}$ is a rotation by angle $\a$ which maps $\Omega_1$ to $\Omega_2$, and therefore $\cD_1$ to $\cD_2$. Thus, in $\R^{d}\setminus\{0\}$,
\be
	v^\mu=\xi^\mu|_{\cD_1}=\a x^1\de^\mu_2.
\ee
Under $\exp^{-1}$, this maps to a time-independent vector of length $\a$, which lies along the spatial directions on $\R\x S^{d-1}$ and points along the great circle from $\cD_1$ to $\cD_2$. Therefore, in the Weyl frame of $\R\x S^{d-1}$, 
\begin{align}
	\ell(z)=\alpha.\label{eq:ellCylinder}
\end{align}

\subsection{Fusion metric}
\label{sec:dilaton}

In the previous subsection we have defined a vector field $v^\mu(z)$ and the corresponding length scale $\ell(z) = \sqrt{v^2(z)}$ which parameterise the relative position of the defects $\cD_1$ and $\cD_2$. An immediate consequence of this is that the combination
\be
	\ell(z)^{-2}g_{\mu\nu}(X_1(z))
\ee
is a Weyl-invariant metric. However, it is only defined on the defect $\cD_1$. If we could extend it away from $\cD_1$ and into the bulk $M$, it would become a Weyl-invariant metric on $M$ and greatly simplify the problem of constructing Weyl-invariant effective actions. In this subsection we show that such an extension is indeed possible.

We first construct a canonical system of coordinates around $\cD_1$, similar to Riemann normal coordinates. For this, we employ a Weyl-invariant equation for curves, analogous to the geodesic equation~\eqref{eq:geodesic}. No such second-order equation exists,\footnote{Equation~\eqref{eq:surface_evolution} is second-order, but contains derivatives of $X_1(z,t)$ with respect to $z$. For reasons which will soon become apparent, here we want to construct a curve for an isolated value of $z$.} but infinitely-many can be found at third order.\footnote{There are infinitely-many in a general curved space~\cite{BaileyCircles}. It is likely that it is unique in a conformally-flat space.} We will use a particular version, as defined in~\cite{BaileyCircles} and called there the conformal circle equation,\footnote{The conformal circle equation is not invariant under general changes of the parameter $t$. The same is true for the geodesic equation~\eqref{eq:geodesic}: its solutions are affinely-parameterised geodesics. Similarly, the solutions to~\eqref{eq:ccequation} are conformal circles with a parameterisation of a special type --- a ``projective'' parameterisation~\cite{BaileyCircles}. In this context, an important property of~\eqref{eq:ccequation} is that it is invariant under projective (fractional-linear) transformations of the parameter $t$.}
\be\label{eq:ccequation}
	\nabla_t a^\mu = \frac{3(u\.a)}{u^2}a^\mu-\frac{3a^2}{2u^2}u^\mu+u^2 u^\nu P^\mu_\nu-2u^\a u^\b u^\mu P_{\a\b}.
\ee
This is a third-order equation for a curve $\g^\mu(t)$, where we used the notation $u^\mu=\nabla_t\g^\mu$, $a^\mu=\nabla_t^2 \g^\mu$, while $P_{\mu\nu}$ is the Schouten tensor defined in~\eqref{eq:schoutendef}.

When $M$ is flat, the geodesics are straight lines, while the solutions to the conformal circle equation are straight lines and circles (see section~\ref{sec:exampleII}). The appearance of the additional parameters --- the radius of the circle and the plane in which it lies --- corresponds to the fact that~\eqref{eq:ccequation} is a third-order equation. In order to solve for $\g$, we need to specify not only the initial position $\g(0)$ and the initial velocity $u(0)$, but also the initial covariant acceleration $a(0)$. The latter transforms inhomogeneously under Weyl transformations as
\begin{align}
	\delta_\w a^\mu&= 2u^\mu u^\nu\ptl_\nu\omega-u^2 \ptl^\mu\omega.
\end{align}
This in particular means that  simply setting $a(0)=0$ breaks Weyl invariance.

In general, there is no canonical way of specifying the initial condition for $a(0)$. This makes it impossible to define canonical coordinates on a general conformal manifold by analogy to Riemann normal coordinates, which in turn makes the classification of general Weyl invariants a famously hard problem~\cite{Fefferman:2007rka}. Fortunately, it turns out that it is possible to choose $a(0)$ canonically in the case at hand, where the conformal circle originates on one the two conformal defects $\cD_1,\cD_2$. For now, we will assume that such a choice has been made; we describe it explicitly at the end of this section.

We can now construct coordinates in a neighbourhood of $\cD_1$ in the following way. Given a point $x\in M$ close to $\cD_1$ we search for $z\in \cD_1$ and a vector $y^\mu$ normal to $\cD_1$ at $z$, such that the conformal circle with the initial condition $\g(0)=X_1(z)$, $u^\mu(0)=y^\mu$ (and the canonically chosen value for $a^\mu(0)$) passes through $x$ at $t=1$, i.e.\ $\g(1)=x$. We will show below that when $x$ is close enough to $\cD_1$, the mapping $x\mapsto (z,y)$ is well-defined, smooth, and one-to-one. We define the coordinates of $x$ to be the pair $(z,y)$, and denote the corresponding conformal circle by $\g_x(t)$.

This construction allows us to extend $\ell(z)$ to a function $\ell(x)$ of Weyl weight 1, defined on a neighbourhood of $\cD_1$. Specifically, we define
\be\label{eq:bulklengthscale}
	\ell^2(x)\equiv \dot\g_x(0)^{-2}\dot\g_x^2(1)\ell^2(z)=y^{-2}\dot\g_x^2(1)\ell^2(z),
\ee
where $(z,y)$ are the coordinates of $x$, defined above. Since $\dot \g_x(t)^2$ is computed using the metric at $\g_x(t)$, this ensures that $\ell(x)$ transforms with Weyl weight $1$ at $x$, given that $\ell(z)$ transforms with Weyl weight $1$ at $X_1(z)=\g_x(0)$. We will show below that $\ell(x)$ thus defined is in fact smooth and restricts to $\ell(z)$ on $\cD_1$. 

Using $\ell(x)$ we can now define the Weyl-invariant metric
\be\label{eq:Weyl_inv_metric}
	\hat g_{\mu\nu}(x) = \ell^{-2}(x) g_{\mu\nu}(x),
\ee
which we refer to as the fusion metric. Any Weyl invariant that can be built out of $g_{\mu\nu}$ and $v^\mu$ can now be written as simply a diffeomorphism invariant built out of $\hat g_{\mu\nu}$ and $v^\mu$. Indeed, since $\hat g_{\mu\nu}$ and $g_{\mu\nu}$ differ by a Weyl transformation, $g_{\mu\nu}$ can be replaced by $\hat g_{\mu\nu}$ in any Weyl invariant. Conversely, any diffeomorphism invariant of $\hat g_{\mu\nu}$ and $v^\mu$  is automatically a Weyl invariant due to the Weyl invariance of $\hat g_{\mu\nu}$ and $v^\mu$.

We now give the details of the construction that were omitted above. To define the initial condition $a^\mu(0)$ we simply observe that the vector
\be\label{eq:qdefn}
	q^\mu = a^\mu(0)+\frac{2}{p}u^\mu(0) u^\s(0) \II_\s -\frac{1}{p}u^2(0) \II^\mu+\half \ptl_a X_1^\mu u^2(0) v^{-2}\ptl^a(v^2),
\ee
where $v^\mu$ parameterises the position of $\cD_2$ relative to $\cD_1$, is a Weyl invariant.  This is essentially the same calculation as verifying Weyl invariance of~\eqref{eq:surface_evolution}. Setting $q^\mu=0$ defines a value for $a^\mu(0)$ in a Weyl-invariant way, providing the necessary initial condition.

We now prove that $\ell(x)$ is smooth. The key subtlety here is that the conformal circle equation~\eqref{eq:ccequation} contains $u^2$ in the denominator. Because of this, it is not entirely obvious whether $\g(1)$ depends smoothly on $u(0)$, and consequently whether $x$ depends smoothly on its coordinates $(z,y)$.

To solve this problem we first rewrite the conformal circle equation in terms of functions $f,h$ appearing in the potentially singular pieces,
\be
	f\equiv \frac{(u\.a)}{u^2},\qquad
	h\equiv \frac{a^2}{u^2}.
\ee
For example, the conformal circle equation becomes
\be
	\nabla_t a^\mu = 3fa^\mu-\tfrac{3}{2}hu^\mu+u^2u^\nu P^\mu_\nu-2P_{\a\b}u^\a u^\b u^\mu.
\ee
It is easy to show (see appendix~\ref{app:CCequation}) that $\ptl_t f$ and $\ptl_t h$ can be expressed in a polynomial way in terms of $u,a,f,h$. Overall, the system of equations for $\g,u,a,f,h$ therefore has smooth functions in the right-hand side and thus the only possible singularities are in the initial conditions. Our initial conditions are given by
$\g(0) = X_1(z)$, $u(0) = y$,  and $q^\mu=0$ (see~\eqref{eq:qdefn}). This gives for $f,h$
\be\label{eq:fhinitial_conditions}
	f(0) = -\frac{1}{p}y\.\II,\quad h(0) = y^2\p{\frac{1}{p^2} \II^2+\frac{1}{4} v^{-4}\ptl_a (v^2)\ptl^a (v^2)}.
\ee
We therefore find that the initial conditions for $\g,u,a,f,h$ depend smoothly on $z$ and $y$, and thus the solution $\g(1)$ depends smoothly on $z$ and $y$ as well. It is easy to verify that the differential of the map $(z,y)\mapsto \g(1)$ is non-degenerate at $y=0$, which shows that it is locally one-to-one and has a smooth inverse.

Finally, we note that 
\be
	f(t) = \thalf \ptl_t \log u^2(t).
\ee
Therefore,
\be
	\dot\g^2(1)\dot\g^{-2}(0)=e^{2\int_0^1 dt f(t)}.
\ee
Since, as discussed above, $f(t)$ depends on $z$ and $y$ in a smooth way, it follows that $\dot\g^2(1)\dot\g^{-2}(0)$ is also a smooth function of $z,y$, and thus of $x$. This shows that~\eqref{eq:bulklengthscale} depends smoothly on $x$. Note, however, that due to the $v^{-4}$ factor in~\eqref{eq:fhinitial_conditions}, $\ell(x)$ is not analytic in $v^\mu$. It is, however, analytic in the derivatives of $v^\mu$, which is what will be important for our applications.

\subsection{Properties of the fusion metric}
\label{sec:fusionmetricprops}

In the previous subsection we described an procedure which constructs the fusion metric $\hat g$ given as initial data the physical metric $g$ in the bulk and the vector field $v$ on $\cD_1$.\footnote{The result only depends on $\sqrt{g_{\mu\nu}v^\mu v^\nu}$, but talking about $v$ is easier since it is Weyl-invariant which saves us from discussing extraneous details.} In this section show that $\hat g$ is not arbitrary and satisfies non-trivial local identities such as~\eqref{eq:II=0} and~\eqref{eq:P=0} regardless of the choice of $g$. These identities will be important when we write down the independent local invariants of $\hat g$ in the following sections. Finally, we make some comments about isometries of $\hat g$.

We will treat $v$ as fixed (together with the bulk manifold $M$ and other data) and view $\hat g$ as a function of $g$. We will express this as $\hat g=\cG_F(g)$, where
\be\label{eq:GFdef}
\cG_F(g)=\ell^{-2}g,
\ee
and $\ell(x)$ depends on $g$ and is defined in section~\ref{sec:dilaton}. Our goal is to determine which metrics are in the image of $\cG_F$.

It is easy to see that $\cG_F$ is not surjective -- there are many metrics which cannot be represented as $\cG_F(g)$ for any $g$. Indeed, by construction $\cG_F$ is Weyl-invariant,
\be
	\cG_F(e^{2\w}g)=\cG_F(g).
\ee
Furthermore, by~\eqref{eq:GFdef} $\cG_F(g)$ is Weyl-equivalent to $g$.
This implies that the image of $\cG_F$ contains precisely one metric from each conformal class. Such metrics are easy to characterise formally; the above identities imply that if $\hat g=\cG_F(g)$ for some $g$, then
\be\label{eq:general_metric_constraint}
	\cG_F(\hat g)=\hat g.
\ee
An important property of $\cG_F(g)$ is that its Taylor coefficients around $\cD_1$ depend locally on $g$. Therefore, the constraint~\eqref{eq:general_metric_constraint} can be rewritten as a set of constraints on the Taylor coefficients of $\hat g$. In the rest of this section we describe a simple way to derive these constraints.

Let $\hat \ell(x)$ be the function $\ell$ computed for a metric $\hat g$ which satisfies~\eqref{eq:general_metric_constraint}. Equation~\eqref{eq:general_metric_constraint} together with~\eqref{eq:GFdef} imply that $\hat \ell(x)=1$. It then follows from the definition~\eqref{eq:bulklengthscale} that the curves $\hat \g_x(t)$ constructed for the fusion metric satisfy 
\be\label{eq:uconst1}
	\dot {\hat \g}_x^2(1)=\dot {\hat \g}_x^2(0).
\ee
It is easy to check that $\hat \g_{\hat \g_x(s)}(t)=\hat \g_x(su)$ from the definition of $\hat \g_x$, and therefore
\be
	\dot {\hat \g}_x^2(st)=s^{-2}\dot{\hat \g}_{\hat \g_x(s)}^2(t).
\ee
Together with~\eqref{eq:uconst1}, this implies
\be\label{eq:uconst2}
	\dot {\hat \g}_x^2(t)=\dot {\hat \g}_x^2(0).
\ee
Conversely, if some metric $\hat g$ satisfies~\eqref{eq:uconst2} for all curves $\hat \g_x$ and $\hat\ell(z)=\sqrt{v\.v}=1$, then~\eqref{eq:bulklengthscale} implies $\hat \ell(x)=1$ and thus $\hat g$ satisfies~\eqref{eq:general_metric_constraint}.

The constraint~\eqref{eq:uconst2} can be equivalently replaced by local constraints, which together with $\hat\ell(z)=\sqrt{v\.v}=1$ read
\be
	\begin{cases}
		\ptl_t^n \dot {\hat \g}_x^2(t)\vert_{t=0}=0,& \forall n> 0\\
		v^2=1.
	\end{cases}
\ee
The above discussion implies that these conditions are equivalent to~\eqref{eq:general_metric_constraint}.
The first few of these constraints are
\be\label{eq:constraints}
	v^2=1, \quad a(0)\. u(0)=0,\quad \hat\nabla_t a(0)\. u(0)+a(0)\.a(0)=0,\quad \cdots,
\ee
where as usual $u=\hat \nabla_t \hat\g_x, \,a=\hat \nabla^2_t \hat\g_x$. In this formulation we can choose $u(0)$ and $\g_x(0)$ at will, since these quantities are the initial conditions for the curves $\hat \g_x$.

The constraint $v^2=\hat g_{\mu\nu}v^\mu v^\nu=1$ is expressed directly in terms of $\hat g$. On the other hand, the subsequent constraints such as $a(0)\. u(0)=0$ are more implicit. Using the initial condition $q^\mu=0$ and~\eqref{eq:qdefn}, we find that $a(0)\. u(0)=0$ is equivalent to
\be
	u^2(0) u^\s(0) \hp_\s =0.
\ee
This has to be satisfied for all $u(0)$ and for all points on $\cD_1$. This implies
\be\label{eq:II=0}
	\hp^\mu=0.
\ee
It then also follows from $q=0$, $v^2=1$ and~\eqref{eq:qdefn} that $a(0)=0$. Note that in the case $p=1$, $\hp^\mu=0$ is equivalent to $\cD_1$ being a geodesic for $\hat g$.

Consider now the condition $\hat\nabla_t a(0)\. u(0)+a(0)\.a(0)=0$. Using $a(0)=0$ and the conformal circle equation~\eqref{eq:ccequation} we find
\be
	0 = u^2(0) u^\mu(0) u^\nu(0) \hat P_{\mu\nu}.
\ee
This has to be satisfied for all $u(0)$, which implies that all normal components of the Schouten tensor $\hat P_{\mu\nu}$ vanish,
\be\label{eq:P=0}
	\hat P_{\mu\nu}=0,\quad \mu,\nu\text{ normal}.
\ee

Note that the number of constraints $v^2=1$,~\eqref{eq:II=0}, and~\eqref{eq:P=0} is precisely the same as the number of the coefficients in the Taylor expansion (in normal directions) of a scalar function up to the second order. This is expected since we are effectively writing out the equation $\hat\ell(x)=1$ order-by-order. Higher-order constraints can be derived in the same way by taking $\nabla_t$ derivatives of the conformal circle equation~\eqref{eq:ccequation}, and using the resulting equation in~\eqref{eq:constraints}. They will involve higher-derivative curvature invariants.

This discussion allows us to give an alternative characterisation of the fusion metric $\hat g$ for a given physical metric $g$: to the quadratic order in the distance from $\cD_1$, the fusion metric $\hat g$ is the unique metric in the conformal class of $g$ which satisfies $\hat g_{\mu\nu}v^\mu v^\nu=0$,~\eqref{eq:II=0}, and~\eqref{eq:P=0}. The scale factor $\ell(x)$ is then defined by $\hat g = \ell^{-2}g$. Existence and uniqueness are easy to check from the Weyl transformation laws of $\hp^\mu$ and $\hat P_{\mu\nu}$. If we classify the analogous constraints at all orders, they will provide an all-order characterisation of $\hat g$.

This characterisation is clearly not unique starting at the second order. Specifically, we could require that the purely normal components of $\hat R_{\mu\nu}$ vanish, or we could use any other purely normal symmetric rank-2 2-derivative tensor. Similar ambiguities exist at higher orders. Of course, one would have to verify that a metric solving such an alternative condition exists. The advantage of the construction in section~\ref{sec:dilaton} is that it is guaranteed to work at all orders. The minor drawback is that the constraints such as~\eqref{eq:P=0} have to be derived order-by-order rather than chosen at will.

\paragraph{Isometries of the fusion metric} Finally, we note that any conformal symmetry of the physical metric $g$ which preserves $\cD_1$ and does not modify\footnote{Note that here we refer to the value of $\ell$ no the defect.} $\ell(z)$ will become an isometry of the fusion metric $\hat g$. This is because $\hat g$ depends only on $\ell(z)$ and the conformal class of $g$. For example, transverse rotations around $\cD_1$ always preserve $\ell(z)$ (regardless of its value), so whenever they are available, they become isometries of $\hat g$. For example, this always applies for spherical defects in flat space.

\subsection{Example: spherical defects in flat space II}
\label{sec:exampleII}

Let us now consider how the construction of section~\ref{sec:dilaton} works in the case of flat space and spherical defects.

First of all, the flat space version of the conformal circle equation~\eqref{eq:ccequation} becomes
\be
		\ptl_t a^\mu = \frac{3(u\.a)}{u^2}a^\mu-\frac{3a^2}{2u^2}u^\mu.
\ee
Let us study its solutions. Using translations, we can assume that $\g^\mu(0)=0$. The remaining initial conditions are given by the vectors $a^\mu(0)$ and $u^\mu(0)$. Using rotations, we can make sure that these vectors lie in the 1-2 plane. Then the whole curve $\g^\mu(t)$ will be contained in this plane, and thus we can fully characterise the solution by the function
\be
	s(t) = \g^1(t)+i\g^2(t).
\ee
It is easy to check that in terms of $s(t)$ the conformal circle equation becomes simply
\be
	\{s(t),t\}\equiv \frac{\dddot{s}(t)}{\dot s(t)}-\frac{3}{2}\p{\frac{\ddot s(t)}{\dot s(t)}}^2=0,
\ee
where $\{s(t),t\}$ is the Schwarzian derivative of $s(t)$. The only functions with vanishing Schwarzian derivative are the fractional-linear transformations, i.e.\ functions of the form
\be
	s(t) = \frac{c_{11}t+c_{12}}{c_{21}t+c_{22}},
\ee
for some $c_{ij}\in\C$. Since $t$ takes values in $\R$, and the image of $\R$ under a fractional-linear map is a circle, it follows that the curve $\g$ follows a circle (possibly of infinite radius, i.e.\ a line). 

To express the solution in terms of initial data, define
\be
	\a = \frac{\ddot s (0)}{\dot s(0)^2},
\ee
and recall $s(0)=\g^1(0)+i\g^2(0)=0$. The solution takes the form
\be\label{eq:ssolution}
	s(t)=\frac{\dot s(0)t}{1-\frac{\a\dot s(0)t}{2}}.
\ee
This explicit form implies that 
\be\label{eq:circleinfinity}
	s(\oo)=-\frac{2}{\a},\quad \dot s(t)=\frac{4}{\a^2\dot s(0) t^2}+O(t^{-3}).
\ee 
These identities will be useful later.

Let us apply the above discussion to the curves $\g_x$ defined in section~\ref{sec:dilaton}. Firstly, the invariant $q$ defined by~\eqref{eq:qdefn} becomes, in flat space and for a flat defect $\cD_1$,
\be
	q^\mu = a^\mu(0)+\frac{1}{2}\ptl_a X_1^\mu u^2(0) v^{-2}\ptl^a v^2.
\ee
Setting $q=0$ then fixes
\be\label{eq:flat_a_initial}
	\frac{a^\mu(0)}{u^2(0)}=-\frac{1}{2}\ptl_a X_1^\mu  v^{-2}\ptl^a v^2=-\frac{1}{2}\ptl_a X_1^\mu  \ell^{-2}\ptl^a \ell^2.
\ee
Note that this in particular says that $a^\mu(0)$ is tangent to the defect. On the other hand, the initial velocity $u^\mu(0)$ is, by construction, orthogonal to the defect. 

Since $a(0)\.u(0)=0$, by choosing the coordinates appropriately, $\dot s(0)$ can be taken to real and positive, while $\ddot s(0)$ can be taken to be imaginary with positive imaginary part, so that $\a$ is imaginary. In this case the real axis is normal to the defect, while the imaginary axis is along the defect. Using~\eqref{eq:circleinfinity} we then see that $s(\oo)$ is imaginary and thus the conformal circle $\g(t)$ meets the flat defect at $t=\oo$. The conformal circle is normal to the defect at that point, since $\dot s(t)$ is real to the leading order according to~\eqref{eq:circleinfinity}.

Furthermore, note that the position of $\g(\oo)$ depends only on $\a$, i.e.\ the vector $a^\mu(0)/u(0)^2$. This vector is fixed by~\eqref{eq:flat_a_initial} and is independent of $u^\mu(0)$. This shows that given $\ell(z)=\sqrt{v^2(z)}$, for every point $z$ on $\cD_1$ there is a corresponding point $r_\ell(z)$ defined by $\g(0)=z$ and $r_\ell(z)=\g(\oo)$, and depending only on $z$ and $\ell$. For different choices of $u^\mu(0)$, the corresponding conformal circles all pass through $z$ at $t=0$ and $r_\ell(z)$ at $t=\oo$, being normal to $\cD_1$ at these intersections. 

While the above statements have been derived for a flat defect, the are conformally-invariant and therefore apply to any spherical defect. To determine the bulk length scale $\ell(x)$ at a given point $x$ in the bulk, one has to first determine on which conformal circle $x$ lies, parameterise the circle appropriately, and then apply~\eqref{eq:bulklengthscale}. These steps can in principle be performed without solving any differential equations by using the above statements.

In practice, it appears to be hard to carry out this procedure for a general defect length scale $\ell(z)$ or even for the $\ell(z)$ which arise from the displacements described in section~\ref{sec:exampleI}. However, some simple observations can be made in full generality. 

Spherical defects in flat space always admit transverse rotations. As discussed in section~\ref{sec:fusionmetricprops}, in the fusion metric transverse rotations become isometries. This places strong constrains on tensors with normal indices -- only their singlet components can be non-zero. For example, for spherical defects we always have
\be\label{eq:IIFull=0}
	\hp_{ab}^\mu = 0,
\ee
which is a strengthening of~\eqref{eq:II=0}. Similarly, condition~\eqref{eq:P=0} gets strengthened to
\be\label{eq:Pmore=0}
	\hat P_{\mu\nu}=0,\quad (\text{if $\mu$ or $\nu$ is normal}).
\ee
Equation~\eqref{eq:IIFull=0} and the Gauss equation~\eqref{eq:gauss} imply that the intrinsic Riemann tensor coincides with the restriction of the bulk Riemann tensor for the fusion metric. These and other constraints serve to reduce the number of independent couplings in the fusion effective action for spherical defects in flat space.

In the rest of this section we will explicitly work out the fusion metric in several examples where $\ell(z)$ is constant.

\paragraph{Example: line defects in $d=3$} 
We first consider the case of circular line defects in $d=3$, in situations when $\ell(z)$ is constant. Without loss of generality, we can assume that $\cD_1$ is a circle of radius $1$ centred around $0$ and lying in the 1-2 plane, $x^3=0$. Since $\ell(z)=\ell$ is constant, the map $z\to r_\ell(z)$ has to respect the Euclidean symmetries of this circle, which include reflections in planes normal to a diameter of $\cD_1$. The only option is
\be
	r_\ell(z) = -z,
\ee
where $-z$ denotes the point on $\cD_1$ diametrically opposite to $z$.

Consider now a bulk point $x_0$ and suppose that it is in the 1-3 plane, $x^2_0=0$. We will be able to recover the general case using rotations around the axis of $\cD_1$. We need to work in a neighbourhood of $\cD_1$, so we assume that $(x^1_0,x^3_0)$ is close to $(1,0)$. 

The conformal circle which passes through $x_0$ has to intersect $\cD_1$ at $x^2=0$ and $x^1=\pm 1$. Using the complex coordinate $s = x^1+ix^3$, the conformal circles which pass through $s=1$ at $t=0$ and through $s=-1$ at $t=\oo$ are given by
\be
	s(t) = \frac{1+\dot s(0)t/2}{1-\dot s(0)t/2}.
\ee
Requiring that $s(1)=s_0=x_0^1+ix_0^3$ we find
\be
	\dot s(0) = 2\frac{s_0-1}{s_0+1},\quad \dot s(1) = \frac{s_0^2-1}{2}.
\ee
This gives
\be
	\ell(x_0) = \frac{|\dot s(1)|}{|\dot s(0)|}\ell=\frac{|s_0+1|^2}{4}\ell.
\ee

If we now return to the case of general $x_0$ and parameterise it using the cylindrical coordinates $(r,\f,x^3)$ with $x^1=r\cos\f$, $x^2=r\sin\f$, we get
\be
	\ell(r,\f,x^3) = \frac{(x^3)^2+(r+1)^2}{4}\ell,
\ee
which is smooth in the neighbourhood of $\cD_1$ given by $r\approx 1, x_3\approx 0$ but has a singularity on the axis $r=0$. 

This metric best interpreted in the toroidal coordinates $\tau,\s,\f$ defined by $\tau\geq 0$, $\s\in[-\pi,\pi]$ and
\be
	r&=\frac{\sinh \tau}{\cosh\tau-\cos\s},\quad
	x^3=\frac{\sin \s}{\cosh\tau-\cos\s}.
\ee
In these coordinates we have
\be
	\ell(\tau,\s,\f)=\frac{e^\tau \ell}{2(\cosh\tau-\cos\s)}.
\ee
The physical flat-space metric is
\be
	ds^2=(\cosh\tau-\cos\s)^{-2}\p{d\tau^2+d\s^2+\sinh^2\tau d\f^2},
\ee
and the fusion metric is
\be
	\hat ds^2=\ell^{-2}(\tau,\s,\f)ds^2=4\ell^{-2}e^{-2\tau}\p{d\tau^2+d\s^2+\sinh^2\tau d\f^2}.
\ee
This metric is invariant under translations in $\s$ and $\f$. Translations in $\s$ are new as compared to the physical metric. They coincide with transverse rotations around $\cD_1$, in agreement with our general expectations. In $r,\f,x^3$ coordinates the only non-zero component of the Schouten tensor is
\be
	\hat P_{\f\f}=\frac{2r}{(1+r^2)+(x^3)^2}.
\ee
This obviously satisfies the condition~\eqref{eq:Pmore=0}.

\paragraph{Example: parallel dimension-$p$ flat defects} This is a continuation of an example in section~\ref{sec:exampleI}. As in section~\ref{sec:exampleI}, we shall place $\cD_1$ so that it passes through $x=0$, and spans the coordinate directions $1, 2, \cdots, p$. By translation and rotational symmetry of the defect, we can see that $r_\ell(z)=\infty$.

If we consider a point $x_0$ in the bulk, then without loss of generality, we can assume that $x_0=c \hat{e}_d$, as any other point can be recovered by transverse rotations and translations along the defect. In terms of $s(t)$, $r_\ell(z)=s(\oo)=\oo$ together with~\eqref{eq:circleinfinity} implies $\a=0$. Equation~\eqref{eq:ssolution} then implies that
\be
	s(t)=\dot{s}(0)t,
\ee
and thus $\g^\mu$ has constant velocity. Equation~\eqref{eq:bulklengthscale} then gives us that $\ell(x_0)=\ell(0)$. By symmetry under transverse rotations, and translations along the defect, this must hold for any $x_0$, and so $\ell(x)$ is a constant everywhere in the bulk.

We saw in \eqref{eq:ellParallel} that if we have a second parallel defect at a distance $L$ away, then $\ell(z)=L$ on the defect, and so we must also have that
\be
	\ell(x)=L
\ee
in the bulk as well. The fusion metric is the flat $d\hat s^2 = L^{-2}ds^2$.

\paragraph{Example: parallel flat line defects in $\R\x S^{d-1}$}
In order to study a pair of parallel flat line defects $\cD_1$ and $\cD_2$ in $\R\x S^{d-1}$, we shall again use the exponential map \eqref{eq:expMapCyl} to map the defects to rays in $\R\setminus\{0\}$ starting at zero and pointing in directions $\Omega_1$ and $\Omega_2$ respectively. We then take a bulk point, which in the $\R\setminus\{0\}$ picture is located at a point $x=e^{\tau_0}\Omega_0$ for some $\tau_0\in\R$, $\Omega_0\in S^{d-1}$.

Without loss of generality, we take
\be
\Omega_1=(1, 0, 0, \dots, 0),\qquad \Omega_0=(\cos\theta, \sin\theta, 0, 0, \dots 0).
\ee

By reflection symmetries $\gamma$ must lie in the $1-2$ plane, and by inversion symmetry under $\tau\to2\tau_0-\tau$ it must lie along the circle going through $x$ and $z=e^{\tau_0}\Omega_1$. The $\SO(d-1)$ rotations which preserve $\Omega_1$ imply that $r_\ell(z)=-z$.\footnote{Though this point is not on the defect, it is on the straight line obtained by extending $\cD_1$ through the origin.}  Letting $s(t)=\gamma^1(t)+i\gamma^2(t)$, $\tl{z}=e^{\tau_0}$, $\tl{x}=e^{\tau_0+i\theta}$ this implies
\be
s(t)=\frac{-\tl{z}t(\tl{z}-\tl{x})+\tl{z}(\tl{z}+\tl{x})}{t(\tl{z}-\tl{x})+(\tl{z}+\tl{x})}.
\ee
We can now use this expression in the conformal frame of $\R\x S^{d-1}$ to find
\be\label{eq:ell(x)forCylinder}
\ell(\tau, \theta) = \a\cos^2\frac{\theta}{2}.
\ee
This result holds in generality where $\alpha$ is the angle on the cylinder between $\cD_1$ and $\cD_2$, and $\theta$ is the angle between $\cD_1$ and $x$.

A useful set of coordinates on $S^{d-1}$ are the stereographic coordinates $u^i$ for $i=2\dots d$
\be
u^i=\frac{2\Omega^i}{1+\Omega^1}.
\ee
Note that the defect $\cD_1$ is at $u=0$. The fusion metric can be written as
\be\label{eq:cuspfusionmetric}
d\hat{s}^{2}&={\alpha ^{-2}}\sum_{i=2}^{d}du_i^2+\a^{-2}{(1+\tfrac{u^2}{4})^2}d\tau^2.
\ee
The only non-zero component of the corresponding Schouten tensor is 
\be
\hat{P}_{\tau\tau}&=-\frac{1}{2}\p{1+\frac{u^2}{4}},
\ee
and the Ricci scalar is given by
\be
\hat{R}&=2(d-1)\hat g^{\tau\tau}\hat{P}_{\tau\tau}=-\frac{(d-1) \alpha ^2}{1+\frac{u^2}{4}}\label{eq:RicciCylinder}.
\ee
Note for future reference that only the even-order derivatives of the fusion metric~\eqref{eq:cuspfusionmetric} are non-zero on the defect at $u=0$.

\section{Effective action}
\label{sec:Seff}

In this section we study the effective action for the fusion of two defects. We use the same conventions as in section~\ref{sec:confgeom}. In this section, $\g$ always refers to the defect metric.

\subsection{General comments}
\label{sec:generalcomments}

Recall that our working assumption is that the fusion of two conformal defects $\cD_1$ and $\cD_2$ can be viewed as an RG flow terminating at an IR conformal defect $\cD_\S$. Explicitly, in terms of correlation functions we expect
\be\label{eq:fusion_section}
	\<\cD_1\cD_2\cdots \>\approx \<\cD_\S[e^{-S_\text{eff}}]\cdots\>,
\ee
where $\cdots$ represent other insertions. Here, the characteristic distance $L$ between the two defects goes to zero, $L\to 0$, while all the other distances remain finite on the order of the IR scale $R$. The effective action $S_\text{eff}$ describes the end tail of the RG flow and has the general form
\be\label{eq:Seff_general_form}
	S_{\text{eff}} = \int d^p z\sqrt{\g} \p{\l_1\mathbf{1}+\sum_{\cO} \l_\cO\cO},
\ee
where $\l_\cO$ are various couplings. Here, we include the identity operator and a sum over (marginally) irrelevant operators primary $\cO$ that exist on $\cD_\S$. Descendant couplings can be omitted because they can be reduced to primary couplings by integration by parts. No (marginally) relevant operators are included because their existence would contradict the assumption that we reach the IR fixed point $\cD_\S$, see section~\ref{sec:transverse}. Finally, we do not include $O(1)$ 0-derivative couplings for the exactly marginal operators. Instead, we use the convention where these are absorbed by choosing for $\cD_\S$ an appropriate point in the conformal manifold. Thus, marginal couplings start at 1-derivative order or $O(L)$.

That the effective action~\eqref{eq:Seff_general_form} has to be an integral of a local Lagrangian density follows from the locality of the defects $\cD_1,\cD_2$, and $\cD_\S$. Correspondingly, the couplings $\l_\cO$ have to be local functionals of the various parameters present in the correlation function. The latter are the usual geometric quantities such as the bulk and defect metrics as well as the extrinsic defect geometry for $\cD_1$ and $\cD_2$ and their relative position.

In practice, it is useful to replace all the geometric information about $\cD_2$ by the displacement vector field $v^\mu(z)$ on $\cD_1$, defined in section~\ref{sec:displacement}, that describes the position of $\cD_2$ relative to $\cD_1$. Furthermore, we can assume that $\cD_\S$ is inserted along the same submanifold as $\cD_1$ --- as we discuss below, any changes to this can be absorbed into the coupling $\l_D$ for the displacement operator. Therefore, the couplings $\l_\cO$ have to be constructed from the bulk geometry, the geometry of $\cD_1$, and the vector field $v^\mu(z)$. 

Dimensional analysis implies that $\l_\cO\sim L^{\De_\cO-p}f(L/R)$, where $R$ is the IR scale. On the other hand, the dependence on $R$ can only appear as the powers of $1/R$ coming from derivatives of the metric and other parameters. Therefore, as is usual, the derivative expansion of the couplings $\l_\cO$ becomes an expansion in powers of $L/R$. If we are only interested in approximating the correlation function in the left-hand side of~\eqref{eq:fusion_section} to order $O(L^a)$, we can truncate the effective action~\eqref{eq:Seff_general_form} to only include terms with $\De_\cO+n-p\leq a$, where $n$ is number of derivatives in $\l_\cO$, and evaluate the conformal perturbation theory to $O(L^a)$. Notice that the only coupling that does not go to zero in the limit $L\to 0$ is the identity operator coupling $\l_1$.

Before proceeding with the more detailed analysis of the effective action, let us comment on the ``normalisation'' of the conformal defects $\cD_1, \cD_2$ and $\cD_\S$. Throughout the paper, we assume that these defects are normalised so that they are both local and Weyl-invariant up to Weyl anomalies. Note that changing the defect normalisation $\cD \to \l\cD$ for $\l\in \C$ is not a local modification of the defect action and thus breaks locality. In other words, the only normalisation changes that we allow are the scheme changes of the form
\be\label{eq:DnormalisationChange}
	\cD\to e^{\int d^pz \sqrt{\g}N}\cD,
\ee
where $N$ is a mass dimension $p$ diffeomorphism-invariant local quantity. In particular, it is impossible in general to set $\<\cD\>=1$, which is a normalisation sometimes used in the literature.
Note that unless $\sqrt{\g}N$ is Weyl-invariant, this change of scheme will change the form of the Weyl anomaly. We do not generally commit to a particular form of the Weyl anomaly.

If the action of $\cD_1$, $\cD_2$, or $\cD_\S$ is modified using~\eqref{eq:DnormalisationChange}, this may change the Wilson coefficients in the effective action $S_\text{eff}$. We do not consider such modifications of the effective action and treat the actions of the defects $\cD_1,\cD_2$, and $\cD_\S$ as fixed.

\paragraph{Conformal boundaries}

A special case of our discussion is when $\cD_1$ and $\cD_2$ are conformal boundaries, and the bulk CFT is sandwiched in between them. In this case, no bulk is left in the $L\to 0$ limit. The IR defect $\cD_\S$ becomes a $(d-1)$-dimensional theory. Generically, one expects $\cD_\S$ to be gapped, in which case only the identity operator appears in $S_\text{eff}$. If $\cD_\S$ is a CFT, our general discussion still applies.

\subsection{Parity and reality conditions}
\label{sec:parity}

Both the bulk CFT and the defects $\cD_1,\cD_2$ can preserve or break space parity. In this section we briefly discuss the possible breaking patterns on general manifolds $M$ and the implications for reality conditions on the effective couplings $\l_\cO$.

If the bulk CFT preserves parity, this means that its partition function on a general manifold $M$ can be computed without specifying an orientation of $M$.\footnote{A parity-preserving theory may have parity-odd local operators, which require specification of an orientation at the point of their insertion.} In particular, it can be studied on non-orientable manifolds. If the bulk CFT breaks parity, we need to specify an orientation of $M$ in order to compute the partition function. In any case, we will assume that the bulk theory satisfies the following Hermiticity axiom \cite{Atiyah:1989vu, Kontsevich:2021dmb},
\be
	\bar{\cZ(M)}=\cZ(\bar M),
\ee
where $\cZ(M)$ denotes the partition function on a manifold $M$ and $\bar M$ denotes $M$ with its orientation reversed. In particular, parity-preserving theories have real partition functions. This identity allows $M$ to have a boundary, in which case $\cZ(M)$ is valued in an appropriate vector space and $\cZ(\bar M)$ in the complex conjugate vector space. This Hermiticity axiom (or a somewhat restricted version thereof) is necessary to be able to formulate reflection positivity.

A conformal defect can have a richer parity breaking structure. If the bulk is parity-preserving, the defect can break reflections along the defect, transverse to the defect, or both. For example, Wilson lines in complex representations break parity along the defect, while monodromy defects generally break (at least the naive) transverse reflections. On a general manifold $M$ this means that in order to compute partition functions involving a defect $\cD$, we may need to specify an orientation on $\cD$, as well as an orientation in the normal bundle of $\cD$.

If the bulk CFT breaks parity, these two orientations are not independent since they can be related by the bulk orientation. By convention, we can agree that, if any orientation needs to be specified for the defect, it is only the defect orientation. Thus, if the bulk CFT breaks parity, the defect still may break or preserve parity along the defect. Preserving defect parity while breaking bulk parity is possible because a defect reflection can be achieved by a bulk rotation. 

In the presence of a defect the Hermiticity axiom needs to be modified as follows,
\be
	\bar{\cZ(M,\cD)}=\cZ(\bar M,\bar \cD),\label{eq:Hermiticity}
\ee
where $\bar\cD$ denotes $\cD$ with the defect orientation reversed, but the normal bundle orientation unchanged.\footnote{To see that the normal orientation has to be the same in $\cD$ and $\bar \cD$, consider the statement of reflection positivity when $\cD$ is normal to the $t=0$ time slice.} We will assume that this axiom is satisfied.

The Hermiticity axiom leads to reality conditions on the couplings $\l_\cO$ in the effective action $S_\text{eff}$. In the simplest scenario, whatever parities are broken by $\cD_\S$, the respective orientations are determined in some way from the orientations of $\cD_1$ and $\cD_2$. This means that if 
\be
	\cD_1\cD_2\approx \cD_\S[e^{-S_\text{eff}}]
\ee
then
\be\label{eq:ccfusion}
	\bar\cD_1\bar\cD_2\approx \bar\cD_\S[e^{-\bar{S_\text{eff}}}],
\ee
where $\bar{S_\text{eff}}$ applies complex conjugation and reverses all orientations\footnote{The bulk and the defect orientations, but not the normal bundle orientation.} that were used to construct the action. Using this in the Hermiticity axiom for $\cZ(M,\cD_1,\cD_2)$ we deduce that 
\be
	S_\text{eff}=\bar{S_\text{eff}}
\ee
on any manifold $M$. This implies that parity-even terms in $S_\text{eff}$ have to be real, while parity-odd terms, if at all allowed, have to be imaginary. Here, parity-even and parity-odd refers to the behaviour under the reversal of bulk and defect orientations, but not the normal bundle orientation.

In principle it might be possible that the orientation of $\cD_\S$ is obtained from some curvature invariants rather than from the orientations of $\cD_i$, in which case~\eqref{eq:ccfusion} might require a more careful treatment.\footnote{In other words, the defects $\cD_1$ and $\cD_2$ might be fully-parity preserving, but the particular fusion configuration might be chiral. E.g.\ already the fusion of two unoriented distinct line defects in $d=3$ breaks the bulk parity by the quantity $v^{[\mu} (\nabla^\perp)^\nu v^{\l]}$.} However, this requires a rather exotic behaviour where the theory $\cD_\S$ is different for flat and curved settings, which can only happen under certain conditions. See section~\ref{sec:transverse} for a discussion of this in the context of transverse symmetry breaking.

\subsection{Weyl invariance}
\label{sec:weyl_inv}

The action~\eqref{eq:Seff_general_form} is constrained by diffeomorphism and Weyl invariance. If we ignore the possible anomalies (we return to them in section~\ref{sec:anomalies}), $S_\text{eff}$ has to be Weyl- and diffeomorphism-invariant.

While it is easy to build diffeomorphism invariants by properly contracting indices, the construction of Weyl invariants is, in general, a much more subtle task (see e.g.~\cite{Fefferman:2007rka}). Fortunately, as we have shown in section~\ref{sec:dilaton}, the geometry of a pair of defects allows us to construct a Weyl-weight 1 scalar function $\ell(x)$ in a neighbourhood of $\cD_1$. On $\cD_1$ this function restricts to a measure of the local distance between $\cD_1$ and $\cD_2$ and in particular $\ell(x)\sim L$. The non-trivial statement of section~\ref{sec:dilaton} is that $\ell(x)$ is also defined in the bulk close to $\cD_1$.

Under a Weyl transformation $g_{\mu\nu}(x)\to e^{2\w(x)}g_{\mu\nu}(x)$ this function transforms as
\be
	\ell(x)\to e^{\w(x)}\ell(x).
\ee
In other words, $\log \ell(x)$ plays the role of a ``dilaton''. It allows us to define a new metric
\be
	\hat g_{\mu\nu}(x)=\ell(x)^{-2}g_{\mu\nu}(x)
\ee
which is invariant under the Weyl transformations of $g$. As explained in section~\ref{sec:dilaton}, any Weyl invariant constructed from $g$ can be written as a diffeomorphism invariant of $\hat g$. We will use hat to denote various geometric quantities in $\hat g$, e.g.\ $\hp^\mu_{ab}$ is the second fundamental form. Note that in this metric $v^\mu$ has unit length,
\be\label{eq:v_is_unit}
	\hat g_{\mu\nu} v^\mu v^\nu = 1.
\ee
We will often raise, lower and contract indices using $\hat g_{\mu\nu}$ instead of $g_{\mu\nu}$. It should always be clear from the context which metric is used.

In order for~\eqref{eq:Seff_general_form} to be Weyl-invariant, the couplings $\l_\cO$ will have to transform non-trivially under Weyl transformations to compensate for the transformations of the operators $\cO$ and the measure $\sqrt{\g}$. It is therefore convenient to introduce new Weyl-invariant couplings using the following procedure. Firstly, we recall that there is no need to include descendants in~\eqref{eq:Seff_general_form}. Secondly, for primary operators we define\footnote{The shifts by the number of indices follow from the relation~\eqref{eq:dimweylrelation} between the scaling dimension and the Weyl weight. The indices on $\cO$ have to be raised and lowered by $g_{\mu\nu}$ while the indices on $\hat\cO$ are raised and lowered by $\hat g_{\mu\nu}$.}
\be
	\hat \cO^{\mu_1\cdots \mu_m}_{\nu_1\cdots \nu_n} \equiv \ell^{\De_\cO-n+m}\cO^{\mu_1\cdots \mu_m}_{\nu_1\cdots \nu_n},
\ee
where the indices can be normal or defect indices. The operators $\hat \cO$ are invariant under Weyl transformations of $g_{\mu\nu}$.\footnote{Anomalous behaviour under Weyl transformations is possible; this is a scheme-dependent question that we discuss in section~\ref{sec:displacment_and_schemes}.}

The effective action~\eqref{eq:Seff_general_form} can now be written as (note that $\hat{\mathbf{1}}=\mathbf{1}$)
\be
	\label{eq:Seff_general_form_Weyl}
	S_\text{eff} = \int d^p z\sqrt{\hat \g} \p{\l_{\hat 1}(z)\mathbf{1}+\sum_{\hat \cO} \l_{\hat \cO}(z)\hat \cO(z)}.
\ee
Weyl-invariance of this action is now equivalent to Weyl-invariance of $\l_{\hat \cO}$. Here we keep the indices implicit; if $\hat \cO$ has spin indices, they have to be contracted with dual indices in $\l_{\hat \cO}$. The couplings $\l_{\hat \cO}$ can be constructed as general diff invariants of $\hat g_{\mu\nu}, v^\mu$, and the shape of~$\cD_1$. The constraints on~$\hat g$ discussed in section~\ref{sec:fusionmetricprops} have to be taken into account.

\subsection{Cosmological constant term, defect spectrum, and unitarity}
\label{sec:cosmo}
The leading allowed term in the effective action is the no-derivatives part of $\l_{\hat 1}$, i.e.\ the leading coupling to the identity operator. The only possible term is
\be
	S_{\text{eff}}^{(0)}= -a_0\int\,d^p z\,\sqrt{\hga},
	\label{cosmo-const}
\ee
where $a_0\in \R$ is a Wilson coefficient, and the sign is introduced for future convenience. Here and in what follows, we often do not write out the identity operator $\mathbf{1}$ explicitly. In terms of the original metric this becomes
\be\label{cosmo-const2}
	S_{\text{eff}}^{(0)}= -a_0\int\,d^p z\,\sqrt{\g}\ell^{-p}\sim L^{-p}.
\ee

This term has a simple interpretation in terms of the Hilbert space. For example, assume that $p=1$ and consider the theory on the cylinder $M=\R\x S^{d-1}$, where $\cD_1$ and $\cD_2$ are running along $\R$, at fixed positions on $S^{d-1}$. Let $H$ be the Hamiltonian generating the time translations along $\R$. We can then ask how the spectrum of $H$ behaves as the distance between $\cD_1$ and $\cD_2$ goes to $0$.

Translated to this language, the relation~\eqref{eq:fusion_section} states that the spectrum of $H$ can be approximated by
\be\label{eq:Happrox}
	H\overset{\text{spectrum}}{\approx}-a_0 \ell^{-1}\mathbf{1}+O(\ell^0)\mathbf{1}+H_\S+\text{small},
\ee
where the ``small'' terms contain irrelevant (or marginal with $O(L)$ couplings) operators and tend to 0 as $L\to 0$, while $H_\S$ is the Hamiltonian for the setup in which $\cD_1$ and $\cD_2$ are replaced by $\cD_\S$. The right-hand side above acts on the Hilbert space of $\cD_\S$.
This shows that $-a_0\ell^{-1}$ gives the leading contribution to the ground state energy in the limit where $\cD_1$ and $\cD_2$ get close to each other. If we treat $\cD_1$ and $\cD_2$ as representing external particles added into the system, $-a_0\ell^{-1}$ becomes the leading term in the effective interaction potential between these particles, as mediated by the CFT. This picture obviously generalises to $p>1$ in which case external point particles are replaced by some extended objects. In $p=1$ case the ground state energy $-a_0\ell^{-1}+O(1)$ of $H$ is related to cusp anomalous dimension in the limit of deflection angle $\f\to \pi$,
\be
	\G_\text{cusp}(\f)=-\frac{a_0}{\pi-\f}+O(1).
\ee
We explore this connection further in section~\ref{sec:Wilson}.

\begin{figure}[t]
	\centering
	\scalebox{0.8}{\begin{tikzpicture}
	\begin{pgfonlayer}{nodelayer}
		\node [style=none] (8) at (-4, 1) {};
		\node [style=none] (9) at (5.25, 1) {};
		\node [style=none] (10) at (-5.25, -1) {};
		\node [style=none] (11) at (4, -1) {};
		\node [style=none] (24) at (0, 1.75) {};
		\node [style=none] (25) at (0, 3.25) {};
		\node [style=none] (32) at (4.075, 1.5) {$\frac{L}2$};
		\node [style=none] (33) at (4.075, -1.5) {$\frac{L}2$};
		\node [style=none] (34) at (5.75, 0.5) {$x^1=0$};
		\node [style=none] (36) at (-4.25, 2.5) {$\mathcal{D}_2$};
		\node [style=none] (37) at (-4.25, -2.5) {$\mathcal{D}_1$};
		\node [style=none] (38) at (0, -0.75) {};
		\node [style=none] (39) at (0, 0.75) {};
		\node [style=none] (40) at (0, -3.25) {};
		\node [style=none] (41) at (0, -1.75) {};
		\node [style=none] (42) at (3.75, 2.25) {};
		\node [style=none] (43) at (3.75, 0.25) {};
		\node [style=none] (45) at (3.75, -2.25) {};
		\node [style=none] (47) at (0, 3.75) {};
		\node [style=none] (48) at (0, -3.75) {};
		\node [style=none] (49) at (3.75, -0.25) {};
	\end{pgfonlayer}
	\begin{pgfonlayer}{edgelayer}
		\draw [style=new edge style 8, line width=1.2pt, in=-180, out=180, looseness=8.50] (24.center) to (25.center);
		\draw [style=new edge style 7, line width=1.2pt, in=0, out=0, looseness=8.50] (24.center) to (25.center);
		\draw [style=new edge style 13, line width=1.2pt, bend left=90, looseness=8.50] (38.center) to (39.center);
		\draw [style=new edge style 13, line width=1.2pt, in=0, out=0, looseness=8.50] (38.center) to (39.center);
		\draw [style=new edge style 14, line width=1.2pt, in=180, out=-180, looseness=8.50] (41.center) to (40.center);
		\draw [style=new edge style 14, line width=1.2pt, in=0, out=0, looseness=8.50] (40.center) to (41.center);
		\draw (8.center) to (9.center);
		\draw (8.center) to (10.center);
		\draw (10.center) to (11.center);
		\draw (11.center) to (9.center);
		\draw [style=new edge style 4] (49.center) to (45.center);
		\draw [style=new edge style 4] (43.center) to (42.center);
	\end{pgfonlayer}
\end{tikzpicture}}
	\caption{The reflection positive set-up for theorem~\ref{thm:oppositesattract}.}
	\label{fig:ref_pos}
\end{figure}
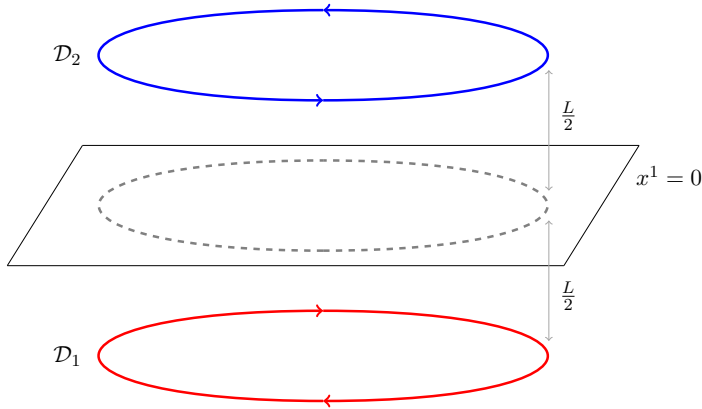

It is easy to show that in a unitary CFT the effective potential is attractive if $\cD_1\simeq \bar{\cD_2}$ (see section~\ref{sec:parity} for the meaning of the conjugation).
\begin{theorem}\label{thm:oppositesattract}
	If $\cD_2$ is conjugate to $\cD_1$ then $a_0\geq 0$.
\end{theorem}
\begin{proof}
	Consider the configuration where $\cD_1$ is a $p$-sphere of radius 1 lying in the hyperplane $x^1=-L/2$, and $\cD_2$ is obtained from it by applying the Hermitian conjugation in the quantisation where Euclidean time $t=x^1$. Then $\cD_2$ is at $x^1=+L/2$, see figure~\ref{fig:ref_pos}. Furthermore,
	\be
		f(L)\equiv \<\cD_2\cD_1\> = \<\Psi(L)|\Psi(L)\>\geq 0,
	\ee
	where $|\Psi(L)\>$ is the state on $x^1=0$ hyperplane created by the path integral over the half-space containing $\cD_1$. We claim that $f(L)$ is a monotonically decreasing function of $L$. Indeed, 
	\be
		f(L+\de L)=\<\Psi(L)|e^{-\de L H}|\Psi(L)\>,
	\ee
	where $H$ is the Hamiltonian appropriate for our quantisation. Since $H\geq0$, it follows that $\|e^{-\de L H}\|\leq 1$ for $\de L>0$ and thus 
	\be
		f(L+\de L)=\<\Psi(L)|e^{-\de L H}|\Psi(L)\>\leq \<\Psi(L)|\Psi(L)\> =f(L).
	\ee
	On the other hand, for small $L$ we have
	\be
		\log f(L) = a_0 \vol S^{p} L^{-p}+\cdots
	\ee
	which can only be decreasing for $a_0\geq 0$.

This proof as well as the proof of theorem~\ref{thm:cauchyschwarz} below relies on the assumption that  $\<\cD_\S\>\neq 0$ for a spherical $\cD_\S$ in flat space. If $\<\cD_\S\>=0$ for a unitary defect, then it is easy to show that all correlation functions of a spherical $\cD_\S$ in flat space with arbitrary insertions on and off $\cD_\S$ have to vanish.
\end{proof}

As an example, consider the free Maxwell theory in 4d, with $\cD_1$ and $\cD_2$ being Wilson lines with charges $q_1$ and $q_2$. Reversing the orientation of a Wilson line is equivalent to changing the sign of the charge, so we partially fix this ambiguity by requiring that the orientations of $\cD_1$ and $\cD_2$ are such that they agree in the limit $L\to0$. Theorem~\ref{thm:oppositesattract} then says that if $q_1=-q_2$ then the potential between these Wilson lines is attractive. In fact, we have in this case
\be
	a_0 = -q_1q_2=q_1^2>0.
\ee

Let us denote the defect conjugate to $\cD_i$ by $\cD_{\bar i}=\bar{\cD_i}$, and the constant $a_0$ appearing in the fusion of $\cD_i$ with $\cD_j$ by $a_0^{ij}$. In the Maxwell theory example we then have $q_{\bar i}=-q_i$ and $a_0^{ij}=-q_iq_j$. In this case it follows from the arithmetic-geometric mean inequality that
\be
	a_0^{ij}=-q_iq_j\leq \frac{q_i^2+q_j^2}{2}= \frac{a^{i\bar i}_0+a^{j\bar j}_0}{2}.
\ee
In other words, the charges $i$ and $j$ cannot attract more strongly than the average attraction in pairs $i\bar i$ and $j\bar j$. This result holds more generally:
\begin{theorem}\label{thm:cauchyschwarz}
For any pair of defects $\cD_i$ and $\cD_j$ the following inequality holds,
\be
	a^{ij}_0\leq \frac{a^{i\bar i}_0+a^{j\bar j}_0}{2}.
\ee
Note that the right-hand side is always non-negative by theorem~\ref{thm:oppositesattract}.
\end{theorem}
\begin{proof}
This follows similarly to theorem~\ref{thm:oppositesattract} from the $L\to 0$ limit of the Cauchy-Schwarz inequality
\be
|\<\Psi_{\bar i}(L)|\Psi_j(L)\>|^2\leq |\<\Psi_{\bar i}(L)|\Psi_{\bar i}(L)\>| |\<\Psi_j(L)|\Psi_j(L)\>|,
\ee
where $|\Psi_i(L)\>$ is created by $\cD_i$.
\end{proof}

A common scenario is when $\cD_1$ and $\cD_2$ (or $\bar\cD_2$) belong to the same sufficiently symmetric conformal manifold. In this case $a_0$ becomes a function $a_0(s)$ of the distance $s$ between $\cD_1$ and ${\cD_2}$ on this conformal manifold. This is the case for magnetic line defects in $O(N)$ models and for Maldacena-Wilson lines in $\cN=4$ SYM, where in both cases $s$ becomes the angle between the scalar couplings. Theorem~\ref{thm:cauchyschwarz} then implies that $a_0(0)$ should be the absolute maximum of the function $a_0(s)$. This is indeed satisfied by explicit results for the Wilson lines at weak and strong coupling~\cite{Maldacena:1998im, Drukker:1999zq}, and for $O(N)$ models in $\e$-expansion~\cite{Diatlyk:2024zkk}.

\subsection{$\cD_\S$ is generically simple}
\label{sec:simplicity}

One objection that can be raised against~\eqref{eq:fusion_section} is that the right-hand side involves a single defect $\cD_\S$ rather than a sum over contributions from many defects. After all, this is what happens in the local operator product expansion and also in the fusion of topological defects.

The resolution to this is twofold. Firstly, sums of defects of the form
\be
	\a_1\cD_{\S_1}+\a_2\cD_{\S_2}+\cdots
\ee
with arbitrary coefficients $\a_i$ are not possible, since considering $\a_1\cD_{\S_1}$ amounts to adding $-\log \a_1$ to the action of $\cD_{S_1}$. But $-\log \a_1$ is not an integral of a local quantity over the defect and thus will break locality of the defect in the sense of the usual cutting and gluing rules for the path integral. It is however possible to consider direct sums of defects 
\be
	\cD_{\S_1}+\cD_{\S_2}+\cdots.
\ee
This can formally be expressed by saying that we allow $\cD_\S$ to be non-simple. Fusion of topological defects generally leads non-simple defects which decompose into direct sums of simple defects.

Secondly, allowing non-simple $\cD_\S$ generically requires going beyond perturbative effective field theory description in the sense that we describe below. If we assume
\be
	\cD_{\S}=\sum_{i}\cD_{\S_i},
\ee
where $\cD_{\S_i}$ are simple, then the effective action~\eqref{eq:Seff_general_form_Weyl} can include local operators from any of the $\cD_{\S_i}$, including a separate identity operator for each of them.\footnote{There is a subtlety in the case of line defects which we discuss below.} In particular, this means that each $\cD_{\S_i}$ comes with its own cosmological constant term $a_{0,i}$, and we have no general reason to expect any coincidences between the values of the $a_{0,i}$. In other words, generically we expect that $a_{0,i}$ are all distinct numbers. As will become clear momentarily, even if the number of defects is infinite, we still expect $a_{0,i}$ to be bounded from above.  Without loss of generality, we can then assume that the defects are labelled so that $a_{0,1}>a_{0,2}>\cdots$.

In this situation, the defect $\cD_{\S_1}$ will give the leading contribution to any partition function, and the defects $\cD_{\S_i}$ with $i>0$ will be suppressed by exponential amounts $e^{-c(a_{0,1}-a_{0,i})L^{-p}}$ for some $c>0$. Note that existence of an upper bound on $a_{0,i}$ is therefore required to have a meaningful leading term. The conformal perturbation theory on $\cD_{\S_1}$ is at any order only power-law suppressed in $L$ relative to the leading term in $\cD_{\S_1}$, and therefore the contributions of $\cD_{\S_i}$ with $i\geq 2$ have to be viewed as non-perturbative effects.

We conclude that if we stay within the framework of perturbative effective field theory, we can generically assume that $\cD_\S$ is a simple defect. This can only be violated if there is some mechanism which makes the leading couplings\footnote{The $n$-derivative couplings to identity operator with $n<p$, scaling as negative powers of $L$.} coincide for two defects. The latter happens, for example, for topological defects, where their topological nature sets all these couplings to zero. 

Another example is the fusion of conformal defects in 2d Ising model, studied in~\cite{Bachas:2013ora}. These conformal defects are denoted $\cD_{\cO,\L}$, labelled by a Virasoro primary $\cO\in \{1,\e,\s\}$ and a matrix $\L\in \mathrm{O}(1,1)/\Z_2$. For $\L=1$ the defects are topological, with $\cD_{1,1}$ being the trivial defect, $\cD_{\e,1}$ the $\Z_2$ symmetry defect, and $\cD_{\s,1}$ the Kramers-Wannier duality defect. For $\L\neq 1$, the conformal defects can be viewed as deformations of these topological defects. In~\cite{Bachas:2013ora} it was shown that the fusion rules for the defects are
\be
	\cD_{\cO_1,\L_1}\x\cD_{\cO_2,\L_2} \to \cD_{\cO_1\x \cO_2,\L_1\L_2},
\ee
where $1\x \cO=\cO,\s\x\e =\s,\e\x\e = 1, \s\x\s = 1+\e$ and the right-hand side is understood as the IR defect $\cD_\S$ that is obtained from the fusion in the left-hand side. In particular,
\be\label{eq:counterexample}
	\cD_{\s,\L_1}\x\cD_{\s,\L_2}\to \cD_{1,\L_1\L_2}+\cD_{\e,\L_1\L_2}.
\ee
Therefore, in this case the fusion of conformal defects yields a non-simple defect. The equality between the coefficients $a_0$ for $\cD_{1,\L_1\L_2}$ and $\cD_{\e,\L_1\L_2}$ can be explained by considering the fusion with the topological $\Z_2$ symmetry defect $\cD_{\e,1}$. It leaves the left-hand side in~\eqref{eq:counterexample} invariant, while exchanging the defects $\cD_{1,\L_1\L_2}$ and $\cD_{\e,\L_1\L_2}$  in the right-hand side.

An interesting subtlety appears in the case of line defects. The local operators on a direct sum defect also include off-diagonal operators that come from ``defect-changing'' local operators that connect $\cD_{\S_i}$ with $\cD_{\S_j}$ (see~\cite{Nagar:2024mjz} for a recent explicit appearance). For concreteness, consider the case when
\be
	\cD_\S=\cD_{\S_1}+\cD_{\S_2},
\ee
and let $\de a_0=a_{0,1}-a_{0,2}\geq 0$. Suppose that we have an operator $\cO_{12}$ that appears in the junction $\cD_{\S_1}\cD_{\S_2}$. Its conjugate $\bar \cO_{12}$ then appears in the junction $\cD_{\S_2}\cD_{\S_1}$. In conformal perturbation theory we will get integrated correlators of products of $\cO_{12}$ and $\bar\cO_{12}$ inserted on $\cD_{\S}$ in an alternating pattern. In between $\bar\cO_{12}$ and $\cO_{12}$ we will have fragments of $\cD_{\S_1}$ and in between $\cO_{12}$ and $\bar\cO_{12}$ we will have fragments of $\cD_{\S_2}$.

A fragment of $\cD_{\S_2}$ of length $s$ will add relative suppression by a factor $e^{-\de a_0 s/L}$, and thus the typical contributions will have fragments of $\cD_2$ of length $L/\de a_0$. In a generic situation $\de a_0=O(1)$ and this is a UV effect. The short fragments of $\cD_{\S_2}$ can be replaced using OPE by local operators on $\cD_{\S_1}$, and the effective theory at the IR scale $R$ can be expressed solely in terms of operators on $\cD_{\S_1}$.

Suppose now that there is a tunable parameter $\l$ and $\de a_0=c_0\l+O(\l^2)$ with $c=O(1)$. Then the typical size of $\cD_{\S_2}$ fragments is $L/\l$ and for $\l\ll 1$ the effective theory in terms of $\cD_{\S_1}$ is valid only at scales $R\gg L/\l$. On the other hand, the effective theory in terms of $\cD_{\S_1}+\cD_{\S_2}$ is valid for $R\gg L$. In connection to the cusp anomalous dimension in $\R\x S^{d-1}$ picture (see section~\ref{sec:GcuspNP}), $L=\a$ and $R=1$ is the radius of the spatial sphere. Thus for $\l\ll 1$ the ``simple'' $\cD_{\S_1}$ effective theory should be used to describe $\a \ll \l $, while if $\a \ll 1$ only then $\cD_{\S_1}+\cD_{\S_2}$ should be used.

This discussion applies, for example, to Wilson lines in gauge theories. In this case $\l$ is the gauge coupling, $\cD_1$ and $\cD_2$ are the fundamental and anti-fundamental lines, while $\cD_{\S_1}$ and $\cD_{\S_2}$ are the trivial line and the adjoint line, see e.g.~\cite{Pineda:2007kz,Correa:2012nk} in the context of $\cN=4$ SYM.\footnote{Confusingly from our point of view, these references call $L$ the soft scale, and $L/\l$ the ultra-soft.} The discussion above appears to be a more abstract version of the pNRQCD effective theory discussed in~\cite{Pineda:2007kz}, although we have not tried to make this connection precise.

If $\cD_1$ and $\cD_2$ are higher-dimensional defects, the contribution of $\cD_2$ is suppressed by $e^{-\de a_0 R^p/L^{p}}$, and so if $\de a_0 = c_0\l+O(\l)$ we cannot ignore $\cD_2$ unless $R\gg L/\l^{1/p}$. However, in this case there is no non-trivial mixing of the two defects.

An important caveat to the above discussion is that we do not know if keeping $\cD_2$ makes sense for $\de a_0 \neq 0$. In general there could be other non-perturbative in $L$ effects that are equally or more important, effectively destroying $\cD_2$ as a well-defined defect.

\subsection{Subleading terms in the identity operator coupling}
\label{sec:subleading}

We now consider the subleading terms in the coupling $\l_{\hat 1}$ of the identity operator. In general, their form depends on the dimension $p$ of the defects and on the dimension $d$ of the bulk CFT. Here we will consider only the one-derivative terms. The two-derivative terms for line defects and for codimension $q=1$ and $q=2$ defects are classified in appendix~\ref{app:two-derivatives}. 

First of all, it is easy to see that for $p>1$ no one-derivative terms exist. Indeed, the only geometric one-derivative tensor at our disposal is $\hp^\mu_{ab}$ and we can also form the covariant derivative $\hat \nabla_a^\perp v^\mu$.\footnote{Note that we do not count the derivative in $\ptl_a X_1^\mu$. This is because $X_1^\mu$ can only appear in this combination, and $X_1$ itself carries a factor of the IR length scale $R$.}  However, there is nothing to contract their defect indices with --- recall that $\hp^\mu_{ab}$ is traceless as required by~\eqref{eq:II=0}.

On line defects, however, it is possible to introduce the unit tangent vector field $t^a$ on $\cD_1$. It defines an orientation of $\cD_\S$; if enough parity is preserved then the effective action has to be an even function of $t$. If $d>3$, this vector doesn't change the situation since $\hp^\mu_{ab}t^at^b=\hp^\mu=0$ and $v\. \nhb_t v=0$. Therefore,
\be\label{eq:SeffLineGenD}
S_\text{eff} = 	\int\,d^pz\,\sqrt{\hga}\p{-a_0+\text{2-derivative terms}}+\cdots\quad \p{p>1 \atop p=1, d>3},
\ee
where $\cdots$ represents non-identity operators. For $p=1$ all the omitted terms in~\eqref{eq:SeffLineGenD} vanish in the limit $L\to 0$.

If $d=3$ then  we can define a unit vector $u^\mu$ that is orthogonal to both $v^\mu$ and $t^\mu$, i.e.\ $u^\mu = \hat \ve^\mu{}_{\r\s}v^\mu t^\s$. This allows us to construct a non-trivial term at one-derivative order,
\be\label{eq:line3doneder}
	S_\text{eff} = 	\int\,dz\,\sqrt{\hga}\p{-a_0+ia_{1}u\.\hat\nabla^\perp_t v+\text{2-derivative terms}}+\cdots\quad (p=1,\,d=3),
\ee
where $a_1\in\R$ (see sec~\ref{sec:parity}). This term breaks parity and requires either the bulk orientation or both the defect and the normal bundle orientations.

\paragraph{Example: pure Chern-Simons theory}
The term multiplying $2\pi i\, a_{1}$,
\be
		I(v)&=\frac{1}{2\pi}\int\,dz\,\sqrt{\hga} u\.\hat\nabla^\perp_t v=\frac{1}{2\pi}\int\,dz\,\sqrt{\hga} \hat \e_{\mu\nu\l} v^\nu t^\l \hat\nabla^\perp_t v^\mu\nn\\
		&=\frac{1}{2\pi}\int\,dz\,\sqrt{\g}  \e_{\mu\nu\l} v^\nu \tau^\l \nabla^\perp_\tau v^\mu,
\ee
where $\tau$ is a tangent unit vector  for the physical metric, contains information about the topological linking of $\cD_1$ and $\cD_2$.\footnote{Note that $I(v)$ is defined on any orientable manifold (provided $\cD_1$ and $\cD_2$ are close enough to each other), while the linking number is only defined if $\cD_1$ and $\cD_2$ are null homologous.} However,  $I(v)$ itself isn't topologically-invariant. A useful description of it is the following. Fix a normal vector field $v_0$ and the corresponding $u_0^\mu = \hat \e^\mu{}_{\nu\l}v_0^\nu t^\l$. Any unit normal vector field $v$ can be written as
\be
	v= \cos \theta v_0 + \sin \theta u_0,
\ee
for some function $\theta(z)$. In terms of $\theta$ we have
\be
	I(v) = \frac{1}{2\pi}\int dz \frac{d\theta}{dz}+I(v_0)\in \Z+I(v_0).
\ee
Therefore, $I(v)$ is a homotopy invariant of $v$ and $I(v)-I(v_0)$ is an integer which describes the difference between the linking numbers of $v$ and $v_0$ with $\cD_1$. However, $I(v)\!\mod 1$ is $v$-independent and, as can be seen by setting $v=\II/|\II|$, is proportional to the total torsion of the curve in the physical metric, which is not a topological invariant and can be equal to any real number even in flat space.

Despite $I(v)$ not being exactly topological, it is relevant for fusion of Wilson lines in 3d pure Chern-Simons theory, which we now consider. Let $\cD_1$ and $\cD_2$ be Wilson loops with framing represented by normal vector fields $f_1$ and $f_2$. To keep track of the framing, we will write $\cD_i[f_i]$. Recall that the framing anomaly implies that $\cD_i[f_i]$ is a homotopy invariant of $f_i$, and whenever $\tl f_i$ differs from $f_i$ by $n$ $2\pi$-twists,
\be
	\cD_i[\tl f_i] = \cD_i[f_i]e^{2\pi i h_i n},
\ee
where $h_i$ is the conformal weight of the corresponding 2d primary~\cite{Witten:1988hf}.

Due to the topological nature of Chern-Simons theory, the only coupling that has a chance of surviving in $S_\text{eff}$ is $a_{1}$ (0- or higher-derivative couplings explicitly depends on the scale $L$). However, two modifications to $S_\text{eff}$ are in order. Firstly, we now can use $f_i$ in addition to $v$. Secondly, since $a_0$ vanishes, the defect $\cD_\S$ may not be simple and is in general a direct sum of Wilson lines $\cD_{\S_k}$ (see section~\ref{sec:simplicity}). Correspondingly, there is a separate identity operator for each $\cD_{\S_k}$, so that the effective action takes the form
\be
	S_\text{eff} = \sum_k \p{2\pi i a_{1,k}I(v)+ib_{1,k}I(f_1)+ib_{2,k}I(f_2)+ib_{\S,k}I(f_{\S,k})}\mathbf{1}_k,
\ee 
where $f_{\S,k}$ is the framing vector field for $\cD_{\S_k}$ and $a,b$ are undetermined coefficients. 

Matching of the framing anomaly requires setting $b_{1,k}=2\pi h_1$, $b_{2,k}=2\pi h_2$ and $b_{\S,k}=-2\pi  h_{\S,k}$. Indeed,
\be
	e^{2\pi i h_1 I(f_1)+2\pi i h_2 I(f_2)-2\pi i h_{\S,k} I(f_{\S,k})}\cD_{\S,k}[f_{\S,k}]
\ee
reproduces the behaviour of $\cD_1[f_1]\cD_2[f_2]$ under the change of framing: it transforms as required under changes of $f_1$ and $f_2$ and doesn't depend on $f_{\S,k}$. While matching the framing anomaly, these couplings introduce dependence on the metric through $I(\.)$.\footnote{Essentially the same term cancels the framing anomaly in the non-topological regularisation of~\cite{Polyakov:1988md}.} This dependence can only be cancelled by the $I(v)$ term, which requires
\be
	a_{1,k}= h_{\S,k}-h_1-h_2.
\ee

Therefore, the effective action for fusion of Wilson lines is
\be
	S_\text{eff} &= 2\pi i\sum_k \p{h_1 (I(f_1)-I(v))+h_2 (I(f_2)-I(v))-h_{\S,k} (I(f_{\S,k})-I(v))}\mathbf{1}_k,
\ee
which is indeed topological. The $v$-dependence predicted by this action is easy to derive directly from Chern-Simons theory by considering the state created by $\cD_1[f_1]\cD_2[f_2]$ on the surface of a thin solid cylinder that contains them, and by studying how this picture transforms under Dehn twists.

\subsection{Displacement operator and other contributions}
\label{sec:displacment_and_schemes}

So far we have focused mostly on the couplings to the identity operator. There is nothing qualitatively new when we consider couplings to other operators, except for two comments. Firstly, the contribution of any other operator to the effective action vanishes in the limit $L\to 0$. This is immediate in the case of the irrelevant operators, and for marginal operators this can be achieved by adjusting the definition of $\cD_\S$ as described in section~\ref{sec:generalcomments}.

Secondly, the couplings to non-trivial operators enter non-trivially into the conformal perturbation theory. The conformal perturbation theory generally requires renormalisation, and Weyl-invariance can be affected by the choice of the renormalisation scheme. Specifically, the behaviour of the couplings $\l_{\cO}$ in~\eqref{eq:Seff_general_form} under Weyl transformations may be anomalous. A standard example of this is the Weyl anomaly which can be viewed as an anomalous transformation law for the coupling $\l_{1}$ to the identity operator.

If present, anomalous terms in these transformation laws must be polynomial in the derivatives of couplings and respect dimensional analysis. For example, suppose the chosen renormalisation scheme is such that under an infinitesimal Weyl transformation $g_{\mu\nu}\to (1+2\w+O(\w^2))g_{\mu\nu}$, the coupling $\l_{\cO_1}$ has to transform as, schematically,
\be
	\de_\w \l_{\cO_1}=N\w\l_{\cO_1}+\a \ptl^m\w\prod _i\l_{\cO_i}^{n_i}
\ee
in order for the partition function to remain Weyl-invariant. Here $N$ is the Weyl weight of the coupling, $n_i,m\geq 0$ are integers\footnote{It must be that $\sum_i n_i>1$ since such anomalies come from divergences which cannot arise at linear level in the couplings.}, $\a\neq 0$ is a dimensionless coefficient, and $\cO_i$ are some primary operators, possibly including $\cO_1$. The derivatives $\ptl$ can be actual derivatives acting on the various $\cO_i$ and $\w$, or instead $\ptl$ can represent various tensors such as the curvature tensor, containing an equivalent number of derivatives. Dimensional analysis then requires
\be
	p-\De_1=m+\sum_i n_i(p-\De_i),
\ee
where $\De_i$ is the scaling dimension of $\cO_i$. Note that this is a rational relation between the scaling dimensions of primary operators $\cO_1$ and $\cO_i$.

Many such relations exist in free theory, which may introduce subtleties when applying this formalism to weakly-coupled theories. For example, a marginally irrelevant operator would come with a coupling that transforms anomalously according to a corresponding beta function. In generic interacting theories, however, we only expect such relations when protected operators are present. In the rest of this section, we will briefly discuss one generic class of such anomalous terms, associated with the displacement operator $D_\mu$.

The displacement operator $D_\mu$ exists on any local defect, is valued in the normal bundle, and has a protected scaling dimension $\De=p+1$. The above relations allow anomalous terms of the form $\de \l_{\cO}\ni \ptl^n\l_{D}^n\l_{\cO}$. In particular, the coupling $\l_D$ itself could have an anomalous transformation law.

Therefore, the problem of writing down a Weyl-invariant effective action seems to generically require a determination of such anomalous terms arising from $\l_D$. Fortunately, we have already implicitly solved this problem in section~\ref{sec:displacement}. Indeed, a choice of the renormalisation scheme is nothing more than a definition of the partition function as a function of its couplings at non-linear level. In section~\ref{sec:displacement} we showed that there exists a Weyl-invariant definition of the displacement coupling $\l_D$, to all orders in $\l_D$, and therefore there are no anomalous terms proportional to $\l_D$ in the corresponding renormalisation scheme. While it might not be obvious how to implement this scheme in practice, our arguments do show that such Weyl-invariant schemes exist in principle, to the extent to which Weyl-covariant partition functions can be defined for an arbitrary shape and position of $\cD_\S$. 

Note that above ``Weyl-invariant'' refers only to the transformation law for $\l_D$ and to anomalous terms of the form $\de \l_{\cO}\ni \ptl^n\l_{D}^n\l_{\cO}$. In general, terms of the form $\de \l_{1}\ni \ptl^{n+p}\l_{D}^n$ responsible for defect Weyl anomalies are still allowed.

\paragraph{Displacement coupling and the location of $\cD_\S$} We define the displacement operator on a defect $\cD$ by the requirement that
\be
	\cD[e^{-\int d^pz \sqrt{\g}v^\mu D_\mu}]
\ee
is equivalent to $\cD$ displaced by the vector field $v^\mu$ as described in section~\ref{sec:displacement}. Displacement operator is a local operator of dimension $\De=p+1$ that can appear in $S_\text{eff}$ as any other irrelevant operator.

There is a trade-off between the choice of the insertion point for $\cD_\S$ and the coupling to the displacement operator in $S_\text{eff}$. Indeed, by definition, adding a displacement coupling
\be
	S_\text{eff}\ni\int d^pz \sqrt{\hat \g} \l^\mu_{\hat D} \hat D_\mu
\ee
is completely equivalent to deforming $\cD_{\S}$ by the vector field $\l^\mu_{\hat D}$. This means that we are free to choose any configuration for $\cD_\S$ as long as it tends to the common limit of $\cD_1$ and $\cD_2$ as $L\to 0$. Differences between possible choices can be compensated by the coupling $\l^\mu_D$. In this paper, we use the convention where the position of $\cD_\S$ coincides with that of $\cD_1$. Although not particularly symmetric, this choice simplifies the form of the effective couplings.\footnote{This is similar to the choice one makes when writing down the OPE of two local operators. It is common to choose $\cO_1(x)\cO_2(y)\sim \sum_{k}\cO_k(y)$, but in principle other choices for the coordinates of $\cO_k$ in the right-hand side are possible.}

The displacement coupling has the derivative expansion
\be
	\l_{\hat D}^\mu = a_{D,0}v^\mu + a_{D,1}\hat \nabla_t^\perp v^\mu+\cdots,
\ee
where the $1$-derivative term is only possible for $p=1$ when one can define the unit tangent vector $t^\mu$. The higher-derivative terms contribute at $O(L^3)$ to the action.

Whenever there is a permutation symmetry between $\cD_1$ and $\cD_2$, it is possible to argue that $a_{D,0}=\thalf$. In other words, the natural position for $\cD_\S$ is halfway between $\cD_1$ and $\cD_2$.

\subsection{Breaking of transverse rotations and relevant couplings}
\label{sec:transverse}

In order for the effective action~\eqref{eq:Seff_general_form} to describe the tail of an RG flow that terminates in $\cD_\S$, it is necessary that it does not contain any relevant operators. Since in a general defect fusion configuration we do not perform any fine-tuning, this normally requires that $\cD_\S$ must not have any relevant operators that can appear in~\eqref{eq:Seff_general_form}.

When discussing relevant operators, one usually focuses on scalars. However, in our case, due to the presence of $v^\mu$ and curvature invariants, it is necessary to also consider the defect operators that are charged under transverse rotations. For example, we can have couplings such as
\be
	v^\mu \hat\cO_{1,\mu}, \quad \hp^{\mu}_{ab}\hp^{\nu ab} \hat\cO_{2,\mu\nu},
\ee
where all the indices are transverse.
Taking into account the derivatives in the curvature invariants, the coupling with $\cO_1$ is relevant if $\De_1<p$, the coupling with $\cO_2$ is relevant if $\De_2<p-2$.

In the case of bulk RG flows, there is usually no need to consider such couplings since they are excluded by unitarity bounds. For example, in 3 dimensions the smallest possible dimension of a tensor operator is $2$ while the simplest curvature invariant is the Riemann tensor, which has 2 derivatives.\footnote{Couplings of Ricci scalar to relevant scalar operators still need to be taken into account even in the bulk.} In the defect case, the unitarity bounds do not depend on the transverse spin\footnote{For instance, consider $\ptl^n \f$-type operators on the trivial $p$-dimensional defect in $d$-dimensional free scalar theory, where all $n$ derivatives are transverse and $p=d-2$. Their dimension is $\De=n+p/2$, which shows that for any fixed transverse spin, $p-\De=p/2-n$ can be arbitrarily large if $p$ is large enough.} and therefore this logic does not apply.

If $p=1$, i.e.\ when we consider line defects, the only relevant couplings of this kind are
\be\label{eq:p=1relevant}
	v^\mu\hat \cO_{\mu},\quad  v^{\mu}v^\nu \hat \cO_{\mu\nu},\cdots
\ee
This implies that $\cD_\S$ should not have any (otherwise neutral) relevant operators which are transverse traceless-symmetric tensors.

Starting at least with $p=3$ surface defects, curvature couplings such as  $\hp^\mu_{ab}\hp^{\nu ab} \hat\cO_{2,\mu\nu}$ or $\hp^\mu_{ab}\hp^{\nu ab} v^\r\cO_{3,\mu\nu\r}$ can appear. In generic configurations, this prohibits $\cD_\S$ from having relevant operators with transverse spin in certain representations. In special configurations, such as for flat defects, the curvature couplings can vanish, which opens the possibility that the fusion defect $\cD_\S$ can be different for different fusion geometries. For example, one can imagine a situation where for flat defects $\cD_1$ and $\cD_2$ the IR defect $\cD_\S$ has a sufficiently relevant transverse mixed-symmetry tensor $\cO_{3,\mu\nu\r}$, which cannot couple to a constant $v^\mu$ alone. When we deform the defects away from the flat configuration, the curvature coupling $\hp^\mu_{ab}\hp^{\nu ab} v^\r\cO_{3,\mu\nu\r}$ will turn on and trigger an RG flow to a new IR defect. In other words, fine-tuning of the fusion RG flow might in some cases be achieved by considering defects of special shape.

Related to the issue of relevant operators with transverse spin is possibility that $\cD_{\S}$ itself breaks the transverse rotations $\SO(q)$ down to a  proper subgroup $H$. Indeed, in the picture where we view the product $\cD_1\cD_2$ as the UV definition of a non-conformal defect that flows to $\cD_{\S}$ in the IR, we have to contend with the fact that this UV definition preserves fewer symmetries than the individual defects $\cD_1$ or $\cD_2$. As in the discussion above, this is manifested by presence of $v^\mu$ and of various curvature invariants.

If $\cD_\S$ does break transverse symmetries down to $H$, this means that in order to compute partition functions with $\cD_\S$ inserted, one needs to specify an element $n(z)\in \SO(q)/H$ at every point of the defect. In the simplest case, when $H=\SO(q-1)$ and $q>2$, we have $\SO(q)/H=S^{q-1}$ and $n$ can be viewed as a unit vector (with respect to $g_{\mu\nu}$). The only possible value for $n$ in a fusion process is
\be
	n=\pm \ell^{-1}v+O(L).
\ee
The choice of sign here should be viewed as a Wilson coefficient that depends on the particular pair of defects $\cD_1$ and $\cD_2$. It can be changed by altering the meaning of $n$ to $-n$ for $\cD_\S$, but this will change the sign in all fusions in which $\cD_\S$ appears (so relative signs are meaningful). 

For completeness, we give the form of the first subleading terms for $p=1$ in dimensions $d>3$. In general, 
\be	
	n&=\pm e^{m_{\mu\nu}r^\mu v^\nu}\ell^{-1}v=\pm\ell^{-1}\p{v^\mu+r^\mu-(r\.v)v^\mu+\cdots},
\ee
where $(m_{\r\s})^{\mu}{}_\nu=\de^\mu_\r \hat g_{\s\nu}-\de^\mu_\s \hat g_{\r\nu}$ is a rotation generator and $r^\mu$ is a Weyl weight-0 vector in the normal bundle, defined modulo $v^\mu$. The leading order derivative expansion for $r$ is
\be
	r^\mu = a_{r,1}\hat\nabla^\perp_t v^\mu+ a_{r,2}(\hat\nabla^\perp_t)^2 v^\mu+\cdots \mod v^\mu,
\ee
where $t^a$ is a unit (in $\hat g$) vector tangent to the line defect, aligned with the orientation. Note that we can't use the Schouten tensor (equivalently, the Ricci tensor) due to~\eqref{eq:P=0}, and $Rv^\mu=0\mod v^\mu$. We also used~\eqref{eq:II=0}.

Recall that there is a tradeoff between the value of $n$ and the coupling to the tilt operators in $S_\text{eff}$. When writing the above formula, we assume the convention in which the tilt operators do not appear in $S_\text{eff}$. Alternatively, and this is the convention that we adopt more broadly in this paper, we can set $n=\pm \ell^{-1}v$ \textit{exactly}, and compensate for this by adding the coupling $r^\mu$ to the tilt operator corresponding to broken transverse rotations. A similar tradeoff is present for any exactly marginal deformation.

In the case $d=3$ and $p=1$ the complication is that there is not in general a canonical identification between the elements of $\SO(2)/\SO(1)$ and the unit normal vectors, even up to a sign.\footnote{If the defect preserves transverse parity, then we have to work with $\mathrm O(2)/\mathrm O(1)$ and a canonical identification up to a sign is again possible.} In practice this means that we can write
\be
	n=R_\theta (\ell^{-1}v),
\ee
where $R_\theta$ is a rotation by angle $\theta$ in the normal plane, and we have the derivative expansion
\be
	\theta = a_{\theta,0}+\cdots,
\ee
where $a_{\theta,0}$ is a Wilson coefficient. Similarly to the sign of $n$, we can adjust $a_{\theta,0}$ by changing the meaning of $n$ to $R_{\vf}(n)$ for $\cD_\S$, but this will affect all fusions in which $\cD_\S$ appears, so relative values of $a_{\theta,0}$ are meaningful.

Breaking to $H$ other than $\SO(q-1)$ is not possible if we consider flat parallel defects, as this configuration will explicitly preserve $\SO(q-1)$. So, such $H$ can be only achieved in the situation of the kind mentioned earlier in this section, when the IR defect $\cD_\S$ is different for generic and flat configurations. Note that in this scenario, the dimension $p$ of the defect plays a role. For example, for $p=1$ no relevant curvature couplings can turn on as we move from flat to curved defects. Therefore, we expect that $H$ smaller than $\SO(q-1)$ cannot be realised for line defects. It is an interesting question whether scenarios of this kind can be realised for higher-dimensional defects.

\section{Cusps}
\label{sec:Wilson}

In this section we first discuss the general properties of a cusp formed by two line defects with a small opening angle $\a$, and then focus on the specific example of supersymmetric Wilson lines in $\cN=4$ SYM.

\subsection{General structure of scaling dimensions at a cusp}
\label{sec:GcuspNP}

\begin{figure}[t]
	\centering
	\scalebox{1}{\begin{tikzpicture}
	\begin{pgfonlayer}{nodelayer}
		\node [style=none] (0) at (-2, 0) {};
		\node [style=none] (1) at (3, 0) {};
		\node [style=none] (2) at (2, 3) {};
		\node [style=none] (3) at (0, 1.5) {};
		\node [style=none] (4) at (0.5, 0) {};
		\node [style=none] (5) at (-0.25, 2) {$\mathcal{D}_2$};
		\node [style=none] (6) at (-1, 0.35) {$\alpha$};
		\node [style=none] (7) at (1, -0.5) {$\mathcal{D}_2$};
		\node [style=none] (8) at (-1.25, 0) {};
		\node [style=none] (9) at (-1.325, 0.5) {};
	\end{pgfonlayer}
	\begin{pgfonlayer}{edgelayer}
		\draw [style=new edge style 14] (3.center) to (2.center);
		\draw [style=new edge style 0] (3.center) to (0.center);
		\draw [style=new edge style 1] (0.center) to (4.center);
		\draw [style=new edge style 8] (4.center) to (1.center);
		\draw [bend left, looseness=0.75] (9.center) to (8.center);
	\end{pgfonlayer}
\end{tikzpicture}}
	\caption{A cusped junction of $\cD_1$ and $\cD_2$}
	\label{fig:cusp}
\end{figure}
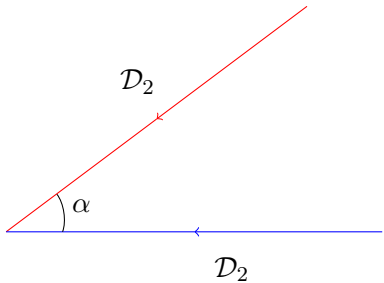

We consider line defects $\cD_1$ and $\cD_2$ and assume that they can be joined at a cusp as in figure~\ref{fig:cusp}. Using the standard arguments, operator-state correspondence implies that the eigenstates of the dilatation operator in the Hilbert space of states on $S^{d-1}$ punctured by $\cD_1$ and $\cD_2$ are in one-to-one correspondence with local operators at the cusp. In particular, one expects infinitely-many choices for local operators that can be present at the cusp.\footnote{The use of ``local operator'' is a bit confusing since it makes it sound as if it were possible to not insert anything at the cusp. Instead of ``local operators at the cusp'', it would perhaps be more useful to talk about different cusped junctions between line defects. Nevertheless, we prefer to use the term ``local operator''`.}${}^,$\footnote{Note that the Hilbert spaces with different angles $\a$ between $\cD_1$ and $\cD_2$ are all isomorphic via a conformal transformation to the Hilbert space with $\a=\pi$. What really depends on $\a$ is the dilatation operator and its spectrum.}

One usually distinguishes the local operator with the lowest scaling dimension. This scaling dimension is denoted $\G_\text{cusp}$ and the corresponding local operator is often described as the ``cusp with no insertions''. In $\cN=4$ SYM literature it is customary to write $\G_\text{cusp}$ for fundamental Wilson lines as the function of $\f=\pi-\a$.

Via the exponential map, the configuration in figure~\ref{fig:cusp} is conformally equivalent to the cylinder $\R\x S^{d-1}$, where $\cD_1$ and $\cD_2$ are inserted along $\R$ at fixed points on $S^{d-1}$ separated by an angle $\a$. The cusp is mapped to the infinite past. Let $H$ be the Hamiltonian that generates the time translations along $\R$. This Hamiltonian is simply the dilatation operator in figure~\ref{fig:cusp} and thus its spectrum is given by the scaling dimensions of cusp local operators  (we set the radius of $S^{d-1}$ to 1). In particular, $\G_\text{cusp}$ is the ground state energy.

In the limit $\a\to 0$ the defects $\cD_1$ and $\cD_2$ on $\R\x S^{d-1}$ fuse. Let $\cD_\S$ be the corresponding IR defect. Recall from~\eqref{eq:Happrox} that the spectrum of $H$ can be approximated as
\be\label{eq:Happrox2}
H\overset{\text{spectrum}}{=}H_\S+\l_{\mathbf{1}}(\a)\mathbf{1}+\sum_{\cO}\l_\cO(\a)\cO,
\ee
where the sum over $\cO$ is over all primary operators on $\cD_\S$ and $H_\S$ is the Hamiltonian when the insertions of $\cD_1$ and $\cD_2$ are replaced by $\cD_\S$. Any possible index contractions are implicit. 

The coefficients $\l_\cO(\a)$ behave as
\be
\l_\cO(\a)\sim \a^{\De_\cO-1}.
\ee
Note that by the assumption that $\cD_\S$ is the IR fixed point, the sum in~\eqref{eq:Happrox2} includes only irrelevant operators, so their couplings vanish as $\a\to 0$.\footnote{\label{footnote:marginal}We might have exactly marginal operators with $\De_\cO=1$. For them we use the convention described in section~\ref{sec:Seff}, which makes their effective couplings vanish as at least $O(\a)$.}

Using~\eqref{eq:Happrox2} we can in principle predict the small-$\a$ behaviour of the ground-state energy $\G_\text{cusp}$ (or of any other energy level) to any order in $\a$ using Rayleigh–Schr\"odinger perturbation theory. The simplest contribution comes from the identity operator, which shifts all energies by $\l_\mathbf{1}(\a)$. In particular, the ground state energy $\G_\text{cusp}$ is given by
\be\label{eq:cuspgeneralsplit}
\G_\text{cusp}=E_{0,\S}+\l_\mathbf{1}(\a)+\de\G_\text{cusp},
\ee
where $E_{0,\S}$ is the ground-state energy of $H_\S$ and $\de \G_\text{cusp}$ is the contribution from non-identity $\cO$ in~\eqref{eq:Happrox2}. The latter contribution vanishes in the limit $\a\to 0$.

The coupling $\l_\mathbf{1}(\a)$ has an expansion in powers of $\a$. In appendix~\ref{app:lines} we classify 2-derivative contributions to $\l_\mathbf{1}$, see~\eqref{1d-4plusd-2nd}. In the conformally-flat case of $\R\x S^3$ the result is, up to two derivatives,
\be\label{eq:SeffIdentityLocal}
\l_\mathbf{1} = &\int dz\sqrt{\hga}\p{
	-a_0+a_{2,1}\hat\nabla_t^\perp v\. \hat\nabla_t^\perp v+a_{2,2}\hat R+ia_{2,3}\hat P_{tv}+\cdots
},
\ee
where $\hat P_{\mu\nu}$ is the Schouten tensor and $t$ is a unit (in the fusion metric) vector field on the line defect.  The terms multiplying $a_{2,1}$ and $a_{2,3}$ vanish. The derivative $\hat \nabla^\perp_t v$ vanishes since the separation between the defects is time-independent. The component $\hat P_{tv}$ of the Schouten tensor vanishes due to the transverse rotational symmetry of the fusion metric.  This symmetry will be present whenever the IR defect is a straight line along $\R$. This is discussed in more detail in section~\ref{sec:fusionmetricprops} and section~\ref{sec:exampleII} where we also show that for this configuration $\hat R=-(d-1)\a^2$, see equation~\eqref{eq:RicciCylinder}.

This shows that
\be\label{eq:cuspIdGeneral}
\l_\mathbf{1}(\a)=-\frac{a_0}{\a}-(d-1)a_{2,2}\a+\cdots.
\ee
In fact, it is easy to see that $\l_\mathbf{1}(\a)$ has an expansion in odd powers of $\a$. This is because it is given by $\a^{-1}$ times local invariants which are all constructed from the derivatives of the fusion metric as neither $v^\mu$ nor the embedding function of $\cD_1$ have interesting derivatives. Each metric derivative effectively gives a power of $\a$. As discussed in section~\ref{sec:exampleII} around~\eqref{eq:ell(x)forCylinder}, in this setting the fusion metric has only even derivatives on $\cD_1$.

Turning now to $\de\G_\text{cusp}$, it is relatively easy to give a formal expression for it. Equation (B.18a) of~\cite{Hogervorst:2021spa} conveniently encodes the all-order perturbation theory for the ground state energy. Using it, we find
\be\label{eq:all-order-Gcusp}
	\de\G_\text{cusp}=
	&\sum_{\cO}\l_\cO(\a)\<0|\cO|0\>-\sum_{\cO_1\cO_2}\l_{\cO_1}(\a)\l_{\cO_2}(\a)\int_0^\oo d\tau\<0|\cO_2(\tau)\cO_1(0)|0\>_\text{conn}\nn\\
	&+\sum_{\cO_1\cO_2\cO_3}\l_{\cO_1}(\a)\l_{\cO_2}(\a)\l_{\cO_3}(\a)\int_0^\oo d^2\tau \<0|\cO_3(\tau_1+\tau_2)\cO_2(\tau_1)\cO_1(0)|0\>_\text{conn}\nn\\
	&-\cdots,
\ee
where $\tau$ is the global time on $\R\x S^{d-1}$. Here, the only dependence on $\a$ is through the couplings $\l_\cO(\a)$. Note that the state $|0\>$ corresponds to an operator on which $\cD_\S$ ends, which is why $\<0|\cO|0\>$ is generally non-zero and the two-point functions $\<0|\cO_2(\tau)\cO_1(0)|0\>$ are not diagonal. If $\cD_\S$ is the trivial defect, the first two orders simplify,
\be\label{eq:all-order-Gcusp_triv}
\de\G_\text{cusp}=&-\sum_\cO \l_\cO^2(\a)\int_0^\oo d\tau\<0|\cO(\tau)\cO(0)|0\>_\text{conn}\nn\\
&+\sum_{\cO_1\cO_2\cO_3}\l_{\cO_1}(\a)\l_{\cO_2}(\a)\l_{\cO_3}(\a)\int_0^\oo d^2\tau \<0|\cO_3(\tau_1+\tau_2)\cO_2(\tau_1)\cO_1(0)|0\>_\text{conn}\nn\\
&-\cdots,
\ee
where we used that $\<0|\cO|0\>=0$ for non-identity operators $\cO$, and that that the two-point functions are diagonal in this case. Since we are working on the cylinder, the correlation functions in~\eqref{eq:all-order-Gcusp} and~\eqref{eq:all-order-Gcusp_triv} are not pure powers and decay exponentially at large times (if they were pure powers, the renormalised integrals would vanish). Finally, note that these expressions contain connected correlators.

One immediate consequence of this is that the small-$\a$ expansion of $\de\G_\text{cusp}$ contains the following powers of $\a$,
\be
	\de\G_\text{cusp}\ni \a^{\De_1+\cdots+\De_n-n+2k},
\ee
where $n\geq 1$ and $k\geq 0$ are integers, while $\De_i$ run over the scaling dimension of the primary operators in~\eqref{eq:Happrox2}. In particular, $\de \G_\text{cusp}=O(\a^{\De_\text{min}-1})$, where $\De_\text{min}$ is the smallest scaling dimension of an operator appearing in~\eqref{eq:Happrox2}.${}^{\ref{footnote:marginal}}$. When $\cD_\S$ is trivial, this is modified in the obvious way and we find that $\de \G_\text{cusp} = O(\a^{2\De_\text{min}-2})$.

The above discussion can be straightforwardly generalised to excited states, although~\eqref{eq:all-order-Gcusp} needs to be modified in that case~\cite{Hogervorst:2021spa}. One important difference is that one-point functions generically do not vanish in excited states even for trivial $\cD_\S$.

\subsection{Supersymmetric Wilson lines in $\cN=4$ SYM}
In this section we briefly examine the fusion of fundamental and anti-fundamental supersymmetric Wilson lines~\cite{Maldacena:1998im} in $\cN=4$ SYM, focusing on the planar theory.

A (locally) supersymmetric Wilson line along a curve $\cD$ is defined by
\be
	Pe^{i \int_\cD A+\int_\cD ds (n\.\f)},
\ee
where $\f$ is the fundamental scalar and $n\in \R^6$ is a unit vector in the vector representation of $\SO(6)$ R-symmetry (which is therefore explicitly broken to $\SO(5)$ by the Wilson line). The path-ordered exponential is computed in the fundamental representation of the $\SU(N)$ gauge group. We will denote the resulting defect by $\cD_n$.

We consider the fusion of $\cD_1 = \cD_{n_1}$ with $\cD_2 = \bar \cD_{n_2}$ and we use $\theta$ to denote the angle between $n_1$ and $n_2$. We expect that the resulting IR defect $\cD_\S$ is the trivial defect. Informally, the pair of quarks represented by the Wilson lines can form a singlet or an adjoint state, and the former is more energetically favorable.\footnote{The potential in the adjoint representation vanishes in the planar theory~\cite{Pineda:2007kz}, while the singlet potential is non-zero and attractive.} More formally, the line $\cD_\S$ can end and we will show that the spectrum of local operators it can end on coincides with the bulk spectrum. In other words, we will describe the spectrum of the Hamiltonian $H_\S$ on $\R\x S^{3}$ when there is a single insertion of $\cD_\S$ along the time direction~$\R$, and we shall see that it agrees with the spectrum of the bulk dilatation operator.

To this end, we consider the configuration discussed in section~\ref{sec:GcuspNP} where the straight lines $\cD_1$ and $\cD_2$ meet at a cusp with an opening angle $\a=\pi-\f$. We also retain the notation from section~\ref{sec:GcuspNP}. Due to~\eqref{eq:Happrox2},~\eqref{eq:cuspgeneralsplit} and~\eqref{eq:cuspIdGeneral}, the spectrum of $H_\S$ is obtained from the spectrum of $H$  in the limit $\a\to 0$ after subtracting the singular part of ground state energy $E_0=\G_\text{cusp}$. We will see in various limits below that $\G_\text{cusp}=-\frac{a_0}{\a}+O(\a)$, and we expect this to hold non-perturbatively.\footnote{Nikolay Gromov has informed us that unpublished numerical data for $\G_\text{cusp}(\a)$ at $g=\sqrt{\l}/4\pi=3/8$ and $\theta=0$ obtained from QSC in~\cite{Gromov:2015dfa} shows that the coefficient of $\a^0$ in $\G_\text{cusp}$ is $0\pm10^{-4}$. By comparison, $a_0\approx 0.80$ and the coefficient of $\a$ is $\approx 0.23$.\label{footnote:constanttermis0}} Therefore, we can simply subtract the full $\G_\text{cusp}$.

We can classify the cusp operators into single-trace and multi-trace. The single-trace operators are obtained from insertions of fundamental fields in the Wilson line trace, while multi-trace operators may have additional traces. An example of a single-trace operator is
\be\label{eq:singletrace}
\Tr\{\cdots Pe^{i \int_{\cD_1} A+\int_{\cD_1} ds (n_1\.\f)}\f_i Pe^{i \int_{\cD_2} A+\int_{\cD_2} ds (n_2\.\f)}\cdots \},
\ee
while an example of a multi-trace operator is
\be\label{eq:multitrace}
\Tr(\f_i \f_i)\Tr\{\cdots Pe^{i \int_{\cD_1} A+\int_{\cD_1} ds (n_1\.\f)}Pe^{i \int_{\cD_2} A+\int_{\cD_2} ds (n_2\.\f)}\cdots \},
\ee
where $\Tr(\f_i \f_i)$ is inserted at the cusp.  Note that in the planar theory the scaling dimension of~\eqref{eq:multitrace} is just $\G_\text{cusp}+\De_K=E_0+\De_K$ where $\De_K$ is the scaling dimension of the bulk Konishi operator $\Tr(\f_i \f_i)$. Therefore, the spectrum of $H_\S$ contains $\De_K$. In a similar way, it contains the scaling dimension of any bulk single- or multi-trace operator. On the other hand, we expect that the excited single-trace states such as~\eqref{eq:singletrace} have energies which behave as $E\sim -\frac{a'_0}{\a}$ with $a'_0<a_0$.\footnote{Note that this is a non-degeneracy assumption on $a_0$. If there were an exact degeneracy for generic couplings, it would also be present in the ladder limit. We will see in section~\ref{sec:ladder} that in the ladder limit $a_0$ is non-degenerate.} Therefore, $E-E_0\to +\oo$  as $\a\to 0$ for states which involve such single-trace factors, and only the multi-trace operators similar to~\eqref{eq:multitrace} survive in the spectrum of $H_\S$. This shows that the spectrum of $H_\S$ agrees with the spectrum of the bulk dilatation operator and supports our claim that $\cD_\S$ is trivial.

If we focus on a finite set of lowest-energy states of $H$, then the above argument shows that for sufficiently small $\a$ their energies are given exactly by
\be\label{eq:Hexact}
	H\overset{\text{low-E spectrum}}{=}\G_\text{cusp}+H_\S.
\ee
In particular, the $\a$-dependence of $H$ energy levels is state-independent. This statement should have non-trivial implications for the couplings $\l_\cO(\a)$ entering in~\eqref{eq:Happrox2}. We leave this question for the future. Here we will focus on the contribution $\l_\textbf{1}(\a)$ of the identity operator. Since all non-trivial bulk local operators in $\cN=4$ SYM have dimension at least 2, it follows that, at least, $\G_\text{cusp}=\l_{\textbf{1}}(\a)+O(\a^2)$.

We therefore find for the cusp anomalous dimension, taking into account~\eqref{eq:cuspIdGeneral},
\be\label{eq:Gcusp_prediction}
	\G_\text{cusp}(\theta,\f=\pi-\a)=-\frac{a_0(\theta)}{\a}-3a_{2,2}(\theta)\a+O(\a^2),
\ee
where $a_{2,2}$ is the Wilson coefficient appearing in~\eqref{eq:SeffIdentityLocal}. The function $\G_\text{cusp}(\theta,\f)$ can in principle be computed using integrability~\cite{Gromov:2015dfa}, and therefore the Wilson coefficients $a_0$, $a_{2,2}$,  as well as other coefficients determining subleading terms in $\G_\text{cusp}$ are accessible through integrability techniques. The coefficient $a_0$ has been studied using Quantum Spectral Curve (QSC) in~\cite{Gromov:2016rrp}. It would be interesting to find a QSC description for the subleading coefficients such as $a_{2,2}$. Numerical data from QSC supports the absence of $\a^0$ term, see footnote~\ref{footnote:constanttermis0}. In the following two subsections we address a simpler problem and compute $a_{2,2}$ in the ladder limit and in the strong coupling limit.

\subsection{Ladder limit}
\label{sec:ladder}

The first limit we consider is the limit $e^{i\theta}\to \oo$ and 't Hooft coupling $\l\to 0$ with
\be
	\hat \l = \frac{\l e^{i\theta}}{4}
\ee
held fixed. As discussed in~\cite{Correa:2012nk}, in this limit only the ladder diagrams contribute to $\G_\text{cusp}$, which can then be determined through the ground state energy of the Schr\"odinger equation~\cite{Correa:2012nk}
\be\label{eq:Sch1}
	-\psi''(x)-\frac{\hat \l}{8\pi^2}\frac{1}{\cosh x-\cos\a}\psi(x)=-\frac{\G_\text{cusp}^2}{4}\psi(x),
\ee
where $x\in \R$.

We are interested in the limit $\a\to 0$. In this limit, the potential becomes singular at $x=0$ and $\G_\text{cusp}$ behaves according to~\eqref{eq:Gcusp_prediction}. To study this limit, we write $x=\a y$ and 
\be\label{eq:GammaOmega}
	\G_\text{cusp}=-\frac{\Omega}{\a}.
\ee
Equation~\eqref{eq:Sch1} becomes
\be
	-\psi''(y)-\frac{\hat \l}{8\pi^2}\frac{\a^2}{\cosh \a y-\cos\a}\psi(y)=-\frac{\Omega^2}{4}\psi(y).
\ee
The potential now has a regular expansion at small $\a$,
\be\label{eq:Uexpansion}
	-\frac{\hat \l}{8\pi^2}\frac{\a^2}{\cosh \a y-\cos\a}=U_0(y)+\a^2 U_2(y)+\cdots,
\ee
where
\be
	U_0(y) = -\frac{\hat\l}{4\pi^2}\frac{1}{1+y^2},\quad U_2(y)=\tfrac{1}{6}U_0(y)+\frac{\hat\l}{48\pi^2}.
\ee
Let $-\Omega_0^2(\hat \l)/4$ be the ground state energy of the leading-order Schr\"odinger equation
\be\label{eq:Sch2}
-\psi''(y)-\frac{\hat \l}{4\pi^2}\frac{1}{1+y^2}\psi(y)=-\frac{\Omega_0^2(\hat \l)}{4}\psi(y).
\ee
Equation~\eqref{eq:Uexpansion} then implies that
\be
	-\Omega^2/4=-\Omega_0^2\p{\hat \l(1+\tfrac{\a^2}{6})}/4+\frac{\hat\l\a^2}{48\pi^2}+O(\a^4).
\ee
In other words,
\be
	\Omega=\Omega_0(\hat\l)-\frac{\tfrac{1}{4\pi^2}-\Omega_0(\hat\l)\ptl_{\hat\l}\Omega_0(\hat\l)}{6\Omega_0(\hat\l)/\hat\l}\a^2+O(\a^4)
\ee

Using~\eqref{eq:GammaOmega} and~\eqref{eq:Gcusp_prediction} this gives for the Wilson coefficients
\be
	a_0=\Omega_0(\hat\l),\quad a_{2,2}=\frac{\Omega_0(\hat\l)\ptl_{\hat\l}\Omega_0(\hat\l)-\tfrac{1}{4\pi^2}}{18\Omega_0(\hat\l)/\hat\l}.
\ee
The weak-coupling expansion of $\Omega_0(\hat\l)$ was computed in~\cite{Correa:2012nk} to $O(\hat\l^3)$. For simplicity, we reproduce the first two orders,
\be
	a_0=\Omega_0(\hat\l)=\frac{\hat\l}{4\pi}+\frac{\hat\l^2}{8\pi^3}\p{\log\tfrac{\hat\l}{2\pi}+\g_E-1}+O(\hat\l^3).
\ee
This yields
\be
	a_{2,2}=-\frac{1}{18\pi}+\frac{\hat\l}{36\pi^3}\p{\log\tfrac{\hat\l}{2\pi}+\g_E-1+\tfrac{\pi^2}{2}}+O(\hat\l^2).
\ee

It may seem surprising that $a_{2,2}$ starts at $O(1)$ while at fixed $\a$ the expansion of $\G_\text{cusp}$ starts at $O(\hat\l)$. This, as well as the presence of the logarithms of the coupling, is due to the fact that the small-$\a$ expansion of $\G_\text{cusp}$ does not commute with the small-$\hat \l$ (more generally, small $\l$) expansion~\cite{Correa:2012nk,Pineda:2007kz}. In particular, the expansion~\eqref{eq:Gcusp_prediction} is only valid for angles $\a\ll \hat\l$. We briefly discuss the origin of this in section~\ref{sec:simplicity} from the fusion point of view.

Finally, we note that the ground state of~\eqref{eq:Sch2} is non-degenerate and therefore the first excited energy of~\eqref{eq:Sch1} scales as $-a_0'/\a$ with $a_0'<a_0$. These energies correspond to other single-trace insertions at the cusp.\footnote{Equation~\eqref{eq:Sch1} describes a subset of single-trace insertions at the cusp (as bound states or resonances, see~\cite{Cavaglia:2018lxi}). The remaining single-trace states have protected scaling dimensions in the ladder limit, i.e.\ for them $a'_0=0$. Still, $a'_0<a_0$ is satisfied. We thank Nikolay Gromov for discussions on this point.} This supports the claim that we made earlier in this section.

\subsection{Strong coupling}

It is also possible to compute $a_0$ and $a_{2,2}$ at strong coupling using the results of~\cite{Drukker:2011za}. In particular,~\cite{Drukker:2011za} gives an explicit expression for the leading $O(\sqrt{\l})$ term in $\G_\text{cusp}$ for generic $\a$ and $\f$ in their appendix B, and they also extract the leading coefficient $a_0$. By a straightforward extension of their analysis we find
\be
	a_0(\theta) &= \frac{2\sqrt{\l}}{k\sqrt{1-k^2}\pi}\p{E-(1-k^2)K}^2+O(\l^0),\nn\\
	a_{2,2}(\theta) &= -\frac{\sqrt{\l}}{36k\sqrt{1-k^2}\pi}\frac{(1-2k^2)^2E^2-(1-k^2)(3k^4-4k^2+2)EK+(1-k^2)^2K^2}
	{
		\p{E-(1-k^2)K}\p{(1-2k^2)E-(1-k^2)K}
	}+O(\l^0),
\ee 
where $K=K(k^2)$ and $E=E(k^2)$ are complete elliptic integrals of the first and second kind, respectively. The elliptic modulus $k$ is related to $\theta$ via
\be
	\theta=2\sqrt{1-2k^2}K(k^2).
\ee
In the special case of $\theta=0$ we obtain
\be
	a_0(0) = \frac{4\pi^2\sqrt\l}{\G(\tfrac{1}{4})^4}+O(\l^0),\quad a_{2,2}(0)=-\frac{\sqrt\l \p{24\pi^2-\G(\tfrac{1}{4})^4}}{576\pi^3}+O(\l^0).
\ee
We note in passing that $a_0(0)>0$ and is the global maximum of $a_0(\theta)$ as required by theorems~\ref{thm:oppositesattract} and~\ref{thm:cauchyschwarz}.

\section{Local operator one-point function asymptotics}
\label{sec:operatorasymptotics}

In this section we study the implications of conformal defect fusion for the OPE of the defect two-point function $\<\cD_1\cD_2\>$. Recall that each of the two defects $\cD_1$ and $\cD_2$ can be expanded in a basis of local operators~\cite{Gadde:2016fbj}, which yields a convergent expansion of the schematic form
\be\label{eq:primaryOPE}
\<\cD_1\cD_2\> = \sum_\cO  C_{\cD_1\cO}C_{\cD_2\cO^\dagger} G_{\cO}.
\ee
Here, the sum is over the primary operators $\cO$, $C_{\cD_i\cO}$ is the one-point function of $\cO$ in the presence of $\cD_i$, and $G_{\cO}$ is a conformal block which depends only on the cross-ratios describing the configuration of $\cD_i$ and on the quantum numbers of $\cO$. Compatibility of this OPE with the fusion~\eqref{eq:mainEFT} leads to relations between the Wilson coefficients in the effective action and the one-point functions $C_{\cD_1\cO}$ of primary operators. This is the relation that we will explore.

We will focus on the case $\cD_1\simeq \bar \cD_2$ in which case $C_{\cD_1\cO}C_{\cD_2\cO^\dagger}=|C_{\cD_1\cO}|^2\geq 0$. Furthermore, we will simplify the problem by studying a simpler expansion in which the descendants and the primaries are treated independently.

\subsection{Defect two-point function from OPE expansion}

We consider the generic configuration of two dimension-$p$ conformal defects that we previously studied as an example in section~\ref{sec:displacement}. Specifically, we take $\cD_1$ to be a $p$-dimensional spherical defect of radius $1$ centred at $0$ and lying in the subspace spanned by the coordinate directions $1,\cdots,p+1$. The position of the defect $\cD_2$ is obtained by the action of
\be\label{eq:general_rotation_OPE_section}
e^{\log r \bd-\sum_{i} \theta_i\bm_{i,p+i+1}}
\ee
where the sum is over $i=1,\cdots,m=\min\{p+1, d-p-1\}$. The total number of parameters $r,\theta_1,\cdots, \theta_m$ is $m+1=\min\{d+2-q,q\}$. This is the same as the number of cross-ratios for a pair of $p$-dimensional defects~\cite{Gadde:2016fbj}, and thus the configurations of the above form cover at least an open neighbourhood of the fusion limit $r=1,\theta_i=0$ in the cross-ratio space.

When $\cD_2$ is conjugate to $\cD_1$, we can write the two-point correlation function as
\be
\<\cD_1\cD_2\>=\<\bar\cD_1|e^{\log r D-\sum_{i} \theta_iM_{i,p+i+1}}|\bar\cD_1\>=\sum_n \<\bar\cD_1|n\>\<n|e^{\log r D-\sum_{i} \theta_iM_{i,p+i+1}}|\bar\cD_1\>,
\ee
where all states are in the Hilbert space of the unit sphere, and the sum over $n$ is a sum over an orthonormal basis of states. Since $D$ and $M_{i,p+i+1}$ with $i=1,\cdots,m$ are mutually commuting, we can assume that they are diagonalised in our basis,
\be
D|n\> &= \De_n|n\>,\\
M_{i,p+i+1} |n\> &= iJ_{i,n}|n\>.
\ee
The expression for the two-point function then becomes
\be
\<\cD_1\cD_2\>=\sum_n |\<\bar\cD_1|n\>|^2 r^{\De_n}e^{-i\sum_i J_{i,n}\theta_i}.
\ee
Notice that $\<\bar\cD_1|n\>$ is equal to the one-point function of the local operator corresponding to the state $|n\>$ in the presence of the defect $\cD_1$. 

We can rewrite the above equation equivalently as
\be\label{eq:full_two_pt_OPE}
\<\cD_1\cD_2\>=\int d\De dJ_1 \cdots dJ_m \r(\De,J_1,\cdots, J_m) r^{\De}e^{-i\sum_i J_{i}\theta_i},
\ee
where
\be
\r(\De,J_1,\cdots, J_m) = \sum_n |\<\bar\cD_1|n\>|^2 \de(\De-\De_n)\de(J_1-J_{1,n})\cdots \de(J_m-J_{m,n})
\ee
is a non-negative density of one-point functions.

The density $\r(\De,J_1,\cdots,J_m)$ receives contributions from both primaries and descendants. It is possible to recover the density of one-point functions of primary operators from $\r$ by analyzing the series expansion of the conformal block $G_\cO$. In this context, it is important to note that only the primary operators in spin representations described by at most $m$-row Young diagrams can have non-zero one-point functions with a $p$-dimensional defect~\cite{Lauria:2018klo}. Therefore, $\r(\De,J_1,\cdots,J_m)$ has the correct number of variables to discern all non-zero one-point functions of primaries.

We will focus on the density $\r(\De,J_1,\cdots,J_m)$. We will see that, within the precision that we work at, the density of primary one-point functions has the same form.

\subsection{Defect two-point function from fusion}

The limit $r\to 1$ and $\theta_i\to 0$ corresponds to the fusion of the defects $\cD_1$ and $\cD_2$, allowing us to obtain an approximation to the two-point function $\<\cD_1\cD_2\>$ using~\eqref{eq:mainEFT}. Specifically, we have the following equality in the $L\to 0$ limit
\be
\<\cD_1\cD_2\> \sim \<\cD_\S[e^{-S_\text{eff}}]\>
\ee
with $r=1+O(L), \theta_i=O(L)$. Note that the right-hand side should be understood as an asymptotic series in the obvious way. In particular, this does not give an exact prediction for the two-point function at any finite value of $L$.

Let us split $S_\text{eff}=S^{\mathbf{1}}_\text{eff}+S^\text{rest}_\text{eff}$, where $S^{\mathbf{1}}_\text{eff}$ contains the contribution of the identity operator, and $S^\text{rest}_\text{eff}$ contains the contributions from all the other operators, which are all irrelevant.\footnote{See section~\ref{sec:Seff} for a more precise discussion, in particular relating to marginal operators.} We then have
\be
\<\cD_\S[e^{-S_\text{eff}}]\>=e^{-S^{\mathbf{1}}_\text{eff}}\<\cD_\S[e^{-S_\text{eff}^\text{rest}}]\>=e^{-S^{\mathbf{1}}_\text{eff}}\p{\<\cD_\S\>+o(1)}.
\ee
Here the $o(1)$ terms go to 0 as $L\to 0$ and are given by conformal perturbation theory in terms of integrated correlation functions of irrelevant operators on $\cD_\S$. For example, the leading contribution is
\be
\half \int d^pz_1 \sqrt{\hat \g} d^pz_2 \sqrt{\hat \g}\l_{\hat \cO}(z_1)\l_{\hat\cO}(z_2) \<\cD_\S[\hat \cO(z_1)\hat \cO(z_2)]\>=O(L^{2\De_\cO-2p}),
\ee 
where $\cO$ is the leading irrelevant operator on $\cD_\S$ (assuming it is Hermitian and unique).

We will only study the implications of the leading cosmological constant term~\eqref{cosmo-const2} in $S_\text{eff}^\mathbf{1}$. That is, we will work with the leading approximation
\be\label{eq:1pt_fusion_approx}
\<\cD_1\cD_2\> &=\exp\left\{-S^{(0)}_\text{eff}+O(L^{-p+1})\right\}
\p{\<\cD_\S\>+O(L^{2\De_\cO-2p})}\nn\\
&=\exp\left\{a_0 \int d^p z\sqrt{\g}\ell^{-p}(z)+O(L^{-p+1})\right\}
\p{\<\cD_\S\>+O(L^{2\De_\cO-2p})}.
\ee

It remains to evaluate the leading term~\eqref{cosmo-const2} in the fusion limit in our kinematics. The relevant scale function $\ell(z)$ was computed in section~\ref{sec:exampleI}, see~\eqref{eq:length-scale}. Let us parameterise the fusion limit as
\begin{align}
	r = e^{-L}, \hspace{1cm}\theta_i= L\Omega_i,
	\label{rescalings}
\end{align}
where $L\to 0$ is now a concrete parameter defined by the above equations. We obtain 
\begin{align}
	S_{\text{eff}}^{(0)} &= -a_0 \int d^p z\sqrt{\g}~\ell^{-p}(z)= -\frac{a_0}{L^p} \int d^p z \sqrt{\g}\frac{1}{\left(1+\sum_{i=1}^m\Omega_{i}^2 x_i^2\right)^{p/2}},
\end{align}
where $x_i$ are the bulk coordinates of the point $z$ on the defect $\md_1$.  In the simplest case, $\Omega_i=0$ and we find
\be\label{eq:Seff0simple}
S_{\text{eff}}^{(0)} &= -a_0 \int d^p z\sqrt{\g}~\ell^{-p}(z)= -\frac{a_0 \vol S^p}{L^p},
\ee
where $\vol S^p$ is the volume of unit $p$-sphere, i.e.\ $\vol S^p=\frac{2\pi^{(p+1)/2}}{\G\left(\frac{p+1}2\right)}$.

For non-zero $\Omega_i$ the integration over $d^pz$ seems to be hard to carry out in full generality, so we focus on the case of $p=1$. In this case $m=\min\{p+1,d-p-1\}=2$ for $d\geq 4$ and $m=1$ in $d=3$. In general, the integral to perform is 
\begin{align}
	S_{\text{eff}}^{(0)} &= -\frac{a_0}{L} \int_0^{2\pi} d\phi \frac{1}{\sqrt{1+\Omega_1^2 \cos^2\phi + \Omega_2^2\sin^2\phi}}= -\frac{4a_0}{L\sqrt{1+\Omega_2^2}}K\left(\frac{-\Omega_1^2+\Omega_2^2}{1+\Omega_2^2}\right),
	\label{EFT-spin-1d-4d}
\end{align}
where $K(z)$ is the complete elliptic integral of the first kind. Note that the right-hand side is symmetric in $\Omega_1,\Omega_2$ despite the appearance.\footnote{This follows from the identity $K(z)=\tfrac{1}{\sqrt{1-z}}K(\tfrac{z}{z-1})$.} In $d=3$ we have to set $\Omega_2=0$ in the above equation.

\subsection{Spinless one-point function density for general defects}

The consistency of the OPE~\eqref{eq:full_two_pt_OPE} and the fusion~\eqref{eq:1pt_fusion_approx} results for the defect two-point function leads to the equation
\be\label{eq:crossing}
\int d\De dJ_1 \cdots dJ_m \r(\De,J_1,\cdots, J_m) r^{\De}e^{-i\sum_i J_{i}\theta_i}\sim\exp\left\{-S^{(0)}_\text{eff}+O(L^{-p+1})\right\}.
\ee
We will now study the implications of this equation, starting with the case $\Omega_i=0$ in~\eqref{rescalings}.

In this case $\theta_i=0$ and the left-hand side of~\eqref{eq:crossing} can be rewritten as
\be
\int d\De \r(\De) r^{\De}
\ee
where 
\be
\r(\De) = \int dJ_1 \cdots dJ_m \r(\De,J_1,\cdots, J_m)
\ee
is the total density of defect one-point functions at scaling dimension $\De$. The right-hand side of~\eqref{eq:crossing} can be evaluated using~\eqref{eq:Seff0simple}, giving
\be
\int d\De \r(\De) r^{\De}=\exp\left\{\frac{a_0 \vol S^p}{(1-r)^p}+O((1-r)^{-p+1})\right\}.
\ee

It is well-known~\cite{Cardy:1986ie, Pappadopulo:2012jk} that such relations constrain the large-$\De$ behaviour of $\r(\De)$. The simple way to extract the density $\rho(\De)$ from the above expression is using an inverse Laplace transform,
\begin{align}
	\r(\De) =\frac1{2\pi i} \int_{y_0-i\infty}^{y_0+i\infty}\,dy \exp\left\{\frac{a_0 \vol S^p}{(1-e^{-y})^p}+O(y^{-p+1})\right\}~e^{\De y},
	\label{density}
\end{align}
where $y_0$ has to be chosen sufficiently large real part so that all the singularities of the integrand are to the left of the integration contour. An approximation to $\r(\De)$ can be obtained under the assumption that for $\De\gg 1$ the integral is dominated by a saddle point at small $y$ where we have an approximation for the integrand.

Working under this assumption, we find the leading-order saddle-point equation at large $\De$ (recall that $a_0\geq 0$)
\begin{align}
	\frac{\ptl}{\ptl y}\frac{a_0 \vol S^p}{y^p} + \De =0
	\hspace{.3cm}\implies\hspace{.3cm}
	y_*\approx \left( \frac{p a_0\vol S^p}{\De}\right)^{\frac1{p+1}},
\end{align}
which yields for the leading one-point function density
\begin{align}\label{eq:spinlessresult}
	\rho(\De)\sim \exp\left[(1+p)\left(a_0\vol S^p\right)^{\frac1{p+1}}\left(\frac{\De}{p}\right)^{\frac{p}{p+1}} \right].
\end{align}
This can be compared with the asymptotic density of states in a $d$-dimensional CFT,
\be
\r_\text{states}(\De) \sim \exp\left[h \De^{\frac{d-1}{d}} \right]
\ee
where $h>0$ can be expressed in terms of free energy density~\cite{Benjamin:2023qsc}. For $p<d-1$ we have $\De^{\frac{p}{p+1}}\ll \De^{\frac{d-1}{d}}$, and thus on average the one-point functions of local operators in the presence of codimension-$q$ defect with $q>1$ have to be very small.

Equation~\eqref{eq:spinlessresult} gives the one-point function density for all operators including descendants. The density of primaries has the same leading term~\eqref{eq:spinlessresult} since we expect the series coefficients in the $r$-expansion of the conformal blocks $G_\cO$ appearing in~\eqref{eq:primaryOPE} to grow only polynomially with degree.\footnote{Note however that we did not prove this in full generality.}

The result~\eqref{eq:spinlessresult} was derived under some assumptions on the behaviour of the inverse Laplace transform. This approach, although not fully justified, can be taken further to compute subleading corrections to~\eqref{eq:spinlessresult} by a more careful evaluation of the saddle point integral and taking into account subleading terms in the fusion effective action.
Alternatively, results of this kind can be obtained rigorously using various Tauberian theorems~\cite{Pappadopulo:2012jk,Qiao:2017xif, Mukhametzhanov:2018zja, Mukhametzhanov:2019pzy}, although the subleading terms are hard to control and depend on the coarse-graining prescription~\cite{Mukhametzhanov:2019pzy}.\footnote{Coarse-graining is necessary to interpret~\eqref{eq:spinlessresult} since the left-hand side is a sum of delta-functions.}

Ignoring these subtleties, let us examine the first subleading term in~\eqref{eq:spinlessresult}. It can come from the saddle-point expansion of the integral or from the subleading terms in the identity coupling in $S_\text{eff}$. The latter start at two-derivative order (the one-derivative term~\eqref{eq:line3doneder} vanishes in these kinematics). Therefore, they are $O(1)$ or smaller for $p\leq 2$. On the other hand, the saddle-point integral will in general produce $\log \De$ terms in the exponent. We thus find the leading corrections for $p=1,2$,
\be
	\rho_{p=1}(\De)&\sim \De^{-3/4}\exp\left[2\sqrt{2\pi a_0}\De^{1/2} \right],\nn\\
	\rho_{p=2}(\De)&\sim \De^{-2/3}\exp\left[3\left(\pi a_0\right)^{1/3}\De^{2/3} \right].
\ee
Further subleading corrections can be straightforwardly incorporated. For $p>2$ the leading corrections come from the 2-derivative terms in the identity coupling.

\subsection{Spin-dependent one-point function density for line defects}

We now consider the more general situation when $\Omega_i\neq 0$ and therefore the expansion~\eqref{eq:crossing} captures information also about the $J_i$-dependence of $\r$. For simplicity, we will focus on line defects ($p=1$). We will assume that $d\geq 4$ (so that $m=2$) and discuss the modifications in the case $d=3$ in the end of this section. 

Using~\eqref{EFT-spin-1d-4d} we obtain the specialised version of~\eqref{eq:crossing}
\be
&\int d\De dJ_1 dJ_2 \r(\De,J_1,J_2) r^{\De}e^{-iJ_1\theta_1-iJ_2\theta_2}=\exp\left\{-S_\text{eff}^{(0)}+O(L^{-p+1})\right\}\nn\\
&\quad=\exp\left\{
\frac{4a_0}{L\sqrt{1+\Omega_2^2}}K\left(\frac{-\Omega_1^2+\Omega_2^2}{1+\Omega_2^2}\right)
+O(L^{-p+1})\right\},
\ee
where $r,\theta_i$ are related to $L,\Omega_i$ by~\eqref{rescalings}. The density $\r(\De,J_1,J_2)$ can be now obtained using inverse Fourier-Laplace transform
\be
\r(\De,J_1,J_2)=\frac1{(2\pi)^3 i} \int d\theta_1 d\theta_2 \int_{y_0-i\infty}^{y_0+i\infty}dy \exp\left\{-S^{(0)}_\text{eff}+O(y^{-p+1})\right\}e^{\De y+i J_1\theta_1+iJ_2\theta_2}.
\ee
We can again obtain a prediction for $\r(\De,J_1,J_2)$ at large $\De$ under the assumption that the above integrals can be computed using a saddle point approximation with a saddle at small $y$ and $\theta_i$.

Introducing the function
\be\label{eq:Fdefn}
F(\Omega_1,\Omega_2) = 
\frac{1}{\sqrt{1+\Omega_2^2}}K\left(\frac{-\Omega_1^2+\Omega_2^2}{1+\Omega_2^2}\right),
\ee
the saddle-point equations become (recall $\theta_i=L\Omega_i\approx y \Omega_i$)
\be	
-\frac{4a_0}{y^2}F(\Omega_1,\Omega_2)+\De+iJ_1 \Omega_1+iJ_2 \Omega_2=0,\quad \frac{4a_0}{y}\frac{\ptl F(\Omega_1,\Omega_2)}{\ptl \Omega_i} + i y J_i=0.
\ee
In other words,
\be
&\qquad\quad y= \sqrt{\frac{4a_0F(\Omega_1,\Omega_2)}{\De+iJ_1 \Omega_1+iJ_2 \Omega_2}},\\
&\frac{\ptl\log F(\Omega_1,\Omega_2)}{\ptl \Omega_i}=-\frac{iJ_i}{\De+iJ_1 \Omega_1+iJ_2 \Omega_2}.\label{eq:O1O2eqn}
\ee
The leading one-point function density is
\be\label{eq:resultDJJ}
\r(\De,J_1,J_2) \sim \exp\left[4\sqrt{a_0(\De+iJ_1 \Omega_1+iJ_2 \Omega_2) F(\Omega_1,\Omega_2)}\right],
\ee
where $\Omega_1$ and $\Omega_2$ are determined from~\eqref{eq:O1O2eqn} with $F$ defined in~\eqref{eq:Fdefn}. The solution for $\Omega_i$ is purely imaginary, and for $J_i\ll \De$ can be expanded as
\be
-i\Omega_1=2j_1 -3j_1(j_1-j_2)(j_1+j_2)+\frac{3}{2}j_1(j_1-j_2)(j_1+j_2)(5j_1^2-j_2^2)+O(j^6_{1,2}),
\ee
where $j_i=J_i/\De$. The expression for $\Omega_2$ is obtained by swapping $1$ and $2$ above. The corresponding expansion for the one-point function density is
\be\label{eq:smallj}
&\log \r(\De,J_1,J_2)\nn\\
&\quad \sim \sqrt{8\pi a_0\De}\p{1-\thalf(j_1^2+j_2^2)-\tfrac{3}{2}j_1^2j_2^2-\tfrac{1}{4}(j_1^2+j_2^2)(j_1^4+5j_1^2j_2^2+j_2^4)+O(j_{1,2}^8)}.
\ee

The above discussion applies to line defects in $d\geq 4$. The case $d=3$ can be obtained by setting $J_2=\Omega_2=0$. In particular, $\Omega=\Omega_1$ is determined by $J=J_1$ via the specialised version of~\eqref{eq:O1O2eqn}
\be\label{eq:O1eqn}
&\frac{\ptl \log K(-\Omega^2)}{\ptl \Omega}=-\frac{iJ}{\De+iJ\Omega}.
\ee 
In this case it is easy to check that the solution $\Omega(j)$ maps $j\in [-1,1]$ to $-i\Omega\in [-1,1]$ monotonically, with $\Omega(\pm 1)=\pm i$.
The leading one-point function density is
\be
\r(\De,J) \sim \exp\left[4\sqrt{a_0(\De+iJ \Omega) K(-\Omega^2)}\right].
\ee
The expansion in small $j=J/\De$ is obtained by setting $j_1=j, j_2=0$ in~\eqref{eq:smallj},
\be
&\log \r(\De,J)\sim \sqrt{8\pi a_0\De}\p{1-\thalf j^2-\tfrac{1}{4}j^6+O(j_{1,2}^8)}.
\ee

\section{Anomalies}
\label{sec:anomalies}

In this section we discuss how presence of Weyl anomalies affects the construction of the effective action $S_\text{eff}$. We show that methods of section~\ref{sec:confgeom} allow a general construction of Weyl-anomaly matching terms.

Note that conformal defects might also have associated gravitational or other anomalies. To simplify the discussion, we assume that there no anomalies are present other than Weyl anomalies. In general, matching of all anomalies has to be carefully taken into account.

\subsection{Weyl anomalies}
\label{sec:weyl}

Generally speaking, the partition functions of a conformal field theory are not necessarily conformally-invariant, but can transform anomalously under Weyl transformation,
\be\label{eq:WeylAnomaly}
	\cZ(e^{2\w}g)=e^{\cA(g,\w)}\cZ(g).
\ee
Here $\cZ(g)$ denotes the partition function (possibly with local operators and defects inserted), and $\cA(g,\w)$ is the Weyl anomaly. The Weyl anomaly has to be a local functional, and is usually specified to the leading order in $\w$. For example, in a 4d bulk CFT, with no operator insertions in the partition function,
\be
	\cA(g,\w)=\int d^4x \sqrt{g}\,\omega\,(-a E_4+ c C_{\mu\nu\rho\alpha}C^{\mu\nu\rho\alpha}+ i\tl c \varepsilon^{\mu\nu\rho\alpha}R_{\mu\nu\gamma\beta}R^{\gamma\beta}{}_{\rho\alpha})+O(\w^2)
\ee
for some anomaly coefficients $a,c,\tl c$ (see~\cite{Duff:1993wm}). When defects are inserted into the partition function, the anomaly splits into the bulk and defect contributions. For example, if we have defects $\cD_1\cdots \cD_n$ inserted, then
\be
	\cA(g,\w)=\cA_{\text{bulk}}(g,\w)+\sum_{i=1}^n \cA_{\cD_i}(g,\w),
\ee
where each $\cA_{\cD_i}$ is a local functional on the defect $\cD_i$. See~\cite{Chalabi:2021jud} for a review of defect Weyl anomalies.

In this section we will address two problems. The first problem is that in the fusion identity~\eqref{eq:mainEFT}
\be\label{eq:EFTweyl}
	\cD_1\cD_2 \sim \cD_\Sigma [e^{-S_\text{eff}}]
\ee
we have different defects on the two sides, and therefore different Weyl anomalies. If $S_\text{eff}$ is Weyl-invariant, the two sides will not in general transform in the same way under Weyl transformations. This means that $S_\text{eff}$ needs to be modified in order to match the Weyl anomalies. We explicitly construct the necessary anomaly-matching terms in section~\ref{sec:anomalymatching}.

The second problem is that, as discussed in section~\ref{sec:transverse}, we in general need to consider transverse-symmetry breaking defects $\cD_\S$. We are not aware of a classification of Weyl anomalies for such defects, and therefore in sections~\ref{sec:anomalyline} and~\ref{sec:anomalysurface} we classify Weyl anomalies for line and surface defects that break transverse rotations to $\SO(q-1)$.

\subsection{Anomaly-matching}
\label{sec:anomalymatching}

First, note that under a Weyl transformation, a partition function involving the left-hand side of~\eqref{eq:EFTweyl} transforms as
\be
	\cZ(e^{2\w}g, \cD_1\cD_2)=e^{\cA_{\text{bulk}}(g,\w)+\cA_{\cD_1}(g,\w)+\cA_{\cD_2}(g,\w)}\cZ(g,\cD_1\cD_2).
\ee
On the other hand, a partition function with the right-hand side transforms as
\be
\cZ(e^{2\w}g, \cD_\S[e^{-S_\text{eff}(e^{2\w}g)}])=e^{\cA_{\text{bulk}}(g,\w)+\cA_{\cD_\S}(g,\w)+S_\text{eff}(g)-S_\text{eff}(e^{2\w}g)}\cZ(g,\cD_\S[e^{-S_\text{eff}(g)}]).
\ee
Here, we explicitly keep the dependence of the effective action on the metric $g$ and we also assumed that $S_\text{eff}(g)-S_\text{eff}(e^{2\w}g)$ only contains the identity operator contribution and thus can be taken out of the partition function. An important comment is in order regarding $\cA_{\cD_\S}(g,\w)$. This anomaly is to be computed for $\cD_\S$, displaced according to the displacement operator coupling as discussed in section~\ref{sec:displacment_and_schemes}.\footnote{There are two ways of seeing this. On the one hand, adding a displacement operator coupling is by definition equivalent to displacing the defect. Therefore, the partition function with the displacement operator coupling turned on and the partition function with the displaced $\cD_\S$ should behave in the same way under Weyl transformations. On the other hand, if we use conformal perturbation theory in the displacement coupling, Weyl anomaly contributions will come from divergences associated with the identity operator appearing in repeated OPEs of $D_\mu$. Consistency of these two points of view leads to relations between correlation functions of $D_\mu$ and defect Weyl anomaly coefficients~\cite{Herzog:2017xha,Herzog:2017kkj}.} Other subtleties in the evaluation of $\cA_{\cD_\S}(g,\w)$ might arise if $S_\text{eff}$ contains operators with special scaling dimensions, see section~\ref{sec:displacment_and_schemes} for a discussion of this.

The Weyl anomaly will be matched if
\be\label{eq:requirement}
	\cA_{\cD_1}(g,\w)+\cA_{\cD_2}(g,\w)-\cA_{\cD_\S}(g,\w)=S_\text{eff}(g)-S_\text{eff}(e^{2\w}g).
\ee

To proceed, we note that the transformation law~\eqref{eq:WeylAnomaly} implies the following addition law for $\cA(g,\w)$,
\be\label{eq:WeylComposition}
	\cA(e^{2\w_2}g,\w_1)=\cA(g,\w_1+\w_2)-\cA(g,\w_2).
\ee
Consider now the term
\be
	\cB(g)=\cA_{\cD_1}(g,-\log \ell(g)),
\ee
where we keep explicitly the dependence of $\ell$ on $g$. From~\eqref{eq:WeylComposition},  this transforms as
\be
	\cB(e^{2\w}g)=\cA_{\cD_1}(e^{2\w}g,-\w-\log \ell(g))=\cB(g)-\cA_{\cD_1}(g,\w).
\ee
Therefore, if we write
\be
	S_\text{eff}=\cA_{\cD_1}(g,-\log \ell)+\cA_{\cD_2}(g,-\log \ell)-\cA_{\cD_\S}(g,-\log \ell)+S^{0}_\text{eff},
\ee
then~\eqref{eq:requirement} implies that $S^{0}_\text{eff}(g)$ is Weyl-invariant and we can use the discussion of the previous sections to construct it. The terms 
\be\label{eq:matching terms}
\cA_{\cD_1}(g,-\log \ell)+\cA_{\cD_2}(g,-\log \ell)-\cA_{\cD_\S}(g,-\log \ell)
\ee
are therefore the necessary Weyl-anomaly matching terms.

Note that in~\eqref{eq:matching terms} different terms evaluate $-\log \ell$ on a different submanifolds -- the Weyl anomaly of $\cD_1$ depends on the Weyl factor on $\cD_1$ and so on. It is therefore crucial for this construction to use the bulk-extended version of $\ell$ that was constructed in section~\ref{sec:dilaton}. When evaluated in terms of data on $\cD_\S$ (as all other terms in the effective action),~\eqref{eq:matching terms} becomes an infinite derivative expansion, suppressed by powers of $L$.

Finally, let us briefly discuss how one can reconstruct $\cA(g,\w)$ from the its leading term in $\w$, which is the term that is usually specified. Suppose
\be
	\cA(g,\w)=\cA^{(1)}(g,\w)+O(\w^2),
\ee
where $\cA^{(1)}(g,\w)$ is linear in $\w$. First, we note that for any Weyl transformation $\tl \w$,
\be
	\cA^{(1)}(e^{2\tl\w}g,\w)=\cA^{(1)}(g,\w)+\de\cA^{(1)}(g,\w;\ptl \tl \w),
\ee
where $\de\cA^{(1)}(g,\w;\ptl \tl \w)$ is a finite-order local polynomial in $\tl\w$, and furthermore $\tl\w$ only enters through its derivatives. This is because $\cA^{(1)}(g,\w)$ is scale-invariant and $\de\cA^{(1)}$ must vanish for constant $\tl\w$. But then $\de\cA^{(1)}$ only depends on derivatives of $\tl\w$, and only finitely many derivatives can enter by dimensional analysis. For example, for a line defect at most $\ptl \tl\w$ can appear and therefore $\de\cA^{(1)}(g,\w;\ptl \tl \w)$ is linear in $\ptl \tl\w$. On the other hand, for a surface defect we can have $\ptl\tl\w$, $\ptl^2\tl\w$, and $(\ptl\tl\w)^2$. 

It remains to note that the composition law~\eqref{eq:WeylComposition} implies that
\be
	\ptl_t \cA(g, t\w)=\cA^{(1)}(e^{2t\w}g,\w)=\cA^{(1)}(g,\w)+\de\cA^{(1)}(g,\w;t\ptl \w).
\ee
This implies
\be
	\cA(g, \w)=\cA^{(1)}(g,\w)+\int_0^1 dt \de\cA^{(1)}(g,\w;t\ptl \w).
\ee
Since the integrand in the right-hand side is an explicit polynomial in $t$ that is easily computable from $\cA^{(1)}$, this expression allows one to compute $\cA(g,\w)$.

\subsection{Line defects}
\label{sec:anomalyline}

In this section we classify Weyl anomalies for line defects $\cD$ that break transverse rotations down to $\SO(q-1)$. We find a non-trivial anomaly in $d=3$ but no anomalies in $d\geq 4$.

 We follow the standard strategy (see, for example,~\cite{Chalabi:2021jud}), where we first the most general local ansatz for $\cA^{(1)}_\cD$, impose Wess-Zumino consistency conditions, and finally isolate the scheme-invariant part of the anomaly. We first consider $d=3$ and generalise to $d\geq 4$ later.

Recall from section~\ref{sec:transverse} that the breaking of transverse rotations is parameterised by a unit normal vector field $n^\mu$ of scaling dimension $0$. This vector field allows us to construct new terms in the ansatz for $\cA^{(1)}_\cD$, making the classification problem different from the case of transverse rotation-preserving defects.

For convenience, we introduce a unit vector field $\tau$ tangent to the defect. Note that in this section we work with the physical metric $g_{\mu\nu}$. Since $\tau$ is a unit vector, it also has dimension $0$. Using $\tau$ and $n$, we can construct a unit vector $m$ orthogonal to both of them,
\begin{align}
	m^\mu=\varepsilon^{\mu}{}_{\nu\rho} n^\nu \tau^\r,
\end{align}

The most general form of the Weyl anomaly is
\begin{align}
	\cA^{(1)}_\cD(g,\w)=\int \,dz \sqrt{\gamma} f,
\end{align}
where $f$ is a local expression of mass dimension 1. We first construct the most general ansatz for $f$ using the standard curvature invariants, $n,\tau, m$, and covariant derivatives. Only linearly independent contributions to $\cA^{(1)}_\cD$ should be included. In particular, in $f$ we should not include total derivatives. This leads to the following ansatz,
\begin{align}
	\cA^{(1)}_\cD(g,\w)=&\int d z \sqrt{\gamma} \w\p{c_1 n_\mu \II^\mu+c_2 m_\mu \II^\mu+ic_3 \tau^a m_\mu\nabla^\perp_a n^\mu}\nn\\
	&+\int d z \sqrt{\gamma} \ptl_\mu \w\p{c_4n^\mu+c_5m^\mu}.
\end{align}

The Weyl variation of $\cA^{(1)}_\cD(g,\w_1)$ is given by
\begin{align}
	\delta_{\w_2}\cA^{(1)}_\cD(g,\w_1)=-\int d z \sqrt{\gamma} \p{c_1 n^\mu \w_1 \partial_\mu \w_2+c_2 m^\mu \w_1 \partial_\mu \w_2}.
\end{align}
Wess-Zumino consistency condition requires $\delta_{\w_1}\cA^{(1)}_\cD(g,\w_2)=\delta_{\w_2}\cA^{(1)}_\cD(g,\w_1)$, which implies $c_1=c_2=0$. Therefore, the most general form for the anomaly allowed by Wess-Zumino consistency condition is
\begin{align}
	\cA^{(1)}_\cD(g,\w)=\int d z \sqrt{\gamma} \p{ic_3\w\tau^a m_\mu{\nabla}^\perp_a n^\mu+\partial_\mu\w(c_4n^\mu+c_5m^\mu)}.
\end{align}
The coefficients $c_4$ and $c_5$ can be set to $0$ using local counter-terms $n\.\II$ and $m\.\II$. Therefore, the final form of the anomaly is
\begin{align}
	\cA^{(1)}_\cD(g,\w)&=i\int d z \sqrt{\gamma} c_3\w\tau^a m_\mu{\nabla}^\perp_a n^\mu.
\end{align}
Note that for a constant $\w$ this reduces to $2\pi i c_3 I(n)$ evaluated in the physical metric, where $I(n)$ is discussed in section~\ref{sec:subleading}.

The only essential modification in $d>3$ is that the terms that involve $m^\mu$ do not exist there. The result is that no defect Weyl anomaly is possible in that case.

\subsection{Surface defects in $d=4$}
\label{sec:anomalysurface}

We now consider surface defects in $d=4$. We once again assume that this defect breaks rotational symmetry along the defect, and that this breaking can be parameterised by a normal vector field $n^\mu$. As the defect is now two-dimensional, we can choose an orthonormal basis in its tangent space at every point, parameterised by vector fields $\tau^a$ and $\sigma^a$. This choice of basis is however arbitrary and unphysical, and therefore the terms of the Weyl anomaly should not depend on this choice. The easiest way to impose such a requirement, is to use complex coordinates
\begin{align}
	\z^a=\frac{\tau^a+i\sigma^a}{\sqrt2},\qquad\bar{\z}^a=\frac{\tau^a-i\sigma^a}{\sqrt2}
\end{align}
for the tangent bundle of the defect. Rotations of these basis vectors map $\z^a\to e^{i\theta}\z^a$, $\bar \z^a\to e^{-i\theta}\bar \z^a$, and so we can easily see how terms transform under these rotations.

As before, we also introduce a final basis vector normal to the defect,
\begin{align}
	m^\mu=\varepsilon^\mu{}_{\nu\rho\sigma} n^\nu \tau^\rho s^\sigma,
\end{align}
and we then can use
\begin{align}
	\chi^\mu=\frac{m^\mu+in^\mu}{\sqrt2},\qquad\bar{\chi}^\mu=\frac{m^\mu-in^\mu}{\sqrt2}
\end{align}
as complex coordinates on the defect. Note that $m$ is invariant under rotations of $\tau,\s$.  Now in each tangent space we have a set of coordinates defined by $\z,\bar\z,\chi,\bar\chi$, and we can use these coordinates to describe the components of various tensors, e.g.\ $\II^\chi_{\z\z}$.

We then proceed as in the case of the line defect. As compared with the last section, we shall give fewer details of the calculations involved in this case. This is because the methods are almost identical, but the list of terms (given in appendix~\ref{app:ansatz}) we have to consider is much longer. This list was once again obtained by writing down a list of all possible terms with two derivatives, and removing any terms that are not linearly independent. In order to remove such linear dependence, the Gauss, Codazzi and Ricci identities have to be considered. 

After applying Wess-Zumino consistency conditions taking local counterterms into account, we find that the scheme-independent part of the defect Weyl anomaly is
\be
	\cA^{(1)}_{\cD}(g,\w)=&\int d^2z \sqrt{\gamma}\w\bigg(
	d_1\bar{\chi}_\mu (\nabla^\perp)^2 \chi^{\mu}
	+\bar d_1\chi_\mu (\nabla^\perp)^2 \bar{\chi}^{\mu}
	+d_2\mathring\II_{\chi\z\z}\mathring\II_{\chi\bar{\z}\bar{\z}}
	+\bar d_2\mathring\II_{\bar{\chi}\z\z}\mathring\II_{\bar{\chi}\bar{\z}\bar{\z}}\nn\\
	&\hspace{2.2cm}
	+e_1\mathring\II_{\chi\z\z}\mathring\II_{\bar{\chi}\bar{\z}\bar{\z}}
	+\bar e_1\mathring\II_{\bar{\chi}\z\z}\mathring\II_{\chi\bar{\z}\bar{\z}}
	+e_2C_{\chi\bar{\chi}\chi\bar{\chi}}
	+ie_3C_{\chi\bar{\chi}\z\bar{\z}}+c_\cD R_\g \bigg)\nn\\
	&+i\int d^2z \sqrt{\gamma}d_5(\II_\chi \partial_\chi \w-\II_{\bar{\chi}} \partial_{\bar{\chi}}\w),
\ee
where $R_\g$ is the Ricci scalar of the defect metric. The coefficients $e_2,e_3,c_\cD,d_5$, are real.\footnote{Note that complex conjugation plus orientation reversal exchanges $w$ and $\bar w$ and does nothing to $\z$ and $\bar \z$.} The terms with coefficients $c_\cD$ and $e_i$ do not break transverse rotations and are already known \cite{Chalabi:2021jud, Graham:1999pm, Schwimmer:2008yh, Cvitan:2015ana, Jensen:2018rxu}, whereas the terms with coefficients $d_i$ do break transverse rotations, and are new. One point to note is that the coefficients $d_i$ and $e_i$ can generically be position-dependent, but by Wess-Zumino consistency conditions, $c_\cD$ must be position-independent and thus independent of marginal couplings. All of the terms are B-type (i.e.\ they transform trivially under Weyl transformations) except from $c_\cD R_\g$, which is A-type.

\section{Fusion and OPE}
\label{sec:fusionandOPE}

In this section we discuss the analogies and differences between the fusion of conformal defects and the OPE of local operators. Formally, a local operator can be viewed as a 0-dimensional defect, which allows many questions to be translated between the two situations. As everywhere in the paper, in this section ``conformal defect'' refers to a conformal defect with $p>0$, unless explicitly stated otherwise.

We first discuss the differences in how one labels conformal defects and primary operators. To a primary local operator $\cO$ one associates a set of quantum numbers, which include its scaling dimension, spin representation, and representations under various global symmetry groups. By contrast, no such assignment is made in the case of conformal defects.

Consider first the spin and global symmetry representations. These describe the action of various symmetry groups on the vector space spanned by different components of the local operator $\cO$. In other words, a given component of $\cO$ transforms non-trivially under these symmetries. This of course can also happen for conformal defects, which can break global symmetries and transverse rotations. A crucial difference is, however, that the local operators form a vector space, and therefore organise into a linear representation of the symmetry. This representation can then be decomposed into irreducible components, which is what allows us to label local operators by irreps. By contrast, as discussed in section~\ref{sec:simplicity}, locality requirements mean that conformal defects do not form a natural vector space, and therefore no such decomposition is possible.\footnote{In this sense the ``spinning conformal defects'' of~\cite{Guha:2018snh, Kobayashi:2018okw} are necessarily non-local and are not the ones considered in this paper.}

The scaling dimensions of local operators also fall under the above discussion since they determine the representation of the dilatation subgroup. However, in this case a different point of view is possible. Specifically, one can interpret the scaling dimension of a local operator as the Weyl anomaly for a 0-dimensional conformal defect. For higher-dimensional defects the classification of Weyl anomalies is different~\cite{Chalabi:2021jud}, and it is therefore not surprising that there isn't a direct analogue of scaling dimensions for all conformal defects.

From this point of view, the leading-order term in the local operator product expansion can be seen as a Weyl anomaly-matching term in the effective action. For example, consider the contribution of a scalar local operator $\cO$ of dimension $\De$ to the OPE of two scalar operators $\f$ of dimension $\De_\f$,
\be
	\f(x)\f(0)\ni |x|^{\De-2\De_\f}\cO(0)+\cdots=e^{(\De-2\De_\f)\log|x|}\cO(0)+\cdots.
\ee
This can be written using a notation similar to the one we used for conformal defects,
\be\label{eq:OPEleading}
	\f(v)\f(0)=e^{(\De-2\De_\f)\log \ell}\cO(0)+\cdots.
\ee
Here, we parameterise the position of one of the operators using a tangent vector $v$ at $0$, and $\ell^2 = g_{\mu\nu}v^\mu v^\nu$. Since on a 0-dimensional defect integration is trivial, this can be viewed the leading-order analogue of equation~\eqref{eq:matching terms} for conformal defects.  From section~\ref{sec:anomalymatching} we know that this leading-order term fails to completely match the anomaly, since $\f(v)$ and $\f(0)$ transform with the Weyl factor evaluated at different points.

In fact, it impossible to complete~\eqref{eq:OPEleading} into an equation that matches the Weyl anomaly to all orders in $L$ without including some analogue of non-trivial operators in $S_\text{eff}$. Indeed, if it were possible, the OPE $\f(v)\f(0)$ would be conformally-invariant even with only the primary $\cO(0)$ is included, but it is well-known that inclusion of descendants is generally necessary. The anomaly-matching terms constructed in section~\ref{sec:anomalies} do not exist here because the differential equation~\eqref{eq:surface_evolution} and the invariant~\eqref{eq:qdefn} are not defined for $p=0$.\footnote{For the same reason, the conformally-invariant definition of $v$ in section~\ref{sec:displacement} does not work for $0$-dimensional defects. In fact it is not hard to show that no conformally-invariant definition of displacement vector $v$ is possible for $p=0$: the group of conformal transformations that preserve two points in $\R^d$ is too large -- it can rotate and rescale any tangent vector $v$ without moving the insertion points of the local operators. Therefore, no tangent vector can parameterise the relative position of two points in a conformally-invariant way.\label{footnote:nop=0v}}

To fix this problem, we should include the descendants of $\cO$. It is natural to treat these descendants as being analogous to insertions of a displacement operator on a 0-dimensional defect $\cD$,
\be
	\<\ptl_{\mu_1}\cdots \ptl_{\mu_n}\cO\cdots \>\sim \<\cD[D_{\mu_1}\cdots D_{\mu_n}]\cdots\>.
\ee
The $n$-point correlation function of $D_\mu$ are subject to scheme ambiguities which affect its definition at coincident points (c.f.\ section~\ref{sec:displacment_and_schemes}). In a generic situation, the possible modifications are given by terms including derivatives of delta-functions and lower-point correlation functions of $D_{\mu}$. In the case of 0-dimensional defects, the only possible configuration of $n$ points is a coincident point configuration, which in terms of $\cO$ means that we can correct the $n$\textsuperscript{th} order derivatives of $\cO$ by its lower-order derivatives and interpret this as a scheme change. For instance, we have the schemes
\be
	\<\cD[D_{\mu_1}\cdots D_{\mu_n}]\cdots\>&=\<\ptl_{\mu_1}\cdots \ptl_{\mu_n}\cO\cdots \>,\\
	\<\cD[D_{\mu_1}\cdots D_{\mu_n}]\cdots\>&=\<\nabla_{\mu_1}\cdots \nabla_{\mu_n}\cO\cdots \>.
\ee
They differ by subleading derivatives of $\cO$ times derivatives of the metric, and only the latter scheme is generally-covariant. As mentioned in footnote~\ref{footnote:nop=0v}, no Weyl-invariant schemes exist in this case.

For conformal defects, we would normally include the displacement operator into the effective action as
\be
	\cD[e^{-\int d^pz \sqrt{\g}\l^\mu D_\mu}],
\ee
where locality requires us to use the exponential of a local integral. This is necessary in order to obtain proper factorisation properties for path integrals. In the case of $0$-dimensional defects no such requirement exists (a 0-dimensional defect cannot be cut in half) and we can therefore modify the partition functions in a more general way,
\be
	\cD[1+\l^\mu_1 D_\mu + \thalf \l^{\mu\nu}_2D_\mu D_\nu + \cdots]\sim \cO+\l^\mu_1 \nabla_\mu \cO + \thalf \l^{\mu\nu}_2\nabla_\mu \nabla_\nu\cO + \cdots,
\ee
where the coefficients $\l_1,\l_2,\cdots$ are all independent. They are constrained by Weyl invariance (which in flat space fixes them completely), but not by a requirement to form an exponential series.

To summarise, we find that there are two key differences between the fusion of local operators and of general conformal defects. Firstly, the conformal geometry of a pair of conformal defects is, in a sense, simpler than that for a pair of local operators, in the sense that the Weyl-invariant parameterisation of section~\ref{sec:displacement} and the scale function of section~\ref{sec:dilaton} exist. Due to this, Weyl-invariance appears to impose fewer constraints on fusion of conformal defects than on local operators. Secondly, the locality constraints are very different in the two cases, giving more constraints in the case of conformal defects. These two effects imply that the two operations have quite different properties despite naively being very similar.

\section{Discussion}
\label{sec:discussion}

In this paper we studied effective field theory for fusion of conformal defects. We have described how such an effective field theory can be written down to any order in derivative expansion, fully taking into account the diffeomorphism- and Weyl-invariance constraints. We have also explained how the terms in the effective action can be evaluated for setups with simple kinematics. Finally, we have studied simple applications of this effective field theory: to the cusp anomalous dimensions, and to the asymptotics of bulk one-point functions.

As we explained, it is possible that transverse-symmetry breaking conformal defects might arise as the result of the fusion. We found in section~\ref{sec:anomalies} that such line defects in $d=3$ can have non-trivial Weyl anomalies, and that new Weyl anomaly terms arise for surface defects in $d=4$.  It would be interesting to find and study examples of such symmetry-breaking defects.

Furthermore, it would be interesting to extend the study of cusp anomalous dimensions from section~\ref{sec:Wilson}. For example, it is important to determine whether any non-identity operators contribute to the planar $\G_\text{cusp}$ and if yes, how is it compatible with~\eqref{eq:Hexact}. Computation of subleading Wilson coefficients in the effective action using integrability techniques is another open direction.

The geometric methods that we developed in section~\ref{sec:confgeom} allow an easy extension to more exotic fusion setups. For example, one can consider fusion of defects of different dimensions. Alternatively, defects of the same dimension can approach each other in more general configurations, so that in the limit the have a lower-dimensional intersection. We believe these questions can be straightforwardly addressed using our approach from section~\ref{sec:confgeom}. On the other hand, limits such as the small-radius limit of a cylindrical surface defect would require development of new tools.

Finally, an important question is to understand whether in the setup of section~\ref{sec:operatorasymptotics} defect fusion can be replaced by a convergent expansion of some sort. Doing so would be important for enabling the application of numerical bootstrap techniques to conformal defects. Fusion effective theory gives only an asymptotic expansion. While being sufficient for the derivation of asymptotic relations of section~\ref{sec:operatorasymptotics}, it falls short of providing a quantitative handle on the low-energy part of the spectrum in the crossed channel.

\section*{Acknowledgments}

We would like to thank Sean Curry, Nadav Drukker, Nikolay Gromov, Christopher Herzog, Zohar Komargodski and David Simmons-Duffin for discussions. We especially thank Nikolay Gromov for many helpful comments on the draft and for providing us with numerical data. The work of PK was supported by the UK Engineering and Physical Sciences Research council grant number EP/X042618/1; and the Science Technology \& Facilities council grant number ST/X000753/1. RS is supported by the Royal Society-Newton International Fellowship NIF/R1/221054-Royal Society.

\appendix

\section{Geometry and conventions}
\label{app:geometry}

Given an embedding $X:N\to M$ of a defect, we can construct a map $TN\to TM$ as $e_a^\mu\equiv\partial_a X^\mu$. We shall use $\nabla$ as notation for the Levi-Civita connections on both $M$ and $N$, and we shall use $\nabla^\perp$ for the connection in the normal bundle on $N$.

When there is no risk of ambiguity we shall often commit an abuse of notation, and give defect tensors bulk indices using $e^\mu_a$, for example defining $t^\mu\equiv e^\mu_a t^a$. We can additionally define the second fundamental form of the defect $\II^\mu_{ab}\equiv\nabla_a e^\mu_b$, and as this is symmetric in $a$ and $b$, it can be split into a trace part $\II^\mu\equiv\II^\mu_{ab}\g^{ab}$, and a traceless symmetric part $\mathring\II^\mu_{ab}\equiv\II^\mu_{ab}-\frac{1}{p}\II^\mu \g_{ab}$. 

As always, given the Levi-Civita connection on either $M$ or $N$, we can define a Riemann tensor $R^\rho{}_{\sigma\mu\nu}$, a Ricci tensor $R_{\mu\nu}$, and a Ricci scalar $R$. We shall additionally need the bulk Schouten tensor
\be
	P_{\mu\nu}\equiv\frac{1}{d-2}\p{R_{\mu\nu}-\frac{1}{2(d-1)}Rg_{\mu\nu}},\label{eq:schoutendef}
\ee
the bulk Weyl tensor
\be\label{eq:RiemannDecomposition}
	C_{\mu\nu\rho\sigma}\equiv R_{\mu\nu\rho\sigma}-P_{\mu\rho}g_{\nu\sigma}+P_{\nu\rho}g_{\mu\sigma}-P_{\nu\sigma}g_{\mu\rho}+P_{\mu\sigma}g_{\nu\rho}.
\ee

We define the Weyl weight $n$ of a quantity $A$ by the requirement that under a Weyl rescaling $g_{\mu\nu}\to e^{2\w}g_{\mu\nu}$ the quantity $A$ transforms as $A\to e^{n\w}A$. Note that if $A$ is a tensor object with $i$ contravariant and $j$ covariant indices, $A=A^{\mu_1\mu_2\dots \mu_i}_{\nu_1\nu_2\dots \nu_j}$, then the relation between its Weyl weight $n$ and what is usually meant by the scaling dimension $\De$ is
\be\label{eq:dimweylrelation}
	\De=-n+j-i.
\ee
The shift between $\De$ and $-n$ appears because the scaling dimension determines the behaviour under conformal transformations, which are coordinate transformations and act on the tensor indices. For example, the metric $g_{\mu\nu}$ has Weyl weight $n=2$, but its scaling dimension is $\De=-2+2=0$. It is often advantageous to work with scaling dimensions because they remain invariant when the indices are raised, lowered or contracted. Relatedly, the primary operators transform under a Weyl rescaling $g_{\mu\nu}\to e^{2\w}g_{\mu\nu}$ as
\be
	\delta_\w\mathcal{O}^{\mu_1\mu_2\dots \mu_i}_{\nu_1\nu_2\dots \nu_j}(x)=(-\Delta+j-i)\w(x) \mathcal{O}^{\mu_1\mu_2\dots \mu_i}_{\nu_1\nu_2\dots \nu_j}(x).
\ee

\subsection{Relations between curvature invariants}

When one considers a submanifold of a Riemannian manifold, in addition to the bulk Riemann tensor $R_{\mu\nu}{}^\r{}_\s$ it is also possible to define the intrinsic submanifold Riemann tensor $\check{R}_{ab}{}^c{}_d$ as well as the curvature tensor for the normal connection $R^\perp_{ab}{}^{\mu}{}_\nu$. Together with the second fundamental form $\II^\mu_{ab}$, these objects are not independent and are related by the Gauss, Codazzi-Mainardi, and Ricci equations. These equations (and their derivatives) are the only non-trivial equations between these curvature invariants. See~\cite{SpivakIV} for details.

The Gauss equation states
\be\label{eq:gauss}
	R_{ab}{}^c{}_d&=\check{R}_{ab}{}^c{}_d-2\II_{\mu[a}{}^c\II^\mu_{b]d}.
\ee
The Codazzi-Mainardi equation is
\be\label{eq:codazzi-mainardi}
	R_{ab}{}^\mu{}_c=\nabla_a \II^\mu_{bc}-\nabla_b \II^\mu_{ac}.\quad (\mu\text{ normal})
\ee
Finally, the Ricci equation reads
\be\label{eq:ricci}
	R_{ab}{}^\r{}_\s=R_{ab}^\perp{}^\r{}_\s-\II^\mu_{ac}\II_{\nu b}{}^c+\II^\mu_{bc}\II_{\nu a}{}^c.\quad(\r,\s\text{ normal})
\ee

The Ricci equation~\eqref{eq:ricci} essentially means that we do not have to consider $R^\perp$ as a separate invariant since it can be expressed in terms of $R$ and $\II$. The Codazzi-Mainardi equation~\eqref{eq:codazzi-mainardi} are mostly irrelevant for this paper since we do not ever have to consider covariant derivatives of $\II$. On the other hand, the Gauss equation~\eqref{eq:gauss} does reduce the number of independent two-derivative terms in certain situations.

It is only the fully-contracted Gauss equation that can appear in our analysis. This gives
\be
	R_{ab}{}^{ab}&=\check{R}-\II_{\mu}\II^\mu+\II_{\mu b}{}^a\II^\mu_{a}{}^b.
\ee
Notice that the left-hand side has the bulk Riemann tensor contracted only along the defect directions. Rearranging the sums, we can rewrite this as~\cite{Chalabi:2021jud}
\be
	R&=\check{R}-\II_{\mu}\II^\mu+\II_{\mu ab}\II^{\mu ab}+2R^j_j-R_{jk}{}^{jk},
\ee
where the contractions over $j,k$ are taken in the normal directions.  

We will also need the following special case of this identity. Assuming that normal components of Schouten tensor vanish as in~\eqref{eq:P=0}, and that $\hat\II^\mu=0$ as in~\eqref{eq:II=0} we then find in the fusion metric (using also~\eqref{eq:RiemannDecomposition})
\be
	\frac{d-1-q}{d-1}\hat R&=\hat{\check{R}}+\hat \II_{\mu ab}\hat \II^{\mu ab}-C_{jk}{}^{jk}.
\ee
For codimension $q=1$ the contraction of the Weyl tensor vanishes by anti-symmetry and we get
\be\label{eq:codim1gaussfusion}
	\frac{d-2}{d-1}\hat R&=\hat {\check{R}}+\hat\II_{\mu ab}\hat \II^{\mu ab}.
\ee
For codimension $q=2$ the normal contraction of the Weyl tensor is equal to
\be
	C_{jk}{}^{jk}=2C_{vuvu},
\ee
where $\{v,u\}$ is any orthonormal basis in the normal bundle. Therefore, for $q=2$ we find
\be\label{eq:codim2gaussfusion}
	\frac{d-3}{d-1}\hat R&=\hat {\check{R}}+\hat\II_{\mu ab}\hat \II^{\mu ab}+2C_{vuvu}.
\ee

\section{Conformal circle equation}

\label{app:CCequation}

In this appendix we rewrite the conformal circle in a form needed in section~\ref{sec:dilaton}. The conformal circle equation is
\be
\nabla_t a^\mu = \frac{3 u\.a}{u^2}a^\mu-\frac{3a^2}{2u^2}u^\mu+u^2u^\nu P^\mu_\nu-2P_{\a\b}u^\a u^\b U^\mu.
\ee
Define 
\be
	f&=\frac{u\.a}{u^2},\quad h = \frac{a^2}{u^2}.
\ee
We calculate
\be
	\nabla_t f &= h-2f^2+\frac{u\.\nabla_t a}{u^2}=-\tfrac{1}{2}h+f^2-u^\mu u^\nu P_{\mu\nu},\\
	\nabla_t h &= -2fh+2\frac{a\.\nabla_t a}{u^2}=fh+2a^\mu u^\nu P_{\mu\nu}-4P_{\a\b}u^\a u^\b f.
\ee
Altogether, we can write this as a system
\be
	\nabla_t \g^\mu &= u^\mu,\\
	\nabla_t u^\mu &= a^\mu,\\
	\nabla_t a^\mu &= 3fa^\mu-\tfrac{3}{2}hu^\mu+u^2u^\nu P^\mu_\nu-2P_{\a\b}u^\a u^\b u^\mu,\\
	\nabla_t f &=-\tfrac{1}{2}h+f^2-u^\mu u^\nu P_{\mu\nu},\\
	\nabla_t h &=fh+2a^\mu u^\nu P_{\mu\nu}-4P_{\a\b}u^\a u^\b f,
\ee
where the right-hand sides are now smooth functions of $\g,u,a,f,h$.

\section{Two-derivative identity operator couplings}
\label{app:two-derivatives}

In this section we classify two-derivative terms in the identity operator coupling. We first consider line defects, and then the codimension-1 and codimension-2 defects which are \textit{not} line defects.

\subsection{Line defects}
\label{app:lines}
Line defects have no intrinsic curvature invariants, and~\eqref{eq:II=0} implies that no extrinsic curvature invariants are available either (since the second fundamental form only has the trace component). Therefore, in building the effective action we can only use bulk curvature invariants and the derivatives of $v^\mu$. For bulk curvatures, it is convenient to decompose the Riemann tensor into Weyl tensor and Schouten tensor~\eqref{eq:schoutendef}: Weyl tensor is Weyl-invariant, while Schouten tensor is constrained by~\eqref{eq:P=0}. Note that the trace of $\hat P$ is proportional to the Ricci scalar.

\paragraph{$d=3$ bulk}
Recall that in section~\ref{sec:subleading} we defined $u^\mu = \hat \ve^\mu{}_{\r\s}v^\mu t^\s$ where $t^\mu$ is the unit vector (for fusion metric) tangent to the line defect. 

We have only two ways to form 2-derivative invariants. The first is to use bulk curvature tensors, of which we only have the Schouten tensor $P$ available (Weyl tensor vanishes in $d=3$). Due to~\eqref{eq:P=0} the only non-zero components are
\be
	\hat P_{tt},\quad \hat P_{tv},\quad \hat P_{tu},
\ee
of which $\hat P_{tt}$ can be expressed in terms of the Ricci scalar $\hat R$ due to $\hat P_{vv}=\hat P_{uu}=0$.

The second option is to take covariant derivatives of $v$. Up to total derivatives, the only two-derivative term is
\be
	\hat\nabla_t^\perp v\. \hat\nabla_t^\perp v=(u\.\nabla_t^\perp v)^2.
\ee

Overall, the 2-derivative contribution to the identity operator part of the effective action for line defects in $d=3$ is
\be
S^{(2)}_\text{eff} = &\int dz\sqrt{\hga}\p{
	a_{2,1}\hat\nabla_t^\perp v\. \hat\nabla_t^\perp v+a_{2,2}\hat R+ia_{2,3}\hat P_{tv}+ia_{2,4}\hat P_{tu}
}.
\label{1d-3d-2nd}
\ee
The Schouten tensor in the above formula can be replaced with the Ricci tensor at the cost of rescaling the Wilson coefficients.

\paragraph{$d\geq 4$ bulk}
Compared to the $d=3$ case the only difference comes from the treatment of bulk curvature invariants. Firstly, for Schouten tensor, the only available components are $\hat P_{tt}$ and $\hat P_{vt}$, but $\hat P_{tt}$ can again be expressed in terms of $\hat R$. 

Secondly, Weyl tensor is now non-zero. If $d>4$ then the only way we can set indices is $C_{vtvt}$ due to the anti-symmetry properties. In $d=4$ we can also form $\tl C_{vtvt}$ where
\be\label{eq:WeylDualTensor}
	\tl C_{\mu\nu\s\r}=\thalf\hat \e_{\mu\nu}{}^{\l\kappa}C_{\lambda\kappa\s\r}.
\ee

Overall, the 2-derivative contribution to the identity operator part of the effective action for line defects in $d\geq 4$ is
\be
S^{(2)}_\text{eff} = &\int dz\sqrt{\hga}\p{
	a_{2,1}\hat\nabla_t^\perp v\. \hat\nabla_t^\perp v+a_{2,2}\hat R+ia_{2,3}\hat P_{tv}+a_{2,4}C_{vtvt}+ia_{2,5}\tl C_{vtvt}
},
\label{1d-4plusd-2nd}
\ee
where the $\tl C_{vtvt}$ has to be omitted for $d>4$. Note that the terms containing the Weyl tensor are zero when the physical metric is conformally-flat. Again, the Schouten tensor in the above formula can be replaced with the Ricci tensor at the cost of rescaling the corresponding Wilson coefficient.

\subsection{Codimension $q=1$ defects ($p>1$)}

In the case of $q=1$ the vector field $v$ is simply the unit normal to $\cD_1$ and thus its covariant derivatives vanish. If the defects $\cD_1$ or $\cD_2$ carry a normal bundle orientation, it might be possible to define a normal unit vector field $n$ on $\cD_1$ using these orientations. In this case, we can form $n\.v=\pm 1$ and all Wilson coefficients may depend on this sign.

The independent components of the Schouten tensor that we can use at 2-derivative order are (taking into account~\eqref{eq:P=0})
\be
	\hat P_{va},\quad \hat P_{ab},
\ee
where as always $\hat P_{va}=v^\mu P_{\mu a}$ and $a,b$ are indices along the defect. However, we do not have any $0$-derivative object to contract the open indices with. For similar reasons, in $d\geq 4$ we cannot form any terms from the Weyl tensor $C$ or (for $d=4$) its dual $\tl C$ defined in~\eqref{eq:WeylDualTensor}, or by using the first order covariant derivatives of $\hp^\mu_{ab}$.

The two-derivative terms therefore have to be constructed from the bulk and defect Ricci scalars $\hat R$ and $\hat{\check{R}}$ or from two copies of $\hp^\mu_{ab}$. Keeping in mind~\eqref{eq:codim1gaussfusion}, we then find the following 2-derivative contribution to the identity operator part of the effective action for codimension $q=1$ defects with $p>1$,
\be
S^{(2)}_\text{eff}=\int d^pz\sqrt{\hga}\p{a_{2,1} \hat R + a_{2,2}\hat{\check{R}}}.
\ee
Note that when $p=2$, the intrinsic Ricci scalar $\hat{\check R}$ is (locally) a total derivative and the corresponding term captures only the topological data about the defect.

\subsection{Codimension $q=2$ defects ($p>1$)}
When $q=2$ we can define a normal vector $u^\mu=\hat\e^{\mu}{}_{\nu}v^\nu$ using the normal bundle Levi-Civita tensor $\hat \e^{\mu}{}_\nu$. The vectors $v$ and $u$ now span the normal bundle. Compared to $q=1$ this allows us to define more terms using the second fundamental form,
\be
	(v\.\hp_{ab})(v\.\hp^{ab}), \quad (u\.\hp_{ab})(u\.\hp^{ab}), \quad (u\.\hp_{ab})(v\.\hp^{ab}), \quad (u\.\hp_{ab})(v\.\hp^{a}{}_c)\hat \e^{bc},
\ee
where in the last term $\hat \e^{bc}$ is the defect Levi-Civita tensor which is available only if $d=4$ and $p=2$. Equation~\eqref{eq:codim2gaussfusion} allows us to eliminate $(u\.\hp_{ab})(u\.\hp^{ab})$ using
\be
	(v\.\hp_{ab})(u\.\hp^{ab})+(v\.\hp_{ab})(u\.\hp^{ab})=\hp_{\mu ab}\hp^{\mu ab}.
\ee

Furthermore, terms involving Weyl tensor are now available. In $d=4$ we can find two independent components (any other suitable components can be expressed in terms of these two)
\be
	C_{vuvu},\quad \tl C_{vuvu},
\ee
where $\tl C$ is defined in~\eqref{eq:WeylDualTensor}. For $d>4$ only $C_{vuvu}$ is available.

Finally, derivatives of $v$ are now nonzero. Since $\hat g_{\mu\nu} v^\mu v^\nu=1$, the only non-trivial components are $u\.\hat \nabla_a^\perp v$ and the only possible term (up to total derivatives) is
\be
	(u\.\hat \nabla_a^\perp v)(u\.\hat \nabla^{\perp a} v)=(\hat \nabla_a^\perp v\.\hat \nabla^{\perp a} v).
 \ee
 
 Overall we find the following 2-derivative contribution to the identity operator part of the effective action for codimension $q=2$ defects with $p>1$,
 \be
 S^{(2)}_\text{eff}=\int dz\sqrt{\hga}\Big(&a_{2,1} (v\.\hp_{ab})(v\.\hp^{ab})+a_{2,2}(u\.\hp_{ab})(v\.\hp^{ab})+ia_{2,3}(u\.\hp_{ab})(v\.\hp^{a}{}_c)\hat \e^{bc}\nn\\
 &+a_{2,4}(\hat \nabla_a^\perp v\.\hat \nabla^{\perp a} v)+ a_{2,5} \hat R + a_{2,6}\hat{\check{R}}+a_{2,7}	C_{vuvu}+ia_{2,8}\tl C_{vuvu}\Big).
 \ee
 The terms involving $\tl C_{vuvu}$ and $(u\.\hp_{ab})(v\.\hp^{a}{}_c)\hat \e^{bc}$ are to be omitted for $d> 4$. Note that when $p=2$, the intrinsic Ricci scalar $\hat{\check R}$ is (locally) a total derivative and the corresponding term captures only the topological data about the defect.

\section{Ansatz for the Weyl anomaly of a surface defect in $d=4$}
\label{app:ansatz}
The initial ansatz for the defect contribution to the Weyl anomaly of a surface defect in a $d=4$ bulk, as used in section~\ref{sec:anomalysurface} is given by:
\be
	\cA^{(1)}_{\cD}(g,\w)=&\int d^2z \sqrt{\gamma}\sum_{i=1}^{22} f_i F_i(\w),\label{eq:WeylAnsatz}
\ee
where:
\begin{multicols}{2}
	\begin{itemize}
		\item $F_1(\w)=\II_\chi^2\w$,
		\item $F_2(\w)=\II_{\bar{\chi}}^2\w$,
		\item $F_3(\w)=\II_\chi\II_{\bar{\chi}}\w$,
		\item $F_4(\w)=\mathring\II_{\chi\z\z}\mathring\II_{\chi\bar\zeta\bar\zeta}\w$,
		\item $F_5(\w)=\mathring\II_{\chi\z\z}\mathring\II_{\bar\chi\bar\zeta\bar\zeta}\w$,
		\item $F_6(\w)=\mathring\II_{\bar\chi\z\z}\mathring\II_{\chi\bar\zeta\bar\zeta}\w$,
		\item $F_7(\w)=\mathring\II_{\bar\chi\z\z}\mathring\II_{\bar\chi\bar\zeta\bar\zeta}\w$,
		\item $F_8(\w)=\check{R}\w$,
		\item $F_9(\w)=R_{\chi\chi}\w$,
		\item $F_{10}(\w)=R_{\chi\chi}\w$,
		\item $F_{11}(\w)=R_{\bar\chi\bar\chi}\w$,
		\item $F_{12}(\w)=C_{\chi\bar\chi\chi\bar\chi}\w$,
		\item $F_{13}(\w)=C_{\chi\bar\chi\zeta\bar\zeta}\w$,
		\item $F_{14}(\w)=\chi_\mu \bar{g}^{ab}\nabla_a\nabla^\perp_b \chi^{\mu}\w$,
		\item $F_{15}(\w)=\bar\chi_\mu \bar{g}^{ab}\nabla_a\nabla^\perp_b \chi^{\mu}\w$,
		\item $F_{16}(\w)=\II_\chi \ptl_\chi \w$,
		\item $F_{17}(\w)=\II_\chi \ptl_{\bar\chi} \w$,
		\item $F_{18}(\w)=\II_{\bar\chi} \ptl_\chi \w$,
		\item $F_{19}(\w)=\II_{\bar\chi} \ptl_{\bar\chi} \w$,
		\item $F_{20}(\w)=\chi^{\mu}\chi^{\nu}\nabla_{\mu}\partial_{\nu}\w$,
		\item $F_{21}(\w)=\chi^{\mu}\bar{\chi}^{\nu}\nabla_{\mu}\partial_{\nu}\w$,
		\item $F_{22}(\w)=\bar{\chi}^{\mu}\bar{\chi}^{\nu}\nabla_{\mu}\partial_{\nu}\w$.
	\end{itemize}
\end{multicols}

\newpage
\bibliographystyle{JHEP}
\bibliography{refs}

\end{document}